\documentclass[50pt]{article}
\usepackage{mathtools}
\usepackage{amsmath}
\usepackage[latin1]{inputenc}
\usepackage{amsmath,amsthm,amssymb}
\usepackage{graphicx}
\usepackage{color}

\renewcommand{\theequation}{\thesection.\arabic{equation}}
\newtheorem{df}{Definition}[section]
\newtheorem{lm}{Lemma}[section]
\newtheorem{thm}{Theorem} [section]
\newtheorem{pro}{Proposition} [section]

\newtheorem{rem}{Remark}[section]
\title{Cauchy problem for a generalized cross-coupled  Camassa-Holm system with  waltzing peakons and higher-order nonlinearities
}
\author{Shouming Zhou$^a$\footnote{E-mail:
zhoushouming76@163.com },\quad   Zhijun Qiao$^b$\footnote{E-mail:
zhijun.qiao@utrgv.edu} and Chunlai Mu$^c$\footnote{E-mail: clmu2005@163.com}\\
\small a. College of Mathematics Science, Chongqing Normal University, Chongqing 401331, China.\\
 \small b. Department of Mathematics, The University of Texas Rio Grande Valley, Edinburg, TX 78541 U.S.A.\\
  \small c. College of Mathematics and statistics, Chongqing University, Chongqing 401331, China.}

\date{}

\begin{document}
\maketitle
\renewcommand{\theequation}{\arabic{section}.\arabic{equation}}
\catcode`@=11 \@addtoreset{equation}{section} \catcode`@=12
\textbf{Abstract.}
In this paper, we study the Cauchy problem for a generalized cross-coupled  Camassa-Holm system with peakons and higher-order nonlinearities.
By the transport equation theory and the classical Friedrichs regularization method, we obtain the local well-posedness of solutions for the system
in  nonhomogeneous Besov spaces $B^s_{p,r}\times B^s_{p,r}$ with $1\leq p,r \leq +\infty$ and $s>\max\{2+\frac{1}{p},\frac{5}{2}\}$.
Moreover, we construct the local well-posedness in the critical Besov space $B^{5/2}_{2,1}\times B^{5/2}_{2,1}$ and the blow-up criteria. In the paper, we also
consider the well-posedness problem in the sense of Hadamard, non-uniform dependence, and H\"{o}lder   continuity  of the data-to-solution map  for the system on both the periodic and the non-periodic case. In light of a Galerkin-type approximation scheme, the system is shown well-posed in the Sobolev spaces $H^s\times H^s,s>5/2$  in the sense of Hadamard, that is, the data-to-solution map is continuous.  However, the solution map is not uniformly continuous.   Furthermore, we prove the H\"{o}lder continuity in the $H^r\times H^r$ topology when $0\leq r< s$ with  H\"{o}lder exponent   $\alpha$ depending on both $s$ and $r$.

\section{Introduction }
 In this paper, we propose the following Cauchy problem
\begin{equation}\label{Eq.(1.1)}
\left\{
\begin{array}{llll}
m_{t}+v^{p}m_{x}+av^{p-1}v_{x}m=0, &t\in\mathbb{R},x\in\mathbb{R},\\
n_{t}+u^{q}n_{x}+bu^{q-1}u_{x}n=0, &t\in\mathbb{R},x\in\mathbb{R},\\
u(x,0)=u_0(x), v(x,0)=v_0(x),  & x\in\mathbb{R},
\end{array}
\right.
\end{equation}
where $m= u-\alpha^2 u_{xx}, \ n= v-\beta^2 v_{xx}$ ($\alpha\geq0, \beta\geq0$), 
 $p,q\in\mathbb{Z}^+$, and $a,b$ are two constant parameters.
Obviously, Eq.(1.1) has nonlinearities of degree $\max\{p+1,q+1\}$. If choosing $m= u,n= v$, then Eq.(1.1) is a cross-coupled Burgers type system; if ordering $m= u_{xx},n=  v_{xx}$, then Eq. (1.1) becomes a cross-coupled  Hunter-Saxton type system; and if selecting $\alpha=\beta=1$, namely, $m= u-u_{xx},n= v-v_{xx}$, then Eq.(1.1) reads as a cross-coupled  Camassa-Holm type system. In this paper, we focus on the cross-coupled  Camassa-Holm type system. Other two limited cases can be dealt with in a similar way.

The water wave and fluid dynamics have been attracting much attention during last few decades due to various mathematical problems and nonlinear physics phenomena interfered \cite{Whitham}.
Since the raw water wave governing equations are 
nearly intractable, the request for suitably simplified model equations was initiated at the early stage of hydrodynamics development. Until the early twentieth century, the study of water waves was confined almost exclusively to the linear theory. Due to the linearization approach losing some important properties, people usually propose some nonlinear models to explain practical behaviors liking breaking waves  and solitary waves \cite{ConL}.
The most prominent example is  the following family of nonlinear dispersive partial differential equations
\begin{equation}
u_t-\gamma u_{xxx}-\alpha^2u_{xxt} = (c_1u^2+c_2u_x^2+c_3uu_{xx})_x,
\end{equation}
where $\gamma,\alpha,c_1,c_2$, and $c_3$ are real constants.   By the
Painlev\'{e} analysis  method in \cite{DP,DHH2,I}, there are  only
three asymptotically integrable members in this family: the KdV
equation  (Eq.(1.2) with $ \alpha=c_1=c_2=0$), the Camassa-Holm equation (Eq.(1.1) with $u=v,p=q =1,a=b=2$),  and the Degasperis-Procesi equation (Eq.(1.1) with $u=v,p=q =1,a=b=3$).
Apparently, 
nonlinearity in Eq.(1.2) is quadratic, and thus a natural question rises 
whether there are  integrable CH-type equations with higher-order nonlinearity.
Recently, such integrable peakon equations with cubic nonlinearity have been developed: one is the FORQ equation, and the other one is the Novikov equation (Eq. (1.1) with $u=v,p=q =2,a=b=4$).

Integrable equations have widely been studied because they usually have very good properties including  infinitely many conservation laws, infinite higher-order symmetries, bi-Hamiltonian
structure, and Lax pair, which make them solved by the inverse scattering method. Conserved quantities are very feasible for proving the existence of global solution in time, while a bi-Hamiltonian
formulation helps in finding conserved quantities effectively.
Discovering a new integrable equation may 
be accomplished via different methods. One of ways is the approach proposed by Fokas and Fuchssteiner \cite{FF} where the KdV equation, the CH equation, and the Hunter-Saxton equation are derived in a unified way.
The approach is based on the following fact: If $\theta_1, \ \theta_2$ are two Hamiltonian operators and for arbitrary number $k$ their combination $\theta_1+k\theta_2$ is also Hamiltonian, then
\begin{equation}
q_t=-(\theta_2\theta_1^{-1})q_x.
\end{equation}
is an integrable equation. Now, letting $\theta_1=\partial_x$ and $\theta_2=\partial_x+\gamma\partial_x^3+\frac{\alpha}{3}(q\partial_x+\partial_xq)$, where
$\alpha$ and $\gamma$ are constants, then Eq.(1.3) reads as the celebrated KdV equation (see \cite{KdV}). If choosing  $\theta_1=\partial_x+\nu\partial_x^3$ and $\theta_2 $
as above  with $q=u+\nu u_{xx}$, then Eq.(1.3) yields the following equation
\begin{equation}
u_t+u_x+\nu u_{xxt}+\nu u_{xxx}+\alpha uu_x +\frac{\alpha\nu}{3}(uu_{xxx}+2u_xu_{xx})=0,
\end{equation}
which could be reduced to the well-known CH equation through selecting the parameters appropriately and making
a change of variables with some scaling (see \cite{CH,CE2,CL}).
If setting  $\theta_1= \nu\partial_x^3$ and  $\theta_2=\partial_x + (q\partial_x+\partial_xq)$ with $q=  \nu u_{xx}$, then Eq.(1.3) leads to the Hunter-Saxton equation (see \cite{ACHFM,HS}):
\begin{equation}
  u_{xxt}+\frac{1}{\nu} u_{x}+ uu_{xxx}+2u_xu_{xx} =0.
\end{equation}
Actually, in light of the Fokas-Fuchssteiner framework \cite{FF}, one may generate generalized KdV-type and CH-type equations possessing bi-Hamiltonian structure and infinitely many conserved quantities.
For instance, letting $\theta_1=\partial_x$ and $\theta_2=\beta\partial_x+\gamma\partial_x^3+\frac{\alpha}{k+2}(q^k\partial_x+\partial_xq^k)$, where
$\alpha,\gamma$ are constants and $k\in\mathbb{Z}^+$ in Eq.(1.3)
produces the following generalized Korteweg-de Vries equation (see \cite{CKSTT,KPV}):
\begin{equation}
q_t+\beta\partial_xq+\gamma\partial_x^3q+\alpha q^kq_x=0.
\end{equation}
And choosing  $\theta_1=\partial_x+\nu\partial_x^3$ and $\theta_2=\beta\partial_x+\gamma\partial_x^3+ \alpha [(b-1)q\partial_x+\partial_xq]$
 with $q=u+\nu u_{xx}$ in 
 Eq.(1.3) generates the following CH-$b$ family equation \cite{HS2,LS}:
\begin{equation}
u_t+\beta u_x+\nu u_{xxt}+\gamma u_{xxx}+\alpha (b+1)uu_x + \alpha\nu (uu_{xxx}+ bu_xu_{xx})=0.
\end{equation}
Taking $b=3,\beta=\gamma=0,\nu=-1$ in Eq. (1.7) yields the remarkable DP equation (see \cite{DP,DHH}).

Furthermore, let us 
take  $\theta_1=\partial_x(1+\nu \partial_x^2)^k$ and $\theta_2=\beta k\partial_xq^{k-1}+\gamma k\partial_x^3q^{k-1}+  \alpha [(b-1)q\partial_x q^{k-1}+\partial_xq^k]$
 with $q=u+\nu u_{xx}$, then Eq.(1.3) yields the following generalized $b$-equation with nonlinearity of degree $k+1$ \cite{GH,ZhouMu}:
\begin{equation}
u_t +\beta\partial_x u^k+\nu u_{xxt}+\gamma\partial_x^3u^k +\alpha (b+1)u^ku_x + \alpha\nu (u^ku_{xxx}+ bu^{k-1}u_xu_{xx})=0.
\end{equation}
Taking $ k=2 $ in Eq. (1.8) gives the Novikov equation through choosing the parameters appropriately and making
a change of variables with some scaling \cite{HW2,N1}. If choosing $\theta_1=\partial_x- \partial_x^3$ and $\theta_2=q^2\partial_x +q_x\partial_x^{-1}q\partial_x $
 with $q=u-u_{xx}$, then Eq.(1.3) reads as
 \begin{equation}
q_t +2q^2u_x+q_x(u^2-u_x^2)=0,
\end{equation}
which is actually the FORQ equation \cite{Fokas,Fuchssteiner,Olver-Rosenau,Qiao1}.

The other attractive feature of the CH type equation (1.8) with $\beta=\gamma=0$ and $\nu=-1$ is: it admits the following peakon solutions 
\cite{GH,ZhouMu}:
\begin{equation*}
\begin{split}
\text{on the line: }  u_c(x,t)=c^{1/k}e^{-|x-ct|};
\text{ on the circle: }  u_c(x,t)= c^{1/k}  \cosh ([x-ct]_{\pi}-\pi),
\end{split}
\end{equation*}
with
\begin{equation}
  [x-ct]_{\pi}\doteq x-ct-2\pi\left[\frac{x-ct}{2\pi}\right].
\end{equation}
Equations (1.8) also has 
the multi-peakon solutions 
in the form of 
\cite{GH,ZhouMu}:
\begin{equation*}
\text{on the line: } u(t,x)=\sum_{i=1}^N p_i(t)e^{-|x-q_i(t)|}; \text{ on the circle: }   u(t,x)=\sum_{i=1}^N p_i(t)\cosh ([x-q_i(t)]_{\pi}-\pi),
\end{equation*}
where the peak positions $q_i(t)$ and amplitudes $p_i(t)$ satisfy 
\begin{equation*}
\begin{split}
p_j'& =\left(\sum_{i=1}^N p_i e^{-|q_j-q_i(t)|}\right)^{k }, \\
q_j'&= (b-k) p_j\left(\sum_{i=1}^N p_i e^{-|q_j-q_i |}\right)^{k-1 } \left(\sum_{i=1}^N p_i sgn(q_j-q_i)e^{-|q_j-q_i |}\right).
\end{split}
\end{equation*}

In recent years, the Camassa-Holm equation has been generalized to integrable multi-component CH models. One of them is the following form
\begin{equation}
\left\{
\begin{array}{llll}
m_t=um_x+k_1u_xm+k_2\rho \rho_x &t>0,x\in\mathbb{R},\\
\rho_t=k_3(u\rho)_x,&t>0,x\in\mathbb{R},\\
u(x,0)=u_0(x),\rho(x,0)=\rho_0(x),&x\in\mathbb{R}.
\end{array}
\right.
\end{equation}
Apparently, the two-component Camassa-Holm system (2CH)
and the two-component Degasperis-Procesi system (2DP) are included in Eq. (1.11) as two special cases with $k_1=2,k_2=\sigma=\pm1$ and with $k_1=3$, respectively.
Constantin and Ivanov \cite{CI} derived the 2CH 
 in the context of shallow water theory. The variable $u(x,
t)$ describes the horizontal velocity of the fluid and the variable
$\rho(x, t)$ is in connection with the horizontal deviation of the
surface from equilibrium, and all are measured in dimensionless units
\cite{CI}.    
The 2DP 
was shown to have solitons,  kink, and antikink solutions \cite{ZTF}.  Escher, Kohlmann and Lenells studied the geometric properties of the 2DP and local well-posedness in various function spaces \cite{EKL}. However, peakon and superposition of multi-peakons were not investigated yet. 


Motivated by the work of  Cotter, Holm, Ivanov and Percival \cite{CHIP}, in this paper we propose Eq.(1.1) to study its Cauchy problem and peakon dynamical system. Eq.(1.1) with the choice of $a=b=2$ and $p=q=1$ (called the CCCH equation)
can be derived from a variational principle via an
Euler-Lagrange system with the following Lagrangian \cite{CHIP}
\begin{equation*}
L(u,v)=\int_\mathbb{R}(uv+u_xv_x)dx.
\end{equation*}
Alternatively, it can be formulated as a two-component system of Euler-Poincar\'{e}  equations in one dimension on $\mathbb{R}$ as follows,
 \begin{equation*}
 \begin{split}
\partial_t m=-ad^*_{\delta h/\delta m}m=-(vm)_x-mv_x \quad\text{with} \quad v\doteq \frac{\delta h}{\delta m}=K*n,\\
\partial_t n=-ad^*_{\delta h/\delta n}n=-(un)_x-nu_x \quad\text{with} \quad u\doteq \frac{\delta h}{\delta n}=K*m,\\
\end{split}
\end{equation*}
with $K(x,y)=\frac{1}{2}e^{-|x-y|}$ being the Green function of the Helmholtz operator,
and $h$ being the Hamiltonian
 \begin{equation*}
 h(n,m)=\int_\mathbb{R}nK*mdx=\int_\mathbb{R}mK*ndx.
\end{equation*}
This Hamiltonian system has a two-component singular momentum map \cite{CHIP,CHIP2}
 \begin{equation*}
m(x,t)=\sum_{a=1}^M m_a(t)\delta(x-q_a(t)),n(x,t)=\sum_{a=1}^N n_b(t)\delta(x-r_b(t)).
\end{equation*}
 Such a formal multi-peakon solution, waltzing peakons and compactons of the CCCH 
 are given in \cite{CHIP}. A geometrical interpretation for the CCCH system along with a large class of peakon equations was discussed in \cite{EIK}.
More recently, the Cauchy problem of Eq.(1.1) has been
studied extensively. Local well-posedness problem of the CCCH system in  Sobolev spaces  and  Besov spaces is investigated \cite{Liu,Zhou}. The  global existence, blow-up  phenomenon, and persistence properties were discussed in \cite{HHI,Liu,Zhou}.

Based on the argument on the approximate solutions in \cite{D3}  and  \cite{ZhouSM,ZMW},
we use the  classical Friedrichs regularization method and the transport equations theory to study 
the local well-posedness problem of Eq.(1.1). However, the structure of Eq.(1.1) with high-order nonlinearities is complicated in comparison with the one 
appearing in \cite{D3,ZhouSM,ZMW}.
Thus, there is difficulty with this problem. 
In this paper, we plan to use the original Eq.(1.1) rather than the nonlocal form (see (3.18) below).
This way may simplify our computation based on  the fact: $\|u\|_{B^{s}_{l,r}} \approxeq \|m\|_{B^{s-2}_{l,r}}$.
In order to show the local well-posedness through the
Littlewood-Paley decomposition and nonhomogeneous Besov spaces, we need to prove the following inequality
$$\|u_{k}(t) \|_{B^{s}_{l,r}}+\|v_{k}(t) \|_{B^{s}_{l,r}}\leq\frac{\| u_{0}\|_{B^{s}_{l,r}}+\| v_{0}\|_{B^{s}_{l,r}}}{ \left(1-2 \kappa C t(\| u_{0}\|_{B^{s}_{l,r}} +\| v_{0}\|_{B^{s}_{l,r}})^\kappa\right)^{1/\kappa} }  \text{ with }  \kappa= \max\left\{p,q\right\},$$
which yields the following theorem.
It should be pointed out that the degree of nonlinearity in Eq.(1.1) is $\max\{p+1,q+1\}$ rather than quadratic or cubic. 

\begin{thm}\label{result1}
Suppose that $1\leq l,r \leq +\infty$ and $s>\max\{2+\frac{1}{l },\frac{5}{2},3-\frac{1}{l}\}$ but $s\neq 3+\frac{1}{l}$. Let $(u_0,v_0)\in B_{l,r}^{s}\times B_{l,r}^{s}$.  Then there exists a time $T>0$ such that the Cauchy problem (1.1) has a unique solution $(u,v)\in B_{l,r}^{s}(T)\times B_{l,r}^{s}(T)$, and map $(u_0,v_0)\mapsto (u,v)$ is continuous from a neighborhood of $(u_0,v_0)$ in $B_{p,r}^{s}\times B_{p,r}^{s}$ into
$$\mathcal{C}([0,T];B_{l,r}^{s'}) \cap \mathcal{C}^1([0,T];B_{l,r}^{s'-1})\times \mathcal{C}([0,T];B_{l,r}^{s'}) \cap \mathcal{C}^1([0,T];B_{l,r}^{s'-1})$$
for every $s'<s$ when $r=+\infty$ and $s'=s$ whereas $r<+\infty$.
\end{thm}

\begin{rem}
When $l=r=2$, the Besov space $B^s_{l,r}$
coincides with the Sobolev space $H^s$. Thus under the condition
$(u_0,v_0)\in H^s\times H^s$ with $s>\frac{5}{2}$, i.e.
$m_0,n_0\in
 H^s$ with $s>\frac{1}{2}$, the above theorem implies
 that there exists a time $T>0$ such that the initial-value problem (1.1) has a unique
 solution $m,n\in \mathcal{C}([0,T];H^s)\cap \mathcal{C}^1([0,T];H^{s-1}) $, and the map $(m_0,n_0) \mapsto (m,n)$ is continuous from a
 neighborhood of $(m_0,n_0)$ in $H^s\times H^s$ into $\mathcal{C}([0,T];H^s)\cap \mathcal{C}^1([0,T];H^{s-1})\times \mathcal{C}([0,T];H^s)\cap \mathcal{C}^1([0,T];H^{s-1})$.
\end{rem}

It is well known that $B_{2,2}^s(\mathbb{R})=H^s$ and for any $s' <5/2<s$:$$H^s\hookrightarrow B^{\frac{5}{2}}_{2,1}\hookrightarrow
H^{\frac{5}{2}}\hookrightarrow
B^{\frac{5}{2}}_{2,\infty}\hookrightarrow H^{s'}, $$ which shows that $H^s$ and $B_{2,2}^s$ are quite close. So,  we may establish the local well-posedness in the critical Besov space $B^{5/2}_{2,1}\times B^{5/2}_{2,1}$.

\begin{thm}\label{result2}
Assume that the initial data $z_0\doteq(u_0,v_0) \in
B^{\frac{5}{2}}_{2,1}\times B^{\frac{5}{2}}_{2,1}$. Then there exists a unique solution $z=(u,v)$
and a maximal time $T=T(z_0)>0$ to the Cauchy problem (1.1) such
that
  $$z=z(\cdot,z_0)\in \mathcal{C}([0,T];B^{\frac{5}{2}}_{2,1})\cap \mathcal{C}^1([0,T];B^{\frac{3}{2}}_{2,1})\times \mathcal{C}([0,T];B^{\frac{5}{2}}_{2,1})\cap \mathcal{C}^1([0,T];B^{\frac{3}{2}}_{2,1}).$$
Moreover, the solution continuously depends  on the initial data, i.e., the mapping
$$z_0\mapsto z(\cdot,z_0):B^{\frac{5}{2}}_{2,1}\times B^{\frac{5}{2}}_{2,1} \mapsto \mathcal{C}([0,T];B^{\frac{5}{2}}_{2,1})\cap \mathcal{C}^1([0,T];B^{\frac{3}{2}}_{2,1})\times\mathcal{ C}([0,T];B^{\frac{5}{2}}_{2,1})\cap \mathcal{C}^1([0,T];B^{\frac{3}{2}}_{2,1})$$ is continuous.
\end{thm}

Next, we show that the solutions of Eq.(1.1) can only have singularities wherever the wave breaking occurs.
\begin{thm}\label{result3}
Let $(m,n)\in H^s(\mathbb{R})$ with $s>\frac{1}{2}$, and
$(m,n)$ be the corresponding solution to the initial-value problem (1.1). Assume that
$T^*_{m_0,n_0}>0$ is the maximum time of existence for the solution of
the problem (1.1). If $T^*_{p_0,q_0}<\infty,$ then
\begin{equation*}
\int^{T^*_{m_0,n_0}}_0\left(\|  n\|_{L^\infty}^{p } +\| m\|_{L^\infty}\|n\|^{p -1}_{L^\infty} +\|  m\|_{L^\infty}^{q } +\| n\|_{L^\infty}\|m\|^{q-1}_{L^\infty}\right) d\tau<\infty,
\end{equation*}
\end{thm}
Our next result describes the precise blowup scenarios for sufficiently regular solutions to the problem (1.1), and the blowup occurs in the form of breaking
waves, namely, the solution remains bounded but its slope becomes unbounded in
finite time.  The result similar to the Camassa-Holm type equations can be found in \cite{ZMW}. 
\begin{thm}\label{result4}
  Let   $z_0=(u_0,v_0) \in  L^1 \cap H^s $ with $s>5/2$, and $T$ be the lifespan of the solution $z(x,t)=(u(x,t),v(x,t))$ to Eq.(1.1)  with the initial data $z_0 $.
If $p=2a,q=2b$, then every solution $z(x,t)$ to Eq. (\ref{Eq.(1.1)}) remains  globally regular in time.
If $p>2a$  { (or $q>2b$)}, then the corresponding solutions $z(x,t)$ blow up in a finite time if and only if  $(v^p)_x$ { (or $(u^q)_x$) }  is unbounded at $-\infty$ in a finite time.
If $p<2a$ { (or $q<2b$)}, then
the corresponding solution  $z(x,t)$ blows up in a finite time if and only if the slope of the function $(v^p)_x$  { (or $(u^q)_x$)}  tends to $+\infty$ 
in a finite time.
\end{thm}
Let us now give some sufficient conditions for the global existence   of solutions to Eq.(1.1).
\begin{thm}\label{result5}
Let  $u_0 \in  H^s \cap W^{2,\frac{p}{a}} $ and $v_0 \in  H^s \cap W^{2,\frac{q}{b}} $ with $s>5/2$ and $0\leq a\leq  p $, $0\leq b\leq q$. Then the solution to Eq.(1.1)  remains smooth for all time.
\end{thm}

As per \cite{CHIP}, 
the CCCH system might not be completely
integrable. However, it does have 
peakon and multi-peakon solutions which display 
interesting dynamics with both propagation and oscillation. Furthermore, if the two initial values $u_0$ or $v_0$ of 
in Eq. (1.1) have a compact support,
then the compact property
will be inherited by $u$ and  $v$ at all times $t\in[0,T)$.
\begin{thm}\label{result6}
{ Assume that $(u_0,v_0)\in H^s\times H^s$ with $s>5/2$, and $m_0=(1-\partial_x^2)u_0$ (or $n_0=(1-\partial_x^2)v_0$)  have a compact
support. Let $T=T(u_0)>0$ be the maximal existence time of the unique solutions $u(x,t)$ and $v(x,t)$
to Eq. (1.1) with initial data $u_0(x)$ and $ v_0(x)$. Then for any $t\in[0,T)$ the $\mathcal{C}^1$ function $x\mapsto m(x,t)$ (or $x\mapsto n(x,t)$)
also has a compact support.}
\end{thm}

Tracking back to 
the proof of the well-posedness and non-uniform dependence problem for the CH equation \cite{HK,HKM,HMP}, the DP equation \cite{HH1}, the $b$-equation \cite{Gray,HH3,HM}, the Novikov equation \cite{HH2}, the FORQ equation \cite{HM22}, and the two-component of Camassa-Holm equation \cite{Thompson}, we plan to use the Galerkin-type approximation method 
to deal with 
Eq.(1.1).
In comparison 
with the  CH equation, the Novikov equation,
 and the two-component CH system, we observe that the well-posedness property for these CH type equations holds for $s>3/2$ whereas the well-posedness problem for  Eq.(1.1)  holds
for $s>5/2$. This difference 
may be caused 
by the presence of the
extra high-order derivative  terms  $u^{q-1}u_xv_{xx}$ and $v^{p-1}v_xu_{xx}$
in the  nonlocal form (see Eq. (3.18) below).
These terms 
make  Eq.(1.1)  not being treated as a first-order ODE.
The novelty in the proof of well-posedness for Eq.(1.1) is the following fact:
  $||m||_{H^{s-2}}=||u||_{H^s}$ and $||n||_{H^{s-2}}=||v||_{H^s}$, which provide the delicate proof for several required nonlinear estimates we need.
  Let us give the well-posedness of the Cauchy problem for Eq.(1.1) in Sobolev spaces in the sense of Hadamard below.

\begin{thm}\label{result7}
Assume that  $z_0=(u_0,v_0)\in H^s \times H^s$ with $s>5/2$. Let $T^*$
be the maximal
existence time of the solution $z=(u,v)$ to equation (1.1) with the initial data $z_0$ for both periodic and non-periodic cases. Then $T^*$
satisfies
\begin{equation}
\begin{split}
T^*\geq T_0:= \frac{2^\kappa-1}{2^{\kappa+1}\kappa C_s||  z_0||_{H^{s  } }^{\kappa} },
\end{split}
\end{equation}
where $C_s$ is a constant depending on $s$ and $\kappa=\max\{p,q\}$. Moreover, we have
\begin{equation}
||u(t)||_{H^s}\leq 2 ||u_0||_{H^s}, \text { for } 0\leq t \leq T_0.
\end{equation}
\end{thm}

Furthermore, 
based on the method of approximate solutions and well-posedness estimates, we show that the data-to-solution map is continuous, but not uniformly continuous on any bounded subset of $H^s\times H^s$ with $s>5/2$. In particular, the novelty of our proof 
does not employ 
any conserved quantity, which is a big difference from the proof for the
CH equation \cite{HK,HKM} and two-component of Camassa-Holm system \cite{Thompson} relying on the conservation laws in the $H^1$ norm
and also from the proof for the DP  equation \cite{HH1} dependent on the conservation laws in a twisted $L^2$
norm.  Inspired by the proof of non-uniform dependence for the Novikov equation \cite{HH2} and the FORQ equation \cite{HM22}, Our
proof is based on the technique of approximate solutions containing terms of
both high and low frequencies. We show that the error between the solution of the Cauchy problem (1.1) and the solution of the Cauchy problem with the approximate initial data composed of low and high frequencies
may be negligible. Due to the different degrees of nonlinearities ($p\neq q$) in Eq.(1.1), the low and high frequency parts of the approximate solutions $(u^{\omega,\lambda}, v^{\omega,\lambda} )$ (see (5.1) - (5.3) below), are more complicated than the usual ones taken in \cite{HK,HKM,Thompson,HH1,HH2,HM22}, but our procedure in this paper may simplify computations.

\begin{thm}\label{result8}
 If $ s > 5/2$, then the data-to-solution map $z_0=(u_0,v_0)\rightarrow z(t)=(u(t),v(t))$ for both the periodic and the non-periodic cases of Eq.(1.1)  is not uniformly continuous from any bounded subset
of $H^s\times H^s$ into $\mathcal{C} ([0, T ]; H^s
 )\times \mathcal{C} ([0, T ]; H^s
 )$.
\end{thm}

In the following, we show that Eq.(1.1) also admits peaked solitary wave and multi-peaked solitray wave solutions.

\begin{thm}\label{result9}
For $c>0$ and $p,q\geq1$, the single peaked solitary wave   takes on the form
\begin{align}
&\text{on the line: } u(t,x)=\alpha e^{-|x-ct-x_{0}|},v(t,x)=\beta e^{-|x-ct-x_{0}|};\\
&\text{on the circle: } u(t,x)=  \frac{\alpha}{\cosh(\pi)}\cosh ([x-ct]_{\pi}-\pi),v(t,x)=\frac{\beta}{\cosh(\pi)} \cosh ([x-ct]_{\pi}-\pi),
\end{align}
which are global weak solutions { if and only if} $\alpha=c^{1/q},\beta=c^{1/p}$. Moreover,
Eq.(1.1) has multi-peaked solitary wave solutions in the form of
\begin{align}
&\text{in the non-periodic case: }  u(t,x)=\sum_{i=1}^M f_i(t)e^{-|x-g_i(t)|}, v(t,x)=\sum_{j=1}^N h_j(t)e^{-|x-k_j(t)|};\\
&\text{in the  periodic case: } u(t,x)=\sum_{i=1}^M f_i(t)\cosh([x-g_i(t)]_{\pi}-\pi),  v(t,x)=\sum_{j=1}^N h_j(t)\cosh([x-k_i(t)]_{\pi}-\pi),
\end{align}
whose peaked positions $g_i(t),h_j(t)$ and amplitudes $f_i(t),k_j(t)$ satisfy the following dynamical system
\begin{equation}
\begin{split}
\dot{g}_i&=v^p(g_i),\dot{f}_i=(p-a)  v^{p-1}(g_i)\langle v_x(g_i) \rangle f_i,\\
\dot{k}_j&=u^q(k_j),\dot{h}_j=(q-b) u^{q-1}(k_j)\langle u_x(k_j)\rangle h_j,
\end{split}
\end{equation}
with $\langle f(x) \rangle=\frac{1}{2}(f(x^-)+f(x^+))$ and $ [x-ct]_{\pi} $ defined by Eq. (1.10).
\end{thm}
 In Theorem 1.9, one may readily see that the peaked solitary wave solutions for Eq.(1.1) are similar to the peakon solutions for the CH equation.  The authors in \cite{GH,ZhouSM} already proved that  the data-to-solution map for the regular form of peakon solutions  is not uniformly continuous in Sobolev spaces with $s<3/2$. This result supplements Theorem 1.8, which shows that the data-to-solution map for
(1.1)  is not uniformly continuous on bounded subsets of $H^s$ with $s>5/2$.

Theorem 1.8 shows that 
the data-solution map depends on the initial data continuously, but sharp as established in Theorem 1.7. In other words, Eq.(1.1) is well-posed in Sobolev spaces $H^s$
both on the line and the circle for $s>5/2$, and its
data-to-solution map is continuous but not uniformly continuous.
Our next result will provide information about stability of the  data-solution
map.
Due to the presence of the high-order derivative terms  $u^{q-1}u_xv_{xx}$ and $v^{p-1}v_xu_{xx}$, we need to  transform Eq.(1.1) into a system 
through differentiating  Eq. (1.1)  
with respect to $x$ (see equation (7.2) for details).
Then transforming back to the original system leads 
the data-solution map for  Eq.(1.1)  to be H\"{o}lder continuous in
$H^r$-topology for all $0\leq r<s$. Let us describe this procedure below.

\begin{thm}\label{result10}
 Let $s>5/2$ and $0\leq r<s$. Then the data-to-solution map for Eq.(1.1) on the line and the circle is H\"{o}lder continuous in $H^s\times H^s$ under the   $H^r\times H^r$ norm. In particular,  for initial data  $ z_0=(u_{0,1},v_{0,1})$ and $w_0=(u_{0,2},v_{0,2} )$ in the ball $B(0,\rho)=\{\psi\in H^s\times H^s: ||\psi||_{H^s}\leq \rho\}$ of $H^s\times H^s$, the corresponding solutions $ z=(u_1,v_1)$ and $ w=(u_2,v_2)$  to Eq.(1.1) satisfy the following inequality
\begin{equation*}
\begin{split}
||z(t)-w(t)||_{H^{r}} \leq C_{r,s,p,q,\rho} ||z_0-w_0||_{H^{r}}^\alpha ,
\end{split}
\end{equation*}
where the parameter $\alpha$ is given by
\begin{equation}
\alpha=\left\{
\begin{array}{llll}
1, &(s,r)\in A_1\doteq\{ 0\leq r\leq 3/2,3-r\leq r \leq s-2\}\cup\{ r>3/2,  r \leq s-1\},\\
\frac{2s-3}{s-r}, &(s,r)\in A_2\doteq\{(s,r): 5/2<s<3, 0\leq r\leq 3-s \},\\
\frac{s-r}{2}, &(s,r)\in A_3\doteq\{(s,r): s>5/2,s-2\leq r\leq 3/2\},\\
s-r, &(s,r)\in A_3\doteq\{(s,r): s>5/2,s-1\leq r<s\}.
\end{array}
\right.
\end{equation}
The lifespan $T$ and the constant $c$ only depend on $s,r,p,q$ and $\rho$.
\end{thm}

The entire paper is organized as follows. In next section,   we establish the local well-posedness in Besov spaces of Eq.(\ref{Eq.(1.1)}) through proving Theorems 1.1-1.2.
In section 3, we deal with the global existence and blow up phenomena of solutions to the problem (\ref{Eq.(1.1)}) and analyze the propagation behaviors of compactly supported solutions 
through proving Theorems 1.3-1.6. 
In section 4,  with the aid of the Galerkin-type approximation scheme we propose, we obtain the well-posedness of the Cauchy problem for Eq.(1.1) in Sobolev spaces in the sense of Hadamard with proving Theorem 1.7.
In section 5, in the light of the Galerkin-type approximation approach and well-posed estimates obtained in section 4,
we show that the data-to-solution map is continuous but not uniformly continuous on any bounded subset of $H^s\times H^s$ with $s>5/2$ through proving Theorem 1.8.
In section 6, the peakon and multi-peakons are derived through proving Theorem 1.9. In the last section, we demonstrate that
the solution map for Eq.(1.1) is H\"{o}lder continuous in
$H^r$-topology for all $0\leq r<s$ through proving Theorem 1.10.

\section{Local well-posedness in Besov spaces}
In this section, we shall discuss the local well-posedness of the Cauchy problem
 Eq.(1.1)  in the nonhomogeneous Besov spaces, i.e., prove Theorem 1.1 and 1.2 . The Littlewood-Paley theory and the properties
of the Besov-Sobolev spaces can refer to \cite{ZhouSM,ZMW} and references therein.

\subsection{Local well-posedness in Besov spaces $B_{l,r}^s$ }
First, we present the following definition.
 \begin{df}
For $T>0,s\in\mathbb{R}$ and $1\leq l\leq +\infty$ and $s\neq 2+\frac{1}{l}$, we define
$$E_{l,r}^s(T)\doteq \mathcal{C}([0,T];B_{l,r}^s)\cap \mathcal{C}^1([0,T];B_{l,r}^{s-1})  \quad \text{ if } r<+\infty,$$
$$E_{l,\infty}^s(T)\doteq L^\infty([0,T];B_{l,\infty}^s)\cap Lip([0,T];B_{l,\infty}^{s-1}),$$
and $E_{l,r}^s\doteq \cap_{T>0}E_{l,r}^s(T)$.
\end{df}

Next, we denote $C>0$ a generic constant only depending on $l, r, s,p,q,|a|,|b|$.
Uniqueness and continuity with respect to the initial data are an immediate consequence of the following result.
\begin{lm}
 Let $1\leq l, r\leq +\infty$ and $s>\max\{2+\frac{1}{l},\frac{5}{2},3-\frac{1}{l}\}$. Suppose that  $(u_{i}, v_{i})\in \{L^\infty([0,T];B_{l,r}^s)\cap \mathcal{C}([0,T];\mathcal{S}')\}^2$ $(i=1,2)$
be two
given solutions of the initial-value problem (1.1) with the initial data $(u_{i}(0),v_{i}(0))\in B_{l,r}^s\times B_{l,r}^s$  ($i=1,2$), and denote $u_{12}=u_{1}-u_{2},v_{12}=v_{1}-v_{2}$, and thus $m_{12}=m_{1}-m_{2},n_{12}=n_{1}-n_{2}$. Then for every $t\in[0,T]$, we have

(i) if $s>\max\{1+\frac{1}{l},\frac{3}{2}\}$ and $s\neq 4+1/p$, then
\begin{equation}
\begin{split}
\|u_{12}\|&_{B^{s-1}_{l,r}}+\|v_{12}\|_{B^{s-1}_{l,r}} \leq \left(\| u_{12} (0)\|_{B^{s-1}_{l,r}} +\| u_{12} (0)\|_{B^{s-1}_{l,r}} \right)  \exp
\left(C\int^{T}_{0}\Gamma_{s}(t,\cdot)d\tau\right) .
\end{split}
\end{equation}
where
\begin{equation*}
\Gamma_{s}(t,\cdot)=\left( \|  v_{1}  \|_{B^{s}_{l,r}}^{p} + \|  u_{1}  \|_{B^{s}_{l,r}}^{q}+\|u_{2}\|_{B^{s}_{l,r}}\sum^{p-1}_{i=0}\| v_{1}\|_{B^{s}_{l,r}}^{p-1-i}\|v_{2}\|_{B^{s}_{l,r}}^{ j} +\|v_{2}\|_{B^{s}_{l,r}}\sum^{q-1}_{j=0}\| u_{1}\|_{B^{s}_{l,r}}^{q-1-i}\|u_{2} \|_{B^{s}_{l,r}}^{ j} \right).
\end{equation*}
(ii) if   $s= 4+1/p$, then
\begin{equation*}
\begin{split}
\|u_{12}\|&_{B^{s-1}_{l,r}}+\|v_{12}\|_{B^{s-1}_{l,r}} \leq C\left(\| u_{12} (0)\|_{B^{s-1}_{l,r}} +\| u_{12} (0)\|_{B^{s-1}_{l,r}} \right) ^\theta \Gamma_{s}^{1-\theta}(t,\cdot)   \exp
\left(C\theta\int^{T}_{0}\Gamma_{s}(t,\cdot)d\tau\right),
\end{split}
\end{equation*}
where $\theta\in(0,1)$(i.e., $\theta=\frac{1}{2}(1-\frac{1}{2l})$) and $\Gamma_{s}(t,\cdot)$ as in case (i).
\end{lm}

\begin{proof}   It is obvious that
$u_{12},v_{12}\in L^\infty([0,T];B_{l,r}^s)\cap \mathcal{C}([0,T];\mathcal{S}'),$
which implies that $u_{12},v_{12}\in \mathcal{C}([0,T];B_{l,r}^{s-1})$, and $(u_{12},v_{12})$ and $m_{12},n_{12}$ solves the transport equations
\begin{equation*}
\left\{
\begin{array}{llll}
&\partial_{t}m_{12}+ v_{1} ^{p}\partial_{x} m_{12}=-[ v_{1} ^{p}- v_{2} ^{p}]\partial_{x}m_{2} -\frac{a}{p}\partial_x v_{1} ^{p} m_{12}  -\frac{a}{p}[\partial_x v_{1} ^{p}-\partial_x v_{2} ^{p}] m_{2} ,\\
&\partial_{t}n_{12}+u_{1}^{q}\partial_{x} n_{12} =-[u_{1}^{q}-u_{2}^{q}]\partial_{x}n_{2} -\frac{b}{q}\partial_xu_{1}^{q} n_{12}  -\frac{b}{q}[\partial_xu_{1}^{q}-\partial_xu_{2}^{q}] n_{2} ,\\
&m_{12}|_{t=0}=m_{12}(0)\doteq m_{1}(0)-m_{2}(0),n_{12}|_{t=0}=n_{12}(0)\doteq n_{1}(0)-n_{2}(0).
\end{array}
\right.
\end{equation*}
According to Lemma 2.2(i) in \cite{ZMW}, we have
\begin{equation}
\begin{split}
 \|m_{12}\|_{B^{s-3}_{l,r}} \leq&  \|m_{12}(0)\|_{B^{s-3}_{l,r}} +C\int_0^t\left(\|\partial_{x} v_{1} ^{p}\|_{B^{s-4}_{l,r}} +\|\partial_{x} v_{1} ^{p}\|_{B^{\frac{1}{l}}_{p,r}\cap L^\infty}\right) \|m_{12}\|_{B^{s-3}_{l,r}}d\tau\\
 &+C\int_0^t \|[ v_{1} ^{p}- v_{2} ^{p}]\partial_xm_{2}-\frac{a}{p}\partial_x v_{1} ^{p} m_{12}  -\frac{a}{p}[\partial_x v_{1} ^{p}-\partial_x v_{2} ^{p}] m_{2} \|_{B^{s-3}_{l,r}}d\tau.
\end{split}
\end{equation}
Since $s>\max\{2+\frac{1}{l},\frac{5}{2},3-\frac{1}{l}\}\geq 2+\frac{1}{l}$, we obtain
\begin{align*}
 \|\partial_{x} v_{1} ^{p}\|_{B^{s-4}_{l,r}} +\|\partial_{x} v_{1} ^{p}\|_{B^{\frac{1}{l}}_{l,r}\cap L^\infty} \leq 2\|\partial_{x} v_{1} ^{p}\|_{B^{s-2}_{l,r}}\leq C\|v_{1}\|_{B^{s}_{l,r}}^{p}.
\end{align*}
where the property $(1-\partial_x^2)\in OP (S^2)$ is applied.
Thus, by Proposition 2.2 (7) in \cite{ZMW}, we find that for all $s\in \mathbb{R}$ 
$$\|u_{i}\|_{B^{s}_{l,r}} \approxeq \|m_{i}\|_{B^{s-2}_{l,r}}\text{ and }\|v_{i}\|_{B^{s}_{l,r}} \approxeq \|n_{i}\|_{B^{s-2}_{l,r}}.$$
If $\max\{2+\frac{1}{l},\frac{5}{2}\}<s\leq 3+\frac{1}{l}$, by Proposition 2.5 (2) in \cite{ZMW} and $B^{s-2}_{l,r}$
being an algebra, we arrive at
\begin{equation}
\begin{split}
 \|[ v_{1} ^{p}&- v_{2} ^{p}]\partial_xm_{2}-\frac{a}{p}\partial_x v_{1} ^{p} m_{12}  -\frac{a}{p}[\partial_x v_{1} ^{p}-\partial_x v_{2} ^{p}] m_{2} \|_{B^{s-3}_{l,r}}   +\|\partial_x v_{1} ^{p}-\partial_x v_{2} ^{p}\|_{B^{s-3}_{l,r}} \| m_{2} \|_{B^{s-2}_{l,r}}  \\
&\leq C\left( \|  v_{1}  \|_{B^{s-1}_{l,r}}^{p} \|u_{12} \|_{B^{s-1}_{l,r}} +\|u_{2}\|_{B^{s}_{l,r}}\|v_{12}  \|_{B^{s-1}_{l,r}}\sum^{p-1}_{i=0}\| v_{1}\|_{B^{s}_{l,r}}^{p-1-i}\|v_{2} \|_{B^{s}_{l,r}}^{ i} \right).
\end{split}
\end{equation}
For $s>3+\frac{1}{l}$, the inequality (2.3) also holds true in view of the fact that $B^{s-3}_{l,r}$ is an algebra. Thus, we may derive
\begin{align*}
 \|u_{12}\|_{B^{s-1}_{l,r}} \leq&  \|u_{12}(0)\|_{B^{s-1}_{l,r}} +C\int_0^t  \left( \|  v_{1}  \|_{B^{s}_{l,r}}^{p} \|u_{12} \|_{B^{s-1}_{l,r}} +\|u_{2}\|_{B^{s}_{l,r}}\|v_{12}  \|_{B^{s-1}_{l,r}}\sum^{p-1}_{i=0}\| v_{1}\|_{B^{s}_{l,r}}^{p-1-i}\|v_{2} \|_{B^{s}_{l,r}}^{ i} \right)(\tau)d\tau.
\end{align*}
Similarly, we can also treat the inequality for the component $v$:
\begin{align*}
 \|v _{12}\|_{B^{s-1}_{l,r}} \leq&  \|v_{12}(0)\|_{B^{s-1}_{l,r}} +C\int_0^t  \left( \|  u_{1}  \|_{B^{s}_{l,r}}^{p} \|v_{12} \|_{B^{s-1}_{l,r}} +\|v_{2}\|_{B^{s}_{l,r}}\|u_{12}  \|_{B^{s-1}_{l,r}}\sum^{q-1}_{j=0}\| u_{1}\|_{B^{s}_{l,r}}^{q-1-i}\|u_{2} \|_{B^{s}_{l,r}}^{ j} \right)(\tau)d\tau.
\end{align*}
Therefore
\begin{align*}
  \|u _{12}\|_{B^{s-1}_{l,r}}  +\|v _{12}\|_{B^{s-1}_{l,r}} \leq &\|u_{12}(0)\|_{B^{s-1}_{l,r}}  + \|v_{12}(0)\|_{B^{s-1}_{l,r}}+C\int_0^t\left(  \|u_{12}  \|_{B^{s-1}_{l,r}} +\|v_{12}  \|_{B^{s-1}_{l,r}}\right)\Gamma_s(\tau,\cdot)d\tau.
\end{align*}
Applying Gronwall's lemma to the above inequality leads to (i).

For the critical case (ii) $s=4+1/p$, we here use the interpolation method to deal with it. Indeed, if we choose $\theta=\frac{1}{2}(1-\frac{1}{2l})\in (0,1)$, then $s-1=3+\frac{1}{l}=\theta (2+1/2l)+(1-\theta)(4+1/2l)$ . According to Proposition 2.2(5) in \cite{ZMW} and the above inequality, we have
\begin{align*}
  \|u _{12}\|_{B^{3+1/p}_{l,r}}  +\|v _{12}\|_{B^{3+1/p}_{l,r}} &\leq \|u _{12}\|_{B^{2+1/2l}_{l,r}}^{\theta} \|u _{12}\|_{B^{4+1/2l}_{l,r}}^{1-\theta} +\|v_{12}\|_{B^{2+1/2l}_{l,r}}^{\theta} \|v _{12}\|_{B^{4+1/2l}_{l,r}}^{1-\theta} \\
  &\leq \left(\|u _{12}\|_{B^{2+1/2l}_{l,r}}+\|v_{12}\|_{B^{2+1/2l}_{l,r}}\right)^{\theta} \left(\|u _{12}\|_{B^{4+1/2l}_{l,r}}^{1-\theta}+\|v _{12}\|_{B^{4+1/2l}_{l,r}}^{1-\theta}\right)\\
   &\leq C\left(\| u_{12} (0)\|_{B^{s-1}_{l,r}} +\| u_{12} (0)\|_{B^{s-1}_{l,r}} \right) ^\theta \Gamma_{s}^{1-\theta}(t,\cdot)   \exp
\left(C\theta\int^{T}_{0}\Gamma_{s}(t,\cdot)d\tau\right),
\end{align*}
which yields the desired result.
\end{proof}
Now, let us start the proof of Theorem 1.1, which is motivated by the proof of local existence theorem about Camassa-Holm type equations  in \cite{D22,ZMW}. We shall use the classical Friedrichs regularization method to construct the approximate solutions to  Eq.(1.1) .
\begin{lm}
Let $l, r$ and $s$ be as in the statement of Lemma 2.1. Assume that $u(0) =
v(0):=0$. There exists a sequence of smooth functions $(u_k,v_k)\in \mathcal{C}(\mathbb{R}^+; B_{p,r}^\infty)^2$ solving
$$(T_k) \quad\left\{
\begin{array}{llll}
\partial_{t}m_{k+1}+v^{p}_{k }\partial_xm_{k+1}+\frac{a}{p} \partial_xv^{p}_{k }m_{k }=0, \\
\partial_{t}n_{k+1}+u^{q}_{k }\partial_xn_{k+1}+\frac{b}{q} \partial_xu^{q}_{k }n_{k }=0,\\
m_{k+1}(0)=S_{k+1}m(0),n_{k+1}(0)=S_{k+1}n(0),
\end{array}
\right.$$
 Moreover, there is a positive time $T$ such that the solutions satisfying the following properties:

(i) $(u_k,v_k)_{k\in \mathbb{N}}$ is uniformly bounded in $E_{p,r}^{s}(T)\times E_{p,r}^{s}(T)$.

(ii) $(u_k,v_k)_{k\in \mathbb{N}}$ is a Cauchy sequence in $\mathcal{C}([0,T];B_{p,r}^{s-1})\times \mathcal{C}([0,T];B_{p,r}^{s-1})$.
\end{lm}
\begin{proof}
Since all the data $S_{k+1}u_0$  belong to $B_{l,r}^\infty$, Lemma 2.3 in \cite{ZMW} indicates 
that for all $n\in\mathbb{N}$, the equation $(T_k)$ has a global solution in 
$\mathcal{C}(\mathbb{R}^+;B_{l,r}^\infty)^2$.

For $s>\max\{2+\frac{1}{2},\frac{5}{2},3-\frac{1}{l}\}$ and $s\neq 3+\frac{1}{l}$, 
Lemma 2.1(i) implies the following
inequality
\begin{equation*}
\begin{split}
 \|m_{k+1}\|_{B^{s-2}_{l,r}} &\leq \exp\left(C\int_0^t\left(\|\partial_{x} v_{k} ^{p}(\tau)\|_{B^{s-4}_{l,r}} +\|\partial_{x} v_{k} ^{p}(\tau)\|_{B^{\frac{1}{l}}_{p,r}\cap L^\infty}\right) d\tau\right)\|m(0)\|_{B^{s-2}_{l,r}}\\
&+C\int_0^t\exp\left(C\int_\tau^t\left(\|\partial_{x} v_{k} ^{p}(\tau')\|_{B^{s-4}_{l,r}} +\|\partial_{x} v_{k} ^{p}(\tau')\|_{B^{\frac{1}{l}}_{p,r}\cap L^\infty}\right) d\tau'\right)\|\partial_xv^{p}_{k }m_{k }\|_{B^{s-2}_{l,r}}d\tau.
\end{split}
\end{equation*}
is true for all  $k\in\mathbb{N}$.
Due to $s>2+\frac{1}{l}$, we know that $B^{s-2}_{l,r}$ is an algebra and $B^{s-2}_{l,r}\hookrightarrow L^\infty$. Thus, we have
\begin{equation*}
\begin{split}
 \|m_{k+1}\|_{B^{s-2}_{l,r}} &\leq \exp\left(C\int_0^t \|  v_{k} (\tau)\|^{p}_{B^{s }_{l,r}}   d\tau\right)\|m(0)\|_{B^{s-2}_{l,r}}\\
&+C\int_0^t\exp\left(C\int_\tau^t \|  v_{k} (\tau')\|^{p}_{B^{s }_{l,r}}  d\tau'\right)\|v_{k }\|^{p}_{B^{s }_{l,r}}\|m_{k }\|_{B^{s-2}_{l,r}}d\tau.
\end{split}
\end{equation*}
Similarly, one can easily get the estimate for $ \|n_{k+1}\|_{B^{s-2}_{l,r}}$. Thus, adding the two resulted inequalities yields
\begin{equation}
\begin{split}
& \|u_{k+1}\|_{B^{s }_{l,r}}+ \|v_{k+1}\|_{B^{s }_{l,r}} \leq \exp\left(C\int_0^t\left( \|  v_{k} (\tau)\|^{p}_{B^{s }_{l,r}}+ \|  u_{k} (\tau)\|^{q}_{B^{s }_{l,r}} \right) d\tau\right)\left(\|u(0)\|_{B^{s }_{l,r}}+\|v(0)\|_{B^{s }_{l,r}}\right)\\
&\quad+C\int_0^t\exp\left(C\int_\tau^t \left( \|  v_{k} (\tau')\|^{p}_{B^{s }_{l,r}}+ \|  u_{k} (\tau')\|^{q}_{B^{s }_{l,r}} \right) d\tau'\right)\left(\|v_{k }\|^{p}_{B^{s }_{l,r}}\|u_{k }\|_{B^{s }_{l,r}}+\|u_{k }\|^{q}_{B^{s }_{l,r}}\|v_{k }\|_{B^{s }_{l,r}}\right)d\tau\\
&\quad\leq e^{CU_k(t)}\left(\|u(0)\|_{B^{s }_{l,r}}+\|v(0)\|_{B^{s }_{l,r}}+C\int_0^te^{-CU_k(\tau)}\left(\|v_{k }\|_{B^{s }_{l,r}} +\|u_{k }\|_{B^{s }_{l,r}} \right)^{\kappa+1}d\tau \right),
\end{split}
\end{equation}
where $U_k(t):=\int_0^t\left( \|  v_{k} (\tau)\| _{B^{s }_{l,r}}+ \|  u_{k} (\tau)\|_{B^{s }_{l,r}} \right)^{\kappa} d\tau\geq\int_0^t\left( \|  v_{k} (\tau)\|^{\kappa}_{B^{s }_{l,r}}+ \|  u_{k} (\tau)\|^{\kappa}_{B^{s }_{l,r}} \right) d\tau$ and $\kappa= \max\left\{p,q\right\}$.  Choosing $0<T<\frac{1}{2\kappa C\left( \|u_0 \|_{B^{s}_{l,r}}+\|v_0\|_{B^{s}_{l,r}}\right) ^\kappa}$, and suppose by induction that for all $t\in [0,T)$
\begin{align}
 \|u_{k}(t) \|_{B^{s}_{l,r}}+\|v_{k}(t) \|_{B^{s}_{l,r}}\leq\frac{\| u_{0}\|_{B^{s}_{l,r}}+\| v_{0}\|_{B^{s}_{l,r}}}{ \left(1-2 \kappa C t(\| u_{0}\|_{B^{s}_{l,r}} +\| v_{0}\|_{B^{s}_{l,r}})^\kappa\right)^{1/\kappa} } \doteq\frac{Z_0}{ \left(1-2 \kappa C tZ_0^\kappa\right)^{1/\kappa} } .
\end{align}
Noticing
$$\exp\left(C\int_\tau^t\left( \|  v_{k} (\tau')\| _{B^{s }_{l,r}}+ \|  u_{k} (\tau')\| _{B^{s }_{l,r}} \right)^{\kappa} d\tau'\right)\leq \exp\left(\int_\tau^t\left( \frac{CZ_0^\kappa}{\left(1-2 \kappa C  Z_0^\kappa\tau'\right)} \right) d\tau'\right)=\left(\frac{1-2 \kappa C  Z_0^\kappa\tau}{1-2 \kappa C  Z_0^\kappa t}\right)^\frac{1}{2\kappa}, $$
and inserting the above inequality and (2.5) into (2.4), we obtain
\begin{align*}
 \|u_{k+1}\|_{B^{s }_{l,r}}+ \|v_{k+1}\|_{B^{s }_{l,r}}&\leq\frac{Z_0}{\left(1-2 \kappa C  Z_0^\kappa t\right)^{\frac{1}{2\kappa}}}+\frac{C}{\left(1-2 \kappa C  Z_0^\kappa t\right)^{\frac{1}{2\kappa}}} \int^{t}_{0} \left(1-2 \kappa C  Z_0^\kappa \tau\right)^{\frac{1}{2\kappa}}\frac{Z_0^{\kappa+1}}{\left(1-2 \kappa C  Z_0^\kappa \tau\right)^{\frac{\kappa+1}{\kappa}}}d\tau\\
 &\leq \frac{Z_0}{\left(1-2 \kappa C  Z_0^\kappa t\right)^{\frac{1}{2\kappa}}}+\frac{Z_0}{-2w\left(1-2 \kappa C  Z_0^\kappa t\right)^{\frac{1}{2\kappa}}} \int^{t}_{0} \frac{d(1-2 \kappa C  Z_0^\kappa \tau)}{(1-2 \kappa C  Z_0^\kappa \tau)^\frac{2\kappa+1}{2\kappa}}\\
 &\leq \frac{Z_0}{\left(1-2 \kappa C  Z_0^\kappa t\right)^{\frac{1}{2\kappa}}}+\frac{Z_0}{ \left(1-2 \kappa C  Z_0^\kappa t\right)^{\frac{1}{2\kappa}}} \left(\frac{1}{\left(1-2 \kappa C  Z_0^\kappa t\right)^{\frac{1}{2\kappa}}}-1\right)\\
 &=\frac{Z_0}{ \left(1-2 \kappa C tZ_0^\kappa\right)^{1/\kappa} },
\end{align*}
which implies that $(u_k,v_k)_{k\in\mathbb{R}}$ is uniformly bounded in $\mathcal{C}([0,T];B_{l,r}^{s}) \times   \mathcal{C}([0,T];B_{l,r}^{s})$. Using the equation $(T_k)$ and the similar argument in the proof Lemma 2.1, one can easily prove that $(\partial_tu_k,\partial_tv_k)_{k\in\mathbb{R}}$ is uniformly bounded in $\mathcal{C}([0,T];B_{l, r}^{s-1}) \times \mathcal{C}([0,T];B_{l, r}^{s-1})$. Hence, $(u_k,v_k)_{k\in\mathbb{R}}$ is uniformly bounded in $E_{l, r}^{s }(T) \times E_{l, r}^{s }(T)$.

Let us now show that $(u_k,v_k)_{n\in\mathbb{R}}$ is a Cauchy sequence in   $\mathcal{C}([0,T];B_{l,r}^{s-1})\times \mathcal{C}([0,T];B_{l,r}^{s-1})$. In fact, for all $k,j\in \mathbb{N}$, from $(T_k)$, we have
$$ \left\{
\begin{array}{llll}
(\partial_t+v_{k+j}^p\partial_x)(m_{k+j+1}-m_{k+1})=F(t,x),\\
(\partial_t+u_{k+j}^q\partial_x)(n_{k+j+1}-n_{k+1})=G(t,x),\\
\end{array}
\right.$$
with
\begin{align*}
F(t,x)=(v_{k+j}^p-v_k^p)\partial_xm_{k+1}+\frac{a}{p}(m_k-m_{k+j})\partial_xv_{k+j}^p+\frac{a}{p}m_k\partial_x(v_k^p-v_{k+j}^p),\\
G(t,x)=(u_{k+j}^q-u_k^q)\partial_xn_{k+1}+\frac{b}{q}(n_k-n_{k+j})\partial_xu_{k+j}^q+\frac{b}{q}n_k\partial_x(u_k^q-u_{k+j}^q).
\end{align*}

Apparently, we have
\begin{align*}
 \| F(t,\cdot) \|_{B^{s-3}_{l,r}} +& \| G(t,\cdot) \|_{B^{s-3}_{l,r}}\leq C\left(\| u_{k+j} -u_k  \|_{B^{s-1}_{l,r}}+ \| v_{k+j} -v_k  \|_{B^{s-1}_{l,r}} \right)H(t),
\end{align*}
with
\begin{align*}
H(t)=&\left[\left(\| u_{k}  \|_{B^{s }_{l,r}}+\| u_{k+1}  \|_{B^{s }_{l,r}}\right)\sum_{i=0}^{p-1}\| v_{k+j}   \|_{B^{s }_{l,r}}^{p-1-i}\|  v_k  \|_{B^{s }_{l,r}}^{i}+\|  v_{k+j}  \|_{B^{s }_{l,r}}^{p}\right.\\
&\quad\left.+\left(\|v_{k}  \|_{B^{s }_{l,r}}+\|v_{k+1}  \|_{B^{s }_{l,r}}\right)\sum_{i=0}^{q-1}\| u_{k+j}   \|_{B^{s }_{l,r}}^{q-1-i}\|  u_k  \|_{B^{s }_{l,r}}^{i}+\|  u_{k+j}  \|_{B^{s }_{l,r}}^{q}\right].
\end{align*}
Similar to the proof of Lemma 2.1, for $s>\max\{2+\frac{1}{l},3-\frac{1}{l},\frac{5}{2}\}$ and $s\neq 3+\frac{1}{l},4+\frac{1}{l} $, we arrive at
\begin{align*}
V_{k+1}^j(t)&\doteq\| u_{k+j+1} -u_{k+1}  \|_{B^{s-1}_{l,r}}+ \| v_{k+j+1} -v_{k+1}  \|_{B^{s-1}_{l,r}} \\
&\leq \exp\left(CU_{k+j}(t)\right)\left(\| u_{k+j+1}(0) -u_{k+1} (0) \|_{B^{s-1}_{l,r}}+ \| v_{k+j+1}(0) -v_{k+1}(0)  \|_{B^{s-1}_{l,r}} \right.\\
&\quad\left.+C\int_0^t  \exp\left(-CU_{k+j}(\tau)\right)\left(\| F(\tau) \|_{B^{s-3}_{l,r}} + \| G(\tau) \|_{B^{s-3}_{l,r}}\right) d\tau\right)\\
&\leq \exp\left(CU_{k+j}(t)\right)\left(V_{k+1}^j(0)   +C\int_0^t  \exp\left(-CU_{k+j}(\tau)\right)V_{k}^j(\tau)
H(\tau) d\tau\right).
\end{align*}
As per Proposition 2.1 in \cite{ZMW}, we have
\begin{equation*}
\begin{split}
\| u_{k+j+1}(0) -u_{k+1} (0) \|_{B^{s-1}_{l,r}}&=\| S_{k+j+1}u(0) -S_{k+1} u(0) \|_{B^{s-1}_{l,r}}=
\left\|\sum_{d=k+1}^{k+j}\Delta_du_0\right\|_{B_{l,r}^{s-1}}\\
&=\left(\sum_{k\geq -1}2^{k(s-1)r}\left\|\Delta_k\left(\sum_{d=k+1}^{k+j}\Delta_du_0\right)\right\|_{L^p}^r\right)^\frac{1}{r}\\
&\leq C\left(\sum_{d=k}^{k+j+1}2^{-dr}2^{drs}\left\|\Delta_d u_0 \right\|_{L^p}^r\right)^\frac{1}{r}\leq C2^{-k}||u_0||_{B_{l,r}^{s }}.
\end{split}
\end{equation*}
Similarly, we obtain
 $$\| v_{k+j+1}(0) -v_{k+1} (0) \|_{B^{s-1}_{l,r}}\leq C2^{-k}||v_0||_{B_{l,r}^{s }}.$$
Due to $\{u_k,v_k\}_{k\in\mathbb{N}}$ being uniformly bounded in $E_{l, r}^{s }(T) $, one may  find a positive constant $C_T$ independent of $k,i$ such that
$$V_{k+1}^j(t)\leq C_T\left(2^{-k}+\int_0^tV_{k}^j(\tau)d\tau\right), \ \forall t\in [0,T]. $$
Employing 
the induction procedure with respect to the index $n$ leads to
\begin{equation*}
\begin{split}
V_{k+1}^j(t)&\leq C_T\left(2^{-k}\sum_{i=0}^k\frac{(2TC_T)^i}{i!}+C_T^{k+1}\int\frac{(t-\tau)^k}{k!}d\tau\right)\leq 2^{-k}\left(C_T\sum_{i=0}^k\frac{(2TC_T)^i}{i!}\right)+C_T  \frac{(TC_T)^{k+1}}{(k+1)!},
\end{split}
\end{equation*}
which reveals the desired result as $k\rightarrow\infty$.
Finally, 
applying the interpolation method 
to the critical case  $s=4+\frac{1}{l}$ concludes the proof of Lemma 2.2.
\end{proof}

Finally, we prove the existence and uniqueness for  Eq.(\ref{Eq.(1.1)}) in Besov space.
\begin{proof} [Proof of Theorem \ref{result1}]
 Let us first show that $(u,v)\in E_{l,r}^{s}(T)\times E_{l,r}^{s}(T)$ solves system (1.1).
According to Lemma 2.2 and Proposition 2.2(6) in \cite{ZMW}, one may see 
$$(u,v)\in L^\infty([0,T];B_{l,r}^{s })\times L^\infty([0,T];B_{l,r}^{s}).$$
For all $s'<s$, applying Lemma 2.2 
with an interpolation argument yields
$$(u_n,v_n)\rightarrow (u,v), \quad \text{ as } n\rightarrow\infty, \quad \text{ in }\mathcal{C}([0,T];B_{l,r}^{s' })\times \mathcal{C}([0,T];B_{l,r}^{s'}).$$
Taking limit in $T_n$ reveals 
that $(u,v)$ satisfy system (1.1) in the sense of $\mathcal{C}([0,T];B_{l,r}^{s'-1 })\times \mathcal{C}([0,T];B_{l,r}^{s'-1})$ for all $s'<s$.

Combining 
the fact $B_{l,r}^s$ is an algebra as $s>2+\frac{1}{l}$
with Lemma 2.2 and Lemma 2.3 in \cite{ZMW} yields $(u,v)\in E_{l,r}^{s}(T)\times E_{l,r}^{s}(T)$.

On the other hand, the continuity with respect to the initial data in
$$\mathcal{C}([0,T];B_{l,r}^{s'}) \cap \mathcal{C}^1([0,T];B_{l,r}^{s'-1})\times \mathcal{C}([0,T];B_{l,r}^{s'}) \cap \mathcal{C}^1([0,T];B_{l,r}^{s'-1})\quad \text {for all } s'<s,$$
can be obtained by Lemma 2.1 and a simple interpolation argument. While the continuity in $\mathcal{C}([0,T];B_{l,r}^{s }) \cap \mathcal{C}^1([0,T];B_{l,r}^{s -1})\times \mathcal{C}([0,T];B_{l,r}^{s }) \cap \mathcal{C}^1([0,T];B_{l,r}^{s -1}$
  when  $r<\infty$  can be proved
through the use of a sequence of viscosity approximation solutions
$(u_\epsilon,v_\epsilon)_{\epsilon>0}$	   for System (1.1) which
converges uniformly in $\mathcal{C}([0,T];B_{l,r}^{s }) \cap \mathcal{C}^1([0,T];B_{l,r}^{s -1})\times \mathcal{C}([0,T];B_{l,r}^{s }) \cap \mathcal{C}^1([0,T];B_{l,r}^{s -1})$.   This completes  the proof of Theorem 1.1.
\end{proof}

\subsection{Local well-posedness in a critical Besov space}
In this section, we shall
establish the local well-posedness of the Eq.(1.1) in critical Besov spaces. More precisely, we give the proof of Theorem 1.2.  More precisely, we give the proof of Theorem 1.2.  Motivated by the proof of
local existence about CH equation \cite{D23}, at first, we construct the approximate solutions of Eq.(1.1) by the classical Friedrichs regularization method.
{ \begin{lm}
Given $(u^0,v^0)=0$ and the initial data $(u_0,v_0)\in B_{2,1}^{\frac{5}{2}}$. Then there exists a sequence $\{(u^k,v^k)\}_{k\in\mathbb{N}}\in \mathcal{C} (\mathbb{R}^+;B_{2,1}^\infty)$ solving the linear initial value problem by induction $(T_k)$ (see, Lemma 2.1). Moreover,   the solutions $(u^k,v^k)$ satisfy the following properties:

(i) $(u^k,v^k)_{k\in \mathbb{N}}$ is uniformly bounded in $E_{2,1}^{\frac{5}{2}}(T)\times E_{2,1}^{\frac{5}{2}}(T)$.

(ii) $(u^k,v^k)_{k\in \mathbb{N}}$ is a Cauchy sequence in $\mathcal{C}([0,T];B_{2,\infty}^{\frac{3}{2}})\times \mathcal{C}([0,T];B_{2,\infty}^{\frac{3}{2}})$.
\end{lm}
\begin{proof}
Similar to   the proof of Lemma 3.1.  We first prove that the sequence $(u^k,v^k)_{k\in\mathbb{N}}$ (defined by ($T_k$))
is a Cauchy sequence
in $\mathcal{C}([0,T); B_{2,\infty}^{ \frac{3}{2}})\times \mathcal{C}([0,T); B_{2,\infty}^{ \frac{3}{2}})$. Then by $||u^k||_{B_{2,1}^{ \frac{5}{2}}}+||v^k||_{B_{2,1}^{ \frac{5}{2}}}\leq M$ and the interpolation inequality we infer 
   that $(u^k,v^k)_{k\in\mathbb{N}}$ tends to $(u^k,v^k)$ in $\mathcal{C}([0,T);B_{2,1}^{s}),\forall s<\frac{5}{2}.$
\end{proof}
By the Osgood Lemma \cite{D23}, we obtain the uniqueness of the solution to Eq.(1.2).
\begin{lm}
Assume that $z_0\doteq(u_0,v_0)$ and $w_0\doteq(u_0',v_0')$  $\in
B^{\frac{5}{2}}_{2,1}\times B^{\frac{5}{2}}_{2,1}$ such that $z\doteq(u,v)$ and $w\doteq(u',v')$  $\in
L^\infty([0,T];B^{\frac{5}{2}}_{2,1}\cap W^{1,\infty})
 \times L^\infty([0,T];B^{\frac{5}{2}}_{2,1}\cap W^{1,\infty})$ two solutions to the Cauchy problem
(1.1) with initial data $z_0$ and $w_0$. Let $(P,Q)=(u-u',v-v')$,
$(P_0,Q_0)=(u_0-u_0',v_0-v_0')$ and $Z(t)=||u||_{B^{\frac{5}{2}}_{2,\infty}\cap W^{1,\infty}} + ||u'||_{B_{2,\infty}^{\frac{5}{2}}\cap W^{1,\infty}}
+ ||v||_{B^{\frac{5}{2}}_{2,\infty}\cap W^{1,\infty}}  +||v'||_{B^{\frac{5}{2}}_{2,\infty}\cap W^{1,\infty}} $.
If there exists a time $T>0$ such that
\begin{equation}
\begin{split}
  \frac{ \|P_{0}-Q_{0}\|_{B^{\frac{3}{2}}_{2,\infty}}}{e}\leq e^{-\exp\left(c\int_0^TZ(\tau)\ln(e+Z(\tau))d\tau\right)} .
\end{split}
\end{equation}
 Then for every $t\in[0,T]$, we have 
\begin{equation}
\begin{split}
&\frac{\|P(t)-Q(t)\|_{B^{\frac{3}{2}}_{2,\infty}}}{e}\leq e^{-c\int_0^TZ(\tau)d\tau}\left(  \frac{ \|P_{0}-Q_{0}\|_{B^{\frac{3}{2}}_{2,\infty}}}{e}\right) ^{-\exp\left(c\int_0^TZ(\tau)\ln(e+Z(\tau))d\tau\right)}.
\end{split}
\end{equation}
\end{lm}}
\begin{proof}
The proof of this lemma is much similar to  Proposition 1 in \cite{D23}.  we omit  it
here.
\end{proof}

\begin{lm}
Denote $\bar{\mathbb{N}}\doteq \mathbb{N}\cup\infty$, For any $z_0\doteq (u_0,v_0)\in B^{\frac{5}{2}}_{2,1}\times B^{\frac{5}{2}}_{2,1}$, there
exists a neighborhood $V$ of $z_0$ in $B^{\frac{5}{2}}_{2,1}$ and a
time $T>0$ such that for any solution of the Cauchy problem (1.1)
$z\in V$, the map $\Phi:$
$$V\rightarrow \mathcal{C}([0,T];B^{\frac{5}{2}}_{2,1})\cap
\mathcal{C}^1([0,T];B^{\frac{3}{2}}_{2,1})\times \mathcal{C}([0,T];B^{\frac{5}{2}}_{2,1})\cap
\mathcal{C}^1([0,T];B^{\frac{3}{2}}_{2,1})$$
$z_0\mapsto z$ is continuous.
\end{lm}
\begin{proof}
Firstly, we prove the continuity of the map $\Phi$
in $\mathcal{C}([0,T];B^{\frac{3}{2}}_{2,1})\times \mathcal{C}([0,T];B^{\frac{3}{2}}_{2,1})$. Fix $z_0\in
B^{\frac{5}{2}}_{2,1}\times B^{\frac{5}{2}}_{2,1}$ and $\delta>0$. Now we claim that there
exists a $T>0$ and $M>0$ such that for any $z'_0\in
B^{\frac{5}{2}}_{2,1}\times B^{\frac{5}{2}}_{2,1}$ with $\|z'_0-z_0\|_{B^{\frac{5}{2}}_{2,1}\times B^{\frac{5}{2}}_{2,1}}\leq
\delta$, the solution $z=\Phi(z)$ of the Cauchy problem (1.1)
belongs to $\mathcal{C}([0,T];B^{\frac{5}{2}}_{2,1}))^2$ and satisfies
$\|z\|_{L^\infty([0,T];B^{\frac{5}{2}}_{2,1}))^2}\leq M.$ In fact,
according to the proof of the local well-posedness, we have that if
we fix a $T$ such that
$$0<T<\frac{1}{2\kappa C   \|z_0\|_{B^{s}_{l,r}} ^\kappa} ,$$
then
\begin{equation}
\begin{split}
 ||z'(t)||_{B_{2,1}^{\frac{5}{2}}}\leq \frac{   C||z'_0||_{B_{2,1}^{\frac{5}{2}}} }{   \left(1-2 \kappa C t||z'_0||_{B_{2,1}^{\frac{5}{2}}} ^\kappa\right)^{1/\kappa} }  \quad \text{ for all }t\in[0,T].
\end{split}
\end{equation}
Due to $\|z'_0-z_0\|_{B^{\frac{5}{2}}_{2,1}}\leq \delta$, it follows that
$\|z'_0\|_{B^{\frac{5}{2}}_{2,1}}\leq\|z_0\|_{B^{\frac{5}{2}}_{2,1}}+\delta$.
Here, one can actually choose some suitable constant $C$ such that
$$T=\frac{1}{4\kappa C(\|z_0\|_{B^{\frac{5}{2}}_{2,1}}+\delta )^\kappa}<\min\left\{\frac{1}{2\kappa C  \|z_0' \|_{B^{s}_{l,r}}  ^\kappa},\frac{1}{2C}\right\}$$
and
$$M=2^{1/\kappa}(\|z_0\|_{B^{\frac{5}{2}}_{2,1}}+\delta).$$
Combining the above uniform bounds with Lemma 2.4 yields
$$\|\Phi(v_0)-\Phi(u_0)\|_{L^\infty(0,T;B^{\frac{5}{2}}_{2,1})}\leq\delta
e^{2\kappa C M  ^\kappa T}.$$
Hence $\Phi$ is H\"{o}lder continuous from
$B^{\frac{5}{2}}_{2,1}\times B^{\frac{5}{2}}_{2,1}$ into $\mathcal{C}([0,T];B^{\frac{3}{2}}_{2,1})\times \mathcal{C}([0,T];B^{\frac{3}{2}}_{2,1})$.

Next we prove the continuity of the map $\Phi$ in
$\mathcal{C}([0,T];B^{\frac{5}{2}}_{2,1})\times \mathcal{C}([0,T];B^{\frac{3}{2}}_{2,1})$. Let $z^\infty(0)\doteq(u^\infty(0),v^\infty(0))\in
B^{\frac{5}{2}}_{2,1}\times B^{\frac{5}{2}}_{2,1}$ and $(z_k(0))_{k\in \mathbb{N}}\doteq(u_k(0),v_k(0))_{k\in \mathbb{N}}$ tend to $z^\infty_0$ in $B^{\frac{5}{2}}_{2,1}\times B^{\frac{5}{2}}_{2,1}$ as $k\rightarrow\infty$. Let $z_k\doteq (u_k,v_k)$ be the
solution of the Cauchy problem (1.1) with the initial data
$z_k(0)$. From the above procedure  for any $n\in
\mathbb{N}, t\in T$, we may obtain 

\begin{equation}
\sup_{k\in
\mathbb{N}}\|z_k\|_{L^\infty_T(B^{\frac{5}{2}}_{2,1})}\leq M.
\end{equation}
Apparently,  proving $z_k\rightarrow z_\infty$ in
$\mathcal{C}([0,T];B^{\frac{5}{2}}_{2,1})\times\mathcal{ C}([0,T];B^{\frac{5}{2}}_{2,1})$ is equivalent to proving that
$m_k=u_k-\partial_x^2u_k ,n_k=v_k-\partial_x^2v_k $ tends to
$p^{(\infty)}=u^{(\infty)}-u^{(\infty)}_{xx},q^{(\infty)}=v^{(\infty)}-v^{(\infty)}_{xx}$ in
$\mathcal{C}([0,T];B^{\frac{1}{2}}_{2,1})$ as $k\rightarrow\infty$.

Let us recall that $(u_k,v_k)$ solves the linear transport equation:
\begin{equation}
\left\{
\begin{array}{llll}
\partial_{t}m_{k}+v^{p}_{k }\partial_xm_{k}+\frac{a}{p} \partial_xv^{p}_{k }m_{k }=0, \\
\partial_{t}n_{k}+u^{q}_{k }\partial_xn_{k}+\frac{b}{q} \partial_xu^{q}_{k }n_{k }=0,\\
u_k|_{t=0}= u_k(0),v_k|_{t=0}= v_k(0).
\end{array}
\right.
\end{equation}
Thanks to the Kato theory \cite{D23}, we decompose $(m_k,n_k)$ into
$m_k=\alpha_k+\beta_k,q_k=\phi_k+\varphi_k$ with

\begin{equation}
\left\{
\begin{array}{llll}
\left[\partial_t+v^{p}_{k }\partial_x\right]\alpha_{(n)}=-\frac{a}{p} \partial_xv^{p}_{k }m_{k }+\frac{a}{p} \partial_xv^{p}_{\infty }m_{\infty },\\
\left[\partial_t+u^{q}_{k }\partial_x\right]\phi_k=-\frac{b}{q} \partial_xu^{q}_{k }n_{k }+\frac{b}{q} \partial_xu^{q}_{\infty }n_{\infty },\\
u_k|_{t=0}= u_k(0)-u_\infty(0),v_k|_{t=0}= v_k(0)- v_\infty(0),
\end{array}
\right.
\end{equation}
and
\begin{equation}
\left\{
\begin{array}{llll}
\left[\partial_t+v^{p}_{k }\partial_x\right]\alpha_k=-\frac{a}{p} \partial_xv^{p}_{\infty }m_{\infty },\\
\left[\partial_t+u^{q}_{k }\partial_x\right]\phi_k=-\frac{b}{q} \partial_xu^{q}_{\infty }n_{\infty },\\
u_k|_{t=0}= u_\infty(0),v_k|_{t=0}= v_\infty(0).
\end{array}
\right.
\end{equation}

In light of the properties of Besov spaces \cite{D23}, one may easily see that $( m_{k},  n_{k})_{k\in \overline{\mathbb{N}}}$ are uniformly bounded in $\mathcal{C}([0,T];B^\frac{1}{2}_{2,1})$. Moreover,
\begin{equation*}
\begin{split}
\| \frac{a}{p} \partial_xv^{p}_{k }m_{k }&-\frac{a}{p} \partial_xv^{p}_{\infty }m_{\infty }\|_{B^{\frac{1}{2}}_{2,1}}\leq C
 \|\partial_xv_k^p\|_{B^{\frac{1}{2}}_{2,1}}
\|m_k -m_\infty\|_{B^{\frac{1}{2}}_{2,1}}\\
&+C\left(  \| \partial_xv_k^p-\partial_xv_k^\infty\|_{B^{\frac{3}{2}}_{2,1}}\right)\|m_\infty\|_{B^{\frac{1}{2}}_{2,1}}.
\end{split}
\end{equation*}
Applying Lemma 4.3 in \cite{D23} and the product law in the
Besov spaces to equation (2.11) leads to
\begin{equation}
\begin{split}
\|\alpha_{(n)}\|_{B^{\frac{1}{2}}_{2,1}}&\leq\exp\left\{C\int^{t}_{0}\| v_k^p  (\tau)\|_{B^{\frac{3}{2}}_{2,1}}\,d\tau\right\}\\
&\quad\cdot\left(\|m_k(0)-m_\infty(0)\|_{B^{\frac{1}{2}}_{2,1}} + \int_0^t \left\| \frac{a}{p} \partial_xv^{p}_{k }m_{k } -\frac{a}{p} \partial_xv^{p}_{\infty }m_{\infty }\right\|_{B^{\frac{1}{2}}_{2,1}}d\tau \right) .
\end{split}
\end{equation}

On the other hand, 
we know that the sequence
$(u_k,v_k)_{n\in \overline{\mathbb{N}} \doteq\mathbb{N}\cup\{\infty\}}$ is uniformly bounded in
$\mathcal{C}([0,T];B^{\frac{5}{2}}_{2,1})\times \mathcal{C}([0,T];B^{\frac{5}{2}}_{2,1})$ and tends to $(u^{\infty},v^{\infty})$ in
$\mathcal{C}([0,T];B^{\frac{3}{2}}_{2,1})\times \mathcal{C}([0,T];B^{\frac{3}{2}}_{2,1})$ as $k\rightarrow\infty$.
Therefore, adopting Proposition 3 in \cite{D23} reveals that $(\alpha_k,\phi_k)$ tends to $(m_{\infty},n_{\infty})$
in $\mathcal{C}([0,T];B^{\frac{1}{2}}_{2,1})\times \mathcal{C}([0,T];B^{\frac{1}{2}}_{2,1})$.
Hence, combining the above result of convergence with estimates (4.3) and (4.7), for large enough $n\in\mathbb{N}$ we obtain
\begin{equation*}
\begin{split}
 &\|m_k-m_\infty \|_{B^{\frac{1}{2}}_{2,1}}+\|n_k-n_\infty \|_{B^{\frac{1}{2}}_{2,1}}\leq\varepsilon+C  M^{p+q}   e^{C M^{p+q} T }\biggl[\|m_k(0)-m_\infty (0) \|_{B^{\frac{1}{2}}_{2,1}}+\|n_k(0)-n_\infty(0) \|_{B^{\frac{1}{2}}_{2,1}}\\
            &+\int^{t}_{0}\left(\|m_k-m_\infty \|_{B^{\frac{1}{2}}_{2,1}}+\|n_k-n_\infty \|_{B^{\frac{1}{2}}_{2,1}} \right) d\tau+\int^{t}_{0} \|u_k-u_\infty \|_{B^{\frac{1}{2}}_{2,1}}+\|v_k-v_\infty \|_{B^{\frac{1}{2}}_{2,1}}d\tau\biggr].
\end{split}
\end{equation*}
By the Gronwall's inequality, we have
$$\|m_k-m_\infty \|_{L^\infty(0,T;B^{\frac{1}{2}}_{2,1})}+\|n_k-n_\infty \|_{L^\infty(0,T;B^{\frac{1}{2}}_{2,1})}\leq C\left(\|m_k(0)-m_\infty(0) \|_{B^{\frac{1}{2}}_{2,1}}+\|n_k(0)-n_\infty(0) \|_{B^{\frac{1}{2}}_{2,1}}+\varepsilon\right) $$
where the constant $C$ depends only on $M$ and $T$. 
Continuity of the map $\Phi$ in
$\mathcal{C}([0,T];B^{\frac{5}{2}}_{2,1})\times \mathcal{C}([0,T];B^{\frac{5}{2}}_{2,1})$ is now completed.
Applying $\partial_t$ to equation (1.1) and  repeating the same
procedure to the resulting equation in terms of $(\partial_t u,\partial_t v)$,
we may check continuity of the map $\Phi$ in
$\mathcal{C}^1([0,T];B^{\frac{3}{2}}_{2,1})\times \mathcal{C}^1([0,T];B^{\frac{3}{2}}_{2,1})$.
\end{proof}

\begin{proof} [Proof of Theorem \ref{result2}]
Thanks to Lemma 2.3, $\{u^n,v^n\}_{n\in \mathbb{N}}$
is uniformly
bounded in $B_{2,1}^\frac{3}{2}$
and converges to $(u,v)$ in $\mathcal{C}([0,T],B_{2,1}^\frac{3}{2})$, applying interpolation
theorem, one can easily get that the convergence holds in  $\mathcal{C}([0,T],B_{2,1}^{s_1}),s_1<5/2$.
Passing to limit in $(T_k)$ (see, Lemma 2.1), we obtain that $(u,v)$ is a solution to Eq.(1.1) and
satisfies
  $$ (u,v)\in \mathcal{C}([0,T];B^{\frac{5}{2}}_{2,1})\cap \mathcal{C}^1([0,T];B^{\frac{3}{2}}_{2,1})\times \mathcal{C}([0,T];B^{\frac{5}{2}}_{2,1})\cap \mathcal{C}^1([0,T];B^{\frac{3}{2}}_{2,1}).$$
Thanks to Lemma 2.4 and  Lemma 2.5, it follows that the uniqueness and
continuity with the initial data  $ (u_0,v_0)\in \mathcal{C}([0,T];B^{\frac{5}{2}}_{2,1})\times  \mathcal{C}([0,T];B^{\frac{5}{2}}_{2,1})$. This concludes the
proof of Theorem 1.2.
\end{proof}

\section{Blow-up criterion}
In this section, we shall build up 
a blow-up criteria.
Let us first recall the following two lemmas.
\begin{lm}
(See \cite{ZMW}) If $r>0$, then $H^r\cap L^\infty$ is an algebra. and
$$||fg||_{H^r}\leq c(||f||_{L^\infty}||g||_{h^r}+||g||_{L^\infty}||h||_{h^r}),$$
where $c$ is a constant depending only on $r$.
\end{lm}
\begin{lm}
(See \cite{ZMW}) If $r>0$, then
$$||[D ^r,f]g||_{L^2}\leq c(||\partial_x f||_{L^\infty}||D ^{r-1}g||_{L^2}+||D ^r f||_{L^2}||g||_{L^\infty}),$$
where $[A,B]$ denote the commutator of linear operator $A$ and $B$, and $c$ is a constant depending only on $r$.
\end{lm}

\begin{proof}[Proof of Theorem \ref{result3}]
We prove this theorem by an inductive
method with respect to index $s$. This can be achieved as follow.

\textbf{Step 1.} For $s\in\left(\frac{1}{2},1\right)$, applying Lemma 2.2 in \cite{ZMW} to the first equation
(1.1), one gets
\begin{equation}
\begin{split}
\|m(t)\|_{H^s}+\|n(t)\|_{H^s}\leq& \|m_0\|_{H^s}+C\int^t_0\left(\| m\partial_xv^{p }  (\tau)\|_{H^s}+\|m (\tau)\|_{H^s}\|\partial_xv^{p} (\tau)\|_{L^\infty}\right) d\tau
\end{split}
\end{equation}
for all $t\in(0,T^*_{m_0,n_0})$. 
Let $u=G\ast
m=(1-\partial^2_x)^{-1}m,v=G\ast
n=(1-\partial^2_x)^{-1}n,$ where $G=\frac{1}{2}e^{-|x|}$ and $\ast$
stands for the convolution on $\mathbb{R}$. Then
$u_x=\partial_xG\ast m$ where
$\partial_xG(x)=-\frac{1}{2}sign(x)e^{-|x|}.$
As per the  Young inequality, we have
\begin{equation}
\begin{split}
\|u\|_{L^\infty}\leq\|G\|_{L^1}\|m\|_{L^\infty}\leq
C\|m\|_{L^\infty},\|v\|_{L^\infty} \leq
C\|n\|_{L^\infty},\\
\|u_x\|_{L^\infty}\leq\|\partial_xG\|_{L^1}\|m\|_{L^\infty}\leq
C\|m\|_{L^\infty},\|v_x\|_{L^\infty} \leq
C\|n\|_{L^\infty}.
\end{split}
\end{equation}
Utilizing Eq. (3.2),  $\|u_x\|_{H^s}\leq C\|m\|_{H^s},\|v_x\|_{H^s}\leq C\|n\|_{H^s}$ and the Moser-type estimates leads
to
 \begin{equation}
\begin{split}
\| m\partial_xv^{p } (\tau)\|_{H^s}&\leq C\left(\| \partial_xv^{p }\|_{L^\infty}\|m\|_{H^s}+\| m\|_{L^\infty}\|\partial_xv^{p }\|_{H^s}\right)\\
&\leq C\left(\|  n\|_{L^\infty}^{p }\|m\|_{H^s}+\| m\|_{L^\infty}\|n\|^{p -1}_{L^\infty}\|n\|_{H^s} \right),\\
\end{split}
 \end{equation}
and
 \begin{equation}
\begin{split}
\|m (\tau)\|_{H^s}\|\partial_xv^{p} (\tau)\|_{L^\infty}&\leq C \|  n\|_{L^\infty}^{p }\|m\|_{H^s}.
\end{split}
 \end{equation}

Plugging Eqs. (3.3) and (3.4) into (3.1) generates
\begin{equation}
\begin{split}
\|m(t)\|_{H^s} \leq& \|m_0\|_{H^s} +C\int^t_0  \left((\|  n\|_{L^\infty}^{p } +\| m\|_{L^\infty}\|n\|^{p -1}_{L^\infty}) (\|m\|_{H^s}+\|n\|_{H^s})\right)\,d\tau.
\end{split}
\end{equation}
Similarly, applying Lemma 2.2 in \cite{ZMW} to the second equation of the system (1.1) yields
\begin{equation*}
\begin{split}
\|n(t)\|_{H^s} \leq&  \|n_0\|_{H^s}+C\int^t_0  \left((\|  m\|_{L^\infty}^{q } +\| n\|_{L^\infty}\|m\|^{q-1}_{L^\infty}) (\|m\|_{H^s}+\|n\|_{H^s})\right)\,d\tau.
\end{split}
\end{equation*}
Therefore
\begin{equation*}
\begin{split}
\|m(t)\|_{H^s} &+\|n(t)\|_{H^s}\leq\|m_0\|_{H^s}+\|n_0\|_{H^s}\\&+C\int^t_0  \left(\|  n\|_{L^\infty}^{p } +\| m\|_{L^\infty}\|n\|^{p -1}_{L^\infty} +\|  m\|_{L^\infty}^{q } +\| n\|_{L^\infty}\|m\|^{q-1}_{L^\infty}\right)(\|n\|_{H^s}+ \|m\|_{H^s})  \,d\tau.
\end{split}
\end{equation*}

Then, by the Gronwall's inequality we obtain
\begin{equation}
\|m(t)\|_{H^s}+\|n(t)\|_{H^s}\leq(\|m_0\|_{H^s}+\|n_0\|_{H^s})\exp\left\{C\int^t_0\left(\|  n\|_{L^\infty}^{p } +\| m\|_{L^\infty}\|n\|^{p -1}_{L^\infty} +\|  m\|_{L^\infty}^{q } +\| n\|_{L^\infty}\|m\|^{q-1}_{L^\infty}\right)d\tau\right\}.
\end{equation}
Moreover, if there exists a maximal time $T^*_{m_0,n_0}<\infty$ such that
\begin{equation*}
\int^{T^*_{m_0,n_0}}_0\left(\|  n\|_{L^\infty}^{p } +\| m\|_{L^\infty}\|n\|^{p -1}_{L^\infty} +\|  m\|_{L^\infty}^{q } +\| n\|_{L^\infty}\|m\|^{q-1}_{L^\infty}\right) d\tau<\infty,
\end{equation*}
then Eq. (3.6) implies that
\begin{equation}
\limsup_{t\rightarrow T^*_{m_0,n_0}}(\|m(t)\|_{H^s}+\|n(t)\|_{H^s})<\infty.
\end{equation}
which is contradicted to the assumption 
$T^*_{m_0,n_0}<\infty.$

\textbf{Step 2.} For $s\in[1,2),$  differentiating (1.1) with respect to $x$ yields
\begin{equation}
\begin{split}
&\partial_t(m_x)+ v^{p}\partial_x(m_x)=-\frac{a+p}{p}(v^{p})_{x}m_x-\frac{a}{p}(v^{p }) _{xx}m,\\
&\partial_t(n_x)+ u^{q}\partial_x(n_x)=-\frac{b+q}{q}(u^{q})_{x}n_x-\frac{b}{q}(u^{q } )_{xx}n.
\end{split}
\end{equation}
By Lemma 2.2 in \cite{ZMW}, we have
\begin{equation}
\begin{split}
\|\partial_xm(t)\|_{H^{s-1}} \leq \|\partial_xm_0\|_{H^{s-1}}
 +C\int^t_0 \| (v^{p})_{x}m_x+(v^{p }) _{xx}m\|_{H^{s-1}} \,d\tau  +C\int^t_0 \| m_x\|_{H^{s-1}}\|(v^p)_x\| _{L^\infty}d\tau\\
\end{split}
\end{equation}
According to the Moser-type estimates in \cite{ZMW} and (3.2), we obtain
\begin{equation}
\begin{split}
\|  (v^{p }) _{xx}m\|_{H^{s-1}}&\leq C \left(\|  (v^{p }) _{xx}\|_{L^\infty}\|m\|_{H^{s-1}}+\| m\|_{L^\infty}\|(v^{p }) _{xx}\|_{H^{s-1}}\right)\\
&\leq C\left(\|  n\|_{L^\infty}^{p }\|m\|_{H^s}+\| m\|_{L^\infty}\|n\|^{p -1}_{L^\infty}\|n\|_{H^s} \right),\\
\end{split}
 \end{equation}
and
\begin{equation}
\begin{split}
\| m_x\partial_xv^{p } \|_{H^{s-1}}&=\| \partial_x[m(v^{p })_x]- (v^{p }) _{xx}m\|_{H^{s-1}}
\leq C\left(\| m_xv^{p } \|_{H^{s }}+\|(v^{p }) _{xx}m\|_{H^{s-1}}\right)\\
&\leq C\left(\|  n\|_{L^\infty}^{p }\|m\|_{H^s}+\| m\|_{L^\infty}\|n\|^{p -1}_{L^\infty}\|n\|_{H^s} \right).
\end{split}
 \end{equation}
Plugging Eqs. (3.10)  and (3.11) into Eq. (3.9) gives
\begin{equation*}
\begin{split}
\|\partial_xm(t)\|_{H^{s-1}}&\leq \|\partial_xm_0\|_{H^{s-1}}
 +C\int^t_0\left(\|  n\|_{L^\infty}^{p }\|m\|_{H^s}+\| m\|_{L^\infty}\|n\|^{p -1}_{L^\infty}\|n\|_{H^s} \right) \,d\tau.
\end{split}
\end{equation*}
A similar procedure run for the second equation in (3.8) produces
\begin{equation*}
\begin{split}
\|\partial_xm(t)\|_{H^{s-1}}&+\|\partial_xn(t)\|_{H^{s-1}}\leq \|\partial_xm_0\|_{H^{s-1}}+\|\partial_xn_0\|_{H^{s-1}}\\
&+C\int^t_0  \left(\|  n\|_{L^\infty}^{p } +\| m\|_{L^\infty}\|n\|^{p -1}_{L^\infty} +\|  m\|_{L^\infty}^{q } +\| n\|_{L^\infty}\|m\|^{q-1}_{L^\infty}\right)(\|n\|_{H^s}+ \|m\|_{H^s})  \,d\tau
\end{split}
\end{equation*}
Considering the estimate for (3.6) and the fact
$$\|\partial_xm(t,\cdot)\|_{H^{s-1}}\leq C\|m(t)\|_{H^s},\|\partial_xn(t,\cdot)\|_{H^{s-1}}\leq C\|n(t)\|_{H^s},$$
one may see 
that (3.6) holds for all $s\in[1,2)$.
Repeating the same procedure 
as shown in Step 1, we know that Theorem 1.3 holds for
all $s\in[1,2)$.

\textbf{Step 3.} Let us assume that Theorem 1.3 holds for $k-1\leq s< k$ and
$2\leq k\in\mathbb{N}$. By the mathematical induction, we need to prove
that it is true for $k\leq s<k+1$ as well. Differentiating (1.1) $k$ times
with respect to $x$ leads to
\begin{equation}
\begin{split}
&\partial_t(\partial_x^km_x)+ v^p\partial_x^k(m_x)=-\frac{a}{p} \partial_x^{k }[(v^p )_x m]-\sum_{l=0}^{k-1} C_k^l\partial_x^{k-l}(v^p ) \partial_x^{l}(m_x),\\
&\partial_t(\partial_x^kn_x)+ u^q\partial_x^k(n_x)=-\frac{b}{q}\partial_x^{k }[(u^q)_x n]-\sum_{l=0}^{k-1} C_k^l\partial_x^{k-l}(v^p) \partial_x^{l}(n_x).\\
\end{split}
\end{equation}
According to Lemma 2.2 in \cite{ZMW}, we have 
\begin{equation}
\begin{split}
\|\partial_x^km(t)\|_{H^{s-k}}   \leq&\|\partial_x^km_0\|_{H^{s-k}}+C\int^t_0\| (v^p)_x  (\tau)\|_{L^\infty} \|\partial_x^km(\tau)\|_{H^{s-n}} \,d\tau \\
&+C\int^t_0\left(\left\|\sum_{l=0}^{k-1} C_k^l\sum_{l=0}^{k-1} C_k^l\partial_x^{k-l}(v^p )\partial_x^{l}(m_x)\right\|_{H^{s-k}}+\|\partial_x^{k }[(v^p )_x m]\|_{H^{s-k}}\right) d\tau.
\end{split}
\end{equation}
By the Moser-type estimate and the Sobolev embedding inequality, we derive 
\begin{equation}
\begin{split}
\|\partial_x^k[(v^p )_x m]\|_{H^{s-k}}&\leq C \|(v^p )_x m\|_{H^{s}}\leq C\left(\|  n\|_{L^\infty}^{p }\|m\|_{H^s}+\| m\|_{L^\infty}\|n\|^{p -1}_{L^\infty}\|n\|_{H^s} \right),
\end{split}
\end{equation}
where we used the Sobolev embedding theorem
$H^{s-\frac{1}{2}+\epsilon_0}(\mathbb{R})\hookrightarrow
L^\infty(\mathbb{R})$ ($s\geq2$) and
\begin{equation}
\begin{split}
\left\|\sum_{l=0}^{k-1} C_n^l\partial_x^{k-l}v^p \partial_x^{l+1}m\right\|_{H^{s-k}}&\leq C \sum_{l=0}^{k-1} \bigl(\left\|\partial_x^{k-l}v^p \right\|_{L^{\infty}} \left\|\partial_x^{l+1}m\right\|_{H^{s-k}}+\left\|\partial_x^{k-l}v^p \right\|_{H^{s-k}} \left\|\partial_x^{l+1}m\right\|_{L^{\infty}}\bigr)\\
&\leq C \sum_{l=0}^{n-1} \bigl(\left\|v^p\right\|_{H^{k-l+\frac{1}{2}+\epsilon_0}} \left\| m \right\|_{H^{s-k+l+1}}+\left\| v^p \right\|_{H^{s-l}} \left\| m \right\|_{H^{l+1+\frac{1}{2}+\epsilon_0}}\bigr)\\
&\leq C  \left\| n \right\|_{H^{s-\frac{1}{2} +\epsilon_0}}^p\left\| m\right\|_{H^{s }} ,
\end{split}
\end{equation}
with $\epsilon_0\in(0,\frac{1}{4})$ and
$H^{\frac{1}{2}+\epsilon_0}(\mathbb{R})\hookrightarrow
L^\infty(\mathbb{R})$.
Plugging Eqs. (3.14) and (3.15) into Eq. (3.13) yields 
\begin{equation*}
\begin{split}
\|m(t)\|_{H^s}& \leq \|m_0\|_{H^s}  +C\int^t_0  \left(\left\| n \right\|^p_{H^{s-\frac{1}{2} +\epsilon_0}}+\left\| m \right\|_{H^{s-\frac{1}{2} +\epsilon_0}}\left\| n \right\|^{p-1}_{H^{s-\frac{1}{2} +\epsilon_0}}\right) (\|m(\tau)\|_{H^s}+\|n(\tau)\|_{H^s})\,d\tau.
\end{split}
\end{equation*}
and
\begin{equation}
\begin{split}
&\|m(t)\|_{H^s}+\|n(t)\|_{H^s} \leq \|m_0\|_{H^s}+\|n_0\|_{H^s}+C\int^t_0 (\|m(\tau)\|_{H^s}+\|n(\tau)\|_{H^s}) \\
  &\times \left(\left\| n \right\|^p_{H^{s-\frac{1}{2} +\epsilon_0}}+\left\| m \right\|_{H^{s-\frac{1}{2} +\epsilon_0}}\left\| n \right\|^{p-1}_{H^{s-\frac{1}{2} +\epsilon_0}}+\left\| m \right\|^q_{H^{s-\frac{1}{2} +\epsilon_0}}+\left\| n \right\|_{H^{s-\frac{1}{2} +\epsilon_0}}\left\| m \right\|^{q-1}_{H^{s-\frac{1}{2} +\epsilon_0}}\right)\,d\tau
\end{split}
\end{equation}
Then, by Gronwall's inequality, we obtain
\begin{equation}
\begin{split}
&\|m(t)\|_{H^s}+\|n(t)\|_{H^s}\leq(\|m_0\|_{H^s}+\|n_0\|_{H^s})\\
&\exp\left\{C\int^t_0\left(\left\| n \right\|^p_{H^{s-\frac{1}{2} +\epsilon_0}}+\left\| m \right\|_{H^{s-\frac{1}{2} +\epsilon_0}}\left\| n \right\|^{p-1}_{H^{s-\frac{1}{2} +\epsilon_0}}+\left\| m \right\|^q_{H^{s-\frac{1}{2} +\epsilon_0}}+\left\| n \right\|_{H^{s-\frac{1}{2} +\epsilon_0}}\left\| m \right\|^{q-1}_{H^{s-\frac{1}{2} +\epsilon_0}}\right)\right\}.
\end{split}
\end{equation}
If there exist a maximal existence time $T^*_{m_0,n_0}<\infty$ such that
\begin{equation*}
\int^{T^*_{m_0,n_0}}_0\left(\|  n\|_{L^\infty}^{p } +\| m\|_{L^\infty}\|n\|^{p -1}_{L^\infty} +\|  m\|_{L^\infty}^{q } +\| n\|_{L^\infty}\|m\|^{q-1}_{L^\infty}\right) d\tau<\infty,
\end{equation*}
then by the solution uniqueness in Theorem 1.1,  we know that   $\left\| n \right\|^p_{H^{s-\frac{1}{2} +\epsilon_0}}+\left\| m \right\|_{H^{s-\frac{1}{2} +\epsilon_0}}\left\| n \right\|^{p-1}_{H^{s-\frac{1}{2} +\epsilon_0}}+\left\| m \right\|^q_{H^{s-\frac{1}{2} +\epsilon_0}}+\left\| n \right\|_{H^{s-\frac{1}{2} +\epsilon_0}}\left\| m \right\|^{q-1}_{H^{s-\frac{1}{2} +\epsilon_0}} $ is uniformly bounded in $t\in (0,T^*_{m_0,n_0})$.
As per the mathematical induction assumption,  we have 
\begin{equation*}
\limsup_{t\rightarrow T^*_{m_0,n_0}}(\|m(t)\|_{H^s}+\|n(t)\|_{H^s})<\infty.
\end{equation*}
which is a contradiction.
Therefore,  Steps 1 to 3 complete the proof of Theorem 1.3.
\end{proof}

To prove Theorem 1.4,   let us 
rewrite the Cauchy problem of the transport equation
(\ref{Eq.(1.1)})  as follows
\begin{equation}
\left\{
\begin{array}{llll}
u_t+v^pu_x+I_1(u,v)=0,\\
v_t+u^qv_x+I_2(u,v)=0,
\end{array}
\right.
\end{equation}
where
\begin{equation*}
\left\{
\begin{array}{llll}
I_1(u,v)=(1-\partial_x^2)^{-1}[ a  v^{p-1}v_xu+(p-a)v^{p-1}v_xu_{xx}]+p(1-\partial_x^2)^{-1}\partial_x(v^{p-1}v_xu_x),\\
I_2(u,v)=(1-\partial_x^2)^{-1}[ b  u^{q-1}u_xv+(q-b)u^{q-1}u_xv_{xx}]+q(1-\partial_x^2)^{-1}\partial_x(u^{q-1}u_xv_x).
\end{array}
\right.
\end{equation*}
Let us first provide 
the sufficient conditions for global existence of the solutions to Eq.(1.1).
\begin{thm}
 Let $z_0=(u_0, v_0) \in  H^s \times H^{s } $ with $s> 5/2$, and $T$ be the maximal time of the solution $z=(u, v)$ to system (1.1) with the initial data $z_0$. If there exists $M>0$ such that
$$\Gamma \doteq\left(||u ||^{q-1}_{L^\infty} +||v ||^{p-1}_{L^\infty}  \right)\left(||u_x||_{L^\infty} +||v _x||_{L^\infty}  \right)\leq M,\quad t\in[0,T),$$
then the $H^s\times H^s$-norm of $z(t,\cdot)$ does not blow up on $[0,T)$.
\end{thm}

\begin{proof}
  Let $z=(u,v)$  be the solution to  system (1.1) with the initial data $z_0\in  H^s \times H^{s} $,  $s> 5/2$, and $T$ be the maximal existence of the solution $z$ 
  as per Theorem 1.1. 

Applying the operator $ D ^s$ to the first and second equations in (3.18), multiplying by $ D ^su$ and $ D ^sv$, and integrating over $\mathbb{R}$, we may arrive at 
\begin{equation}
\frac{1}{2}\frac{d}{dt}||u||^2_{H^s}+ (v^p u_x,u)_s+ (u,I_1(u,v))_s=0,
\end{equation}
\begin{equation}
\frac{1}{2}\frac{d}{dt}||v||^2_{H^s}+ (u^q v_x,v)_s+ (v,I_2(u,v))_s=0,
\end{equation}
where
\begin{equation*}
\begin{split}
I_1(u,v)=(1-\partial_x^2)^{-1}\left[ \frac{a}{p} (v^{p})_xu+ \frac{p-a }{p}(v^{p})_xu_{xx}\right]+(1-\partial_x^2)^{-1}\partial_x[(v^{p})_xu_x],\\
I_2(u,v)=(1-\partial_x^2)^{-1}\left[ \frac{b}{q} (u^{q})_xv+ \frac{q-b }{q}(u^{q})_xv_{xx}\right]+(1-\partial_x^2)^{-1}\partial_x[(u^{q})_xv_x].
\end{split}
\end{equation*}
Let us estimate the right-hand side of (3.19).
\begin{equation*}
\begin{split}
|(v^pu_x,u)_s|&=|(D ^sv^pu_x,D ^su)_0|=|([D ^s,v^p]u_x,D ^su)_0+(v^pD ^su_x,D ^su)_0|\\
&\leq ||[D ^s,v^p]u_x||_{L^2}||D ^su||_{L^2}+\frac{1}{2}|( (v^p)_xD ^su ,D ^su)_0|\\
&\leq c (|| (v^p)_x||_{L^\infty}||u||_{H^s}+||u_x||_{L^\infty}||  v^p ||_{H^s})\times ||u||_{H^s}+\frac{1}{2}||  (v^p)_x||_{L^\infty}||u||^2_{H^s} \\
&\leq c(||  (v^p)_x||_{L^\infty}||u||^2_{H^s}+||u_x||_{L^\infty}||  v^p ||_{H^s}||u||_{H^s}).
\end{split}
\end{equation*}
where 
Lemma 3.2 with $r=s$ is used.

By Lemma 3.1 and the mathematical induction, we have  $||  v^p ||_{H^s}\leq p||  v ||^{p-1}_{L^\infty}||  v  ||_{H^s}$ and 
\begin{equation*}
\begin{split}
|(v^pu_x,u)_s| \leq c||  v ||^{p-1}_{L^\infty}(||v_x||_{L^\infty}+||u_x||_{L^\infty})(||  v||_{H^s}+||u||_{H^s})^2.
\end{split}
\end{equation*}
Therefore, we obtain
\begin{equation*}
\begin{split}
|(I_1(z),u)_s|&\leq ||I_1(z)||_{H^s}||u||_{H^s}\leq c (||(v^{p})_xu||_{H^{s-2}}+||(v^{p})_xu_{xx}||_{H^{s-2}}+ || (v^{p})_xu_x||_{H^{s-1}} )||u||_{H^s}\\
&\leq c||  (v^p)_x||_{L^\infty} ||u||_{H^s}^2,
\end{split}
\end{equation*}
which reveals
$$\frac{1}{2}\frac{d}{dt} ||u||^2_{H^s}  \leq  c||  v ||^{p-1}_{L^\infty}(||v_x||_{L^\infty}+||u_x||_{L^\infty})(||u||_{H^s} +||v||_{H^s})^2.$$
In a similar way, from (3.20) we can get the estimate for $||v||^2_{H^s} $. So, we arrive at
\begin{equation*}
\begin{split}
\frac{1}{2}\frac{d}{dt}&\left(||u||_{H^s}+||v||_{H^s}\right)^2\leq \frac{d}{dt}\left(||u||_{H^s}^2+||v||_{H^s}^2\right)\\
&\leq c\Gamma\left(||u_x||_{L^\infty},||v_x||_{L^\infty}\right)(||u||_{H^s} +||v||_{H^s})^2.
\end{split}
\end{equation*}
Adopting the Gronwall's inequality and the assumption of the theorem imply 
$$||u|| _{H^s}+||v|| _{H^s}\leq \exp (cMt)(||u_0|| _{H^s}+||v_0|| _{H^s}),$$
which completes the proof of Theorem 3.1.
\end{proof}

In the following, we apply Theorem 3.1 to show the blow-up scenario for Eq.(\ref{Eq.(1.1)}).

\begin{proof}[Proof of Theorem \ref{result4}]
 Let $z=(u,v)$ be the
solution to Eq.(1.1)   with the initial data $(u_0,v_0)\in
H^s\times H^s $ and $s> 5/2$ and $T$ be the maximal existence
of the solution $(u,v)$.

Multiplying both sides of Eq.(\ref{Eq.(1.1)}) by $m$ and integrating by parts, we have
\begin{equation}
\begin{split}
\frac{d}{dt}\int_\mathbb{R} m^2dx&=2\frac{d}{dt}\int_\mathbb{R} mm_tdx=-2\int_\mathbb{R}m(v^pm_x+\frac{a}{p}(v^p)_xm)dx=\frac{p-2a}{p}\int_\mathbb{R}m^2(v^p)_xdx.\\
\frac{d}{dt}\int_\mathbb{R} n^2dx& =\frac{q-2b}{q}\int_\mathbb{R}n^2(u^q)_xdx.
\end{split}
\end{equation}
We also notice 
\begin{align*}
\|u(t,\cdot)\|^{2}_{H^{2}}\leq \|m(t,\cdot)\|_{L^{2}}^{2}\leq 2\|u(t,\cdot)\|_{H^{2}}^{2},\quad\|v(t,\cdot)\|^{2}_{H^{2}}\leq \|n(t,\cdot)\|_{L^{2}}^{2}\leq 2\|n(t,\cdot)\|_{H^{2}}^{2}.
\end{align*}
 Casting  $p=2a,q=2b$ in Eq. (3.21) yields 
$$\|u_x(t,\cdot)\|^{2}_{L^\infty}\leq\|u(t,\cdot)\|^{2}_{H^{2}}\leq \|m(t,\cdot)\|_{L^{2}}^{2}= \|m(0,\cdot)\|_{L^{2}}^{2}<\infty,\quad\|v_x(t,\cdot)\|^{2}_{L^\infty}\leq  \|n(0,\cdot)\|_{L^{2}}^{2}<\infty.$$
In view of Theorem 3.1 and Sobolev inequality $\|u (t,\cdot)\|^{2}_{L^\infty}\leq\|u(t,\cdot)\|^{2}_{H^{1}}$, one may see that
every solution to the problem (\ref{Eq.(1.1)}) remains globally regular in time.

If $p>2a$ (or $q>2b$) and the slope of the function $v^{p}$ (or $u^q$)  is lower bounded or if $p<2a$ (or $q<2b$) and the
slope of the function $v^{p}$ (or $u^q$)  is upper bounded on $[0,T)\times\mathbb{R}$, then there exists a positive constant $M>0$ such
that
$$\frac{d}{dt}\int_\mathbb{R}m^2dx\leq M \int_\mathbb{R}m^2dx,\quad\frac{d}{dt}\int_\mathbb{R}n^2dx\leq M\int_\mathbb{R}n^2dx.$$
By means of the Gronwall's inequality, we have
$$||m(t,\cdot)||_{L^2}\leq ||m(0,\cdot)||_{L^2}\exp\{Mt\},\quad ||n(t,\cdot)||_{L^2}\leq ||n(0,\cdot)||_{L^2}\exp\{Mt\}\quad \forall t\in [0,T),$$
which implies that the solution does not blow up in a finite time.

On the other hand, by Theorem 3.1 and Sobolev's imbedding theorem, one may see that if
the slope of the functions $v^{p},u^{q}$  becomes unbounded either lower or upper in a finite time, then
the solution will blow up in a finite time. This completes the proof of Theorem 1.4.
\end{proof}

Next, let us consider the following initial value problem:
\begin{equation}
\left\{
\begin{array}{llll}
\phi_t= v^p(t,\phi(t,x)), &t\in[0,T),x\in\mathbb{R},\\
\varphi_t= u^q(t,\varphi(t,x)), &t\in[0,T),x\in\mathbb{R},\\
\phi(0,x)=x, \varphi(0,x)=x, &x\in\mathbb{R},
\end{array}
\right.
\end{equation}
where $u, v$ denote the solution to the problem (1.1).
Adopting classical results in the theory of ordinary differential equations leads to the following 
results on $p,q$, which are crucial for the blow-up scenarios.
\begin{lm}
Let $u_0,v_0\in H^s$ with $s> 5/2$, and $T>0$ be the life span of the solution to Eq.(1.1).
Then there exists a unique solution $\phi,\varphi\in \mathcal{C}^1([0,T),\mathbb{R})$ to Eq. (3.22).
Moreover, the map $\phi(t,\cdot),\varphi(t,\cdot)$ is an increasing diffeomorphism  over $\mathbb{R}$, where
$$\phi_x (t,x)=\exp\left\{\int_0^t (v^p)_\phi(s,\phi(s,x))ds\right\}>0, \varphi_x (t,x)=\exp\left\{\int_0^t (u^q)_\varphi(s,\varphi(s,x))ds\right\}>0, $$
for all $(t,x)\in[0,T)\times\mathbb{R}$.
\end{lm}
\begin{proof}
From Theorem 1.1, we have  $u,v\in\mathcal{ C}([0,T);H^s )\cap \mathcal{C}^1([0,T);H^{s-1} )$.
Thus, 
both functions $u(t,x),v(t,x)$ and   $u_x(t,x),v_x(t,x)$  are bounded,  Lipschitz in space and $\mathcal{C}^1$ in time.
As per the classical existence and uniqueness theorem of ordinary differential equations,  equation (3.22) has a unique solution $p,q\in \mathcal{C}^1([0,T),\mathbb{R})$.

Differentiating both sides of equation (3.22) respect to $x$ yields
\begin{equation*}
\left\{
\begin{array}{llll}
\frac{d}{dt}\phi_x= (v^p)_\phi(t,\phi(t,x))\phi_x, &t\in[0,T),x\in\mathbb{R},\\
\frac{d}{dt}\varphi_x= (u^q)_\varphi(t,\varphi(t,x))\phi_x, &t\in[0,T),x\in\mathbb{R},\\
\phi_x (0,x)=1,
\varphi_x (0,x)=1, &x\in\mathbb{R},
\end{array}
\right.
\end{equation*}
which implies 
$$\phi_x (t,x)=\exp\left\{\int_0^t (v^p)_\phi(s,\phi(s,x))ds\right\}>0, \varphi_x (t,x)=\exp\left\{\int_0^t (u^q)_\varphi(s,\varphi(s,x))ds\right\}>0. $$
For every $T'<T$, employing the Sobolev embedding theorem gives
$$\sup_{(\tau,x)\in[0,T)\times\mathbb{R}}|(v^p)_x(\tau,x)|<\infty,\sup_{(\tau,x)\in[0,T)\times\mathbb{R}}|(u^q)_x(\tau,x)|<\infty.$$
So, there exists two constants $K_1,K_2>0$ such that $\phi_x\geq e^{-K_1t},\varphi_x\geq e^{-K_2t}$ for $(\tau,x)\in[0,T)\times\mathbb{R}$, which concludes the proof of the lemma.
\end{proof}

\begin{lm}
Let $z_0=(u_0,v_0)\in  H^s \times H^{s} $ with $s>5/2$ and  $T>0$ be the maximal existence time of the corresponding solution $z=(u, v)$ to system (1.1).  Then, we have
\begin{equation}
\left\{
\begin{array}{llll}
m(t,\phi(t,x)) \phi_x ^\frac{a}{p}(t,x)=m_0(x),&\text{ for all } (t,x)\in[0,T)\times\mathbb{R},\\
n(t,\varphi(t,x)) \varphi_x ^\frac{a}{p}(t,x)=n_0(x) &\text{ for all } (t,x)\in[0,T)\times\mathbb{R}.
\end{array}
\right.
\end{equation}
Moreover, if there exist $M_1>0$  and $M_2>0$ such that
$\frac{a}{p}(v^p)_\phi(t,\phi)\geq -M_1$ and  $\frac{b}{q}(u^q)_\varphi(t,\varphi)\geq -M_2$  for all $ (t,x)\in[0,T)\times\mathbb{R}$, then
$$||m(t,\cdot)||_{L^\infty}=||m(t,\phi(t,\cdot))||_{L^\infty}\leq \exp\{2M_1T\}||m_0(\cdot)||_{L^\infty} \quad \text{ for all } t\in[0,T)$$ and
$$||n(t,\cdot)||_{L^\infty}=||n(t,\varphi(t,\cdot))||_{L^\infty}\leq \exp\{2M_2T\}||n_0(\cdot)||_{L^\infty}.$$
Furthermore, { if  $\int_\mathbb{R}|m_0(x)|^{p/a}dx$ (or $\int_\mathbb{R}|n_0(x)|^{q/b}dx$) converges with $a\neq0$ (or $b\neq0$)}, then
$$\int_\mathbb{R}|m(t,x)|^{p/a}dx=\int_\mathbb{R}|m_0(x)|^{p/a}dx \text{ for all } t\in[0,T).$$ (respectively, $\int_\mathbb{R}|n(t,x)|^{p/b}dx=\int_\mathbb{R}|n_0(x)|^{q/b}dx$
for all $t\in[0,T)$).
\end{lm}
\begin{proof}
Noticing $\frac{d\phi_x (t,x)}{dt}=\phi_{xt} =  (v^p)_\phi (t,\phi(t,x))\phi_x (t,x)$, differentiating the left-hand side of the first equation in (3.23) with respect to $t$, and using the first equation in (1.1), we obtain 
\begin{equation*}
\begin{split}
\frac{d}{dt}\{m(t,\phi (t,x))\phi_x ^{a/p}(t,x)\}&=\left[m_t(t,\phi )+m_\phi(t,\phi)\phi_t (t,x)\right]\phi_x ^{a/p}(t,x)+\frac{a}{p}m(t,\phi )\phi_x^{(a-p)/p}  (t,x)\phi_{xt} (t,x)\\
&=\left[m_t(t,\phi )+m_\phi(t,\phi )v^p(t,\phi )+\frac{a}{p}m(t,\phi ) (v^p)_\phi (t,\phi )\right]\phi_x ^{a/p}(t,x)\\
&=0.
\end{split}
\end{equation*}
In a similar way, we would arrive at
$$\frac{d}{dt} \{n(t,q (t,x))\varphi_x ^{b/q}(t,x)\}=0,$$
which means that $m(t,\phi (t,x)) \phi_x ^{a/p}(t,x)$ and $n(t,\varphi(t,x)) \varphi_x ^{b/q}(t,x)$ are independent on the time $t$.
By (3.22), we know $\phi_x (x,0)=1$. So, Eq. (3.23) holds.

By Lemma 3.3, 
Eq. (3.23), and $\phi_x (0,x)=1$ 
we have
\begin{equation*}
\begin{split}
||m(t,\cdot)||_{L^\infty}&=||m(t,\phi (t,\cdot))||_{L^\infty}=|| \phi_x ^{-a/p}m_0||_{L^\infty}\\
&=||\exp\left\{-\frac{a}{p}\int_0^t (v^p)_\phi(s,\phi (s,x))ds\right\}m_0(\cdot)||_{L^\infty}\\
&\leq \exp\{2M_1T\}||m_0(\cdot)||_{L^\infty} \quad\text{ for all } t\in[0,T),
\end{split}
\end{equation*}
and 
 \begin{equation*}
\begin{split}
\int_\mathbb{R}|m_0(x)|^{p/a}dx&=\int_\mathbb{R}|m(t,\phi (t,x))|^{p/a}\phi_x (t,x)dx=\int_\mathbb{R}|m(t,\phi (t,x))|^{p/a}d\phi(t,x)\\
&=\int_\mathbb{R}|m(t,x)|^{p/a}dx \text{ for all } t\in[0,T),
\end{split}
\end{equation*}
which guarantee the lemma is true.
\end{proof}
Let us now come to prove Theorems 1.5-1.6 using Lemma 3.4.
\begin{proof} [Proof of Theorem \ref{result5}]
Since $u_0 \in H^s \cap W^{2,\frac{p}{a}} $ for $s>5/2$, Lemma 3.4 tells us that 
$$\int_\mathbb{R}|m(t,x)|^\frac{p}{a}dx\leq\int_\mathbb{R}|m_0(x)|^\frac{p}{a}dx \leq ||u_0||_{W^{2,\frac{p}{a}} }\quad \text{ if }0<a\leq p,$$
and
$$||m(t,x)||_{L^\infty }\leq ||m(0,x)||_{L^\infty }\quad \text{ if } a=0,$$
which imply $m=(1-\partial_x^2) u\in L^\frac{p}{a} $ and therefore 
$u\in W^{{2,\frac{p}{a}}} $. By the Sobolev imbedding theorem, we have $W^{{2,\frac{p}{a}}} \subset \mathcal{C}^1 $ for $0\leq a\leq p$.
Thus, the solution of the problem (\ref{Eq.(1.1)}) remains smooth for all time, which says that Theorem 1.5 is true.
\end{proof}

\begin{proof}[Proof of Theorem \ref{result6}]
Since $(u_0,v_0)\in H^s\times H^s$ ($s>5/2$), $m_0=(1-\partial_x^2) u_0$ has a compact
support. Without loss of generality, let us assume that $m_0$ is supported in the compact interval $[a,b]$. By Lemma 3.3, we have $\phi_x (x,t)>0$
on $\mathbb{R}\times[0,T)$. Therefore, by Lemma 3.4, we conclude that the $\mathcal{C}^1$ function $m(x,t)$ has its support in the
compact interval $[\phi (a,t),\phi (b,t)]$ for any $t\in[0,T)$, which completes the proof of Theorem 1.6.
\end{proof}

\section{Well-posed in the sense of Hadamard}
In this section, we shall prove that the Cauchy problem for Eq.(1.1) with the initial data $z_0=(u_0,v_0)\in H^s\times H^s$ ($s>5/2$) is not only well-posed in the sense of Hadamard, but also satisfies the estimate (1.13)  on the line and on the circle, namely, prove Theorem 1.7 is true.  Let us start from 
the following lemmas.
\begin{lm}(See\cite{HM33,K}.)
If $r>0$, then $H^r\cap L^\infty$
is an algebra. Moreover, we have

(i) $||fg||_{H^r }\leq c(||f||_{L^\infty }||g||_{H^r }+||g||_{L^\infty }||f||_{H^r }), \text{ for } r>0.$

(ii) $||fg||_{H^{r } }\leq c ||f||_{H^{r+1} }||g||_{H^{r }}, \text{ for }  r>-1/2.$

(iii) $||fg||_{H^{r-1 } }\leq c ||f||_{H^{s-1 } }||g||_{H^{r-1 }}, \text{ for } 0\leq r\leq 1, s>3/2, r+s\geq 2,$

where $c$ is a constant depending only on $r,s$.
\end{lm}
\begin{lm}(See \cite{K,Taylor2}.) If $ [ D ^r,f]g = D ^r(fg)-f D ^r g$ where $ D =(1-\partial_x^2)^\frac{1}{2}$, then we have

(i) $||[D ^r,f]g||_{L^2 }\leq c (||\partial_x f||_{L^\infty }||D ^{r-1}g||_{L^2 }+||D ^r f||_{L^2 }||g||_{L^\infty }),r>0.$

(ii) $||[D ^r\partial_x,f]g||_{L^2 }\leq c ||f||_{H^{s-1} }|| g||_{H^r } ,r+1\geq0, s-1>3/2,r+1\leq s-1.$\\
where  $c$ is a constant depending only on $r$.
\end{lm}
  The proof of Theorem 1.7 consists of the following several steps.
\subsection{Priori estimates  for $u_\epsilon$ and $v_\epsilon$}
Applying $J_\epsilon$ to system (1.1) leads to the following system
\begin{equation}
\left\{
\begin{array}{llll}
&\partial_t m_\epsilon+(J_\epsilon v)^{p}\partial_{x}J_\epsilon m+\frac{a}{p}\partial_x(J_\epsilon v)^p J_\epsilon m=0, \\
&\partial_t n_\epsilon+(J_\epsilon u)^{q}\partial_{x}J_\epsilon n+\frac{b}{q}\partial_x(J_\epsilon u)^q J_\epsilon n=0, \\
&m=u-u_{xx},n=v-v_{xx},u(x,0)=u_0(x), v(x,0)=v_0(x)
\end{array}
\right.
\end{equation}
where  the operator $J_\varsigma$ is called the Friedrichs mollifier defined by
$$J_\varsigma f(x)=J_\varsigma(f)(x)=j_\varsigma*f,  \ \ \forall \varsigma\in(0,1],$$
$j_\varsigma(x)=\frac{1}{\varsigma}j(\frac{x}{\varsigma})$, and $j(x)$ is a $\mathcal{C}^\infty$ function supported in the interval $[-1,1]$ such that $j(x)\geq0, \int_\mathbb{R}j(x)dx=1$.
 Multiplying both sides of the first equation in (4.1) by $D^{s-2}J_\epsilon mD^{s-2} $ and integrating with respect to $x\in \mathbb{R}$, we get
\begin{align}
&\frac{1}{2}\frac{d}{dt}|| m_\epsilon ||_{H^{s-2}}^2 =   \int_\mathbb{R}D^{s-2}J_\epsilon m D^{s-2}J_\epsilon ((J_\epsilon v)^{p}\partial_{x}J_\epsilon m)dx+ \frac{a}{p}\int_\mathbb{R}D^{s-2}J_\epsilon m D^{s-2}J_\epsilon\left(\partial_x(J_\epsilon v)^p J_\epsilon m \right)dx.
\end{align}
We need to estimate the right-hand side of (4.2). Apparently, both $D^{s}$ and $J_\epsilon$ are commutative and $J_\epsilon$ satisfies
\setcounter{equation}{4}
\begin{equation}
(J_\epsilon f,g)_0=(f,J_\epsilon g)_0,||J_\epsilon u||_{H^{s}}\leq ||u||_{H^{s}}.
\end{equation}
Lemma 4.2(ii) reveals
\begin{equation}
\begin{split}
&\left|\int_\mathbb{R}D^{s-2}J_\epsilon m D^{s-2}J_\epsilon ((J_\epsilon v)^{p}\partial_{x}J_\epsilon m)dx \right|
\lesssim\left|\int_\mathbb{R}D^{s-2}J_\epsilon m D^{s-2}J_\epsilon [\partial_{x}((J_\epsilon v)^{p}J_\epsilon m)-J_\epsilon m \partial_{x} (J_\epsilon v)^{p}]dx \right|\\
&\lesssim\left| \int_\mathbb{R}[D^{s-2}\partial_{x},(J_\epsilon v)^{p} ]J_\epsilon m D^{s-2}J_\epsilon mdx\right|+\left|\int_\mathbb{R}(J_\epsilon v)^{p}D^{s-2} \partial_xJ_\epsilon m D ^{s-2} J_\epsilon mdx \right|+||m_\epsilon||^2_{H^{s-2}}||v_\epsilon||^p_{H^{s-1}}\\
&\lesssim\left|\int_\mathbb{R}[D^{s-2}\partial_{x},(J_\epsilon v)^{p} ]J_\epsilon m D^{s-2}J_\epsilon mdx \right|+\frac{1}{2}\left|\int_\mathbb{R}(J_\epsilon v)^{p}   \partial_x(D ^{s-2} J_\epsilon m)^2dx \right|+||m_\epsilon||^2_{H^{s-2}}||v_\epsilon||^p_{H^{s-1}}\\
&\lesssim\left\|[D^{s-2}\partial_{x},(J_\epsilon v)^{p} ]J_\epsilon m   \right\|_{L^2}||m_\epsilon||_{H^{s-2}}+\frac{1}{2}\left|\int_\mathbb{R} \partial_x(J_\epsilon v)^{p}  (D ^{s-2} J_\epsilon m)^2dx \right|+||m_\epsilon||^2_{H^{s-2}}||v_\epsilon||^p_{H^{s-1}}\\
&\lesssim(||\partial_x (J_\epsilon v)^{p} ||_{L^\infty }||D^{s-3}\partial_xJ_\epsilon m  ||_{L^2 }+||D^{s-2} (J_\epsilon v)^{p} ||_{L^2 }||J_\epsilon m ||_{L^\infty })||m_\epsilon||_{H^{s-2}}+  ||m_\epsilon||^2_{H^{s-2}}||v_\epsilon||^p_{H^{s-1}}\\
&\lesssim ||m_\epsilon||^2_{H^{s-2}}||v_\epsilon||^p_{H^{s-1}}.
\end{split}
\end{equation}
Employing Lemma 4.1(ii) and (4.5) yields
\begin{equation}
\begin{split}
\left|\int_\mathbb{R}D^{s-2}J_\epsilon m D^{s-2}J_\epsilon\left(\partial_x(J_\epsilon v)^p J_\epsilon m \right)dx\right|
&\lesssim  ||v_\epsilon||^p_{s-1}||m_\epsilon||^2_{H^{s-2}}.
\end{split}
\end{equation}
For all $s\in \mathbb{R}$, we have
\begin{equation}
\|u \|_{H^s} = \|m \|_{H^{s-2} }\text{ and }\|v \|_{H^{s} } = \|n \|_{H^{s-2} }.
\end{equation}
Therefore,
\begin{equation*}
\begin{split}
 \frac{d}{dt} || u_\epsilon(t)||_{H^{s} }  &\lesssim C_s \left(|| u_\epsilon|| _{H^{s }} +|| v_\epsilon||_{H^{s }}\right)^{p+1}.
\end{split}
\end{equation*}
Adopting a similar procedure for $v_\epsilon$ produces 
\begin{equation}
\begin{split}
  \frac{d}{dt}  || z_\epsilon|| _{H^{s }}  \doteq\frac{d}{dt} \left(|| u_\epsilon|| _{H^{s }} +|| v_\epsilon||_{H^{s }}\right)\lesssim  2C_s\left(|| u_\epsilon|| _{H^{s }} +|| v_\epsilon||_{H^{s }}\right)^{\kappa+1}=2C_s|| z_\epsilon||_{H^{s }} ^{\kappa+1},
\end{split}
\end{equation}
where $\kappa=\max\{p,q\}$.

Solving the differential inequality (4.9) generates
\begin{equation}
\begin{split}
 ||  z_\epsilon (t)||_{H^{s } } \leq \frac{ ||  z_0||_{  H^{s-1 }}}{  \sqrt[\kappa]{1-  2\kappa C_s||  z_0||_{H^{s  } }^{\kappa} t}  }.
\end{split}
\end{equation}
Let $T_0=\frac{2^\kappa-1}{2^{\kappa+1}\kappa C_s||  z_0||_{H^{s  } }^{\kappa} }$, then from Eq. (4.10) we see that there exist the solutions $u,v$ for $0\leq t\leq T_0$ with the following bound
\begin{equation}
\begin{split}
 ||  z(t)||_{H^s}  &\leq 2||  z_0||_{H^s}, \text{  for } 0\leq t\leq T_0.
\end{split}
\end{equation}
Moreover, by Eq.(4.1) we may obtain the following estimates for $\partial_t   u_\epsilon(t) $ and $\partial_t   v_\epsilon(t) $:
\begin{equation}
\begin{split}
  ||\partial_t   u_\epsilon(t) ||_{H^{s-1} }&\cong||\partial_t   m_\epsilon(t) ||_{H^{s-3} }\lesssim  || (J_\epsilon v)^{p}\partial_{x}J_\epsilon m ||_{H^{s-3} }+||\partial_x(J_\epsilon v)^p J_\epsilon m ||_{H^{s-3} }\lesssim ||z_0||^{p+1}_{H^{s} },\\
    ||\partial_t   v_\epsilon(t) ||_{H^{s-1} }&\cong||\partial_t   n_\epsilon(t) ||_{H^{s-3} } \lesssim ||z_0||^{q+1}_{H^{s} }.
\end{split}
\end{equation}

\subsection{Existence of solutions on the line}
\begin{thm}
 There exists a solution $z=(u,v)$ to the Cauchy problem (4.1) in the space
$\mathcal{C}([0,T];H^{s  }\times H^s) $ with $s  >5/2$. Furthermore, the  $H^{s  }\times H^{s  }$ norm of $z$ satisfies Eqs.(4.11) and (4.12).
\end{thm}

So far, we have studied the existence of a unique solution $z_\epsilon \in  \mathcal{C}([0,T];H^{s }\times H^s) ,s>5/2$ to the initial value problem (4.1) with life span $T =\frac{2^\kappa-1}{2^{\kappa+1}\kappa C_s||  z_0||_{H^{s  } \times H^s}^{\kappa} }$ as well as 
the size estimates (4.11) and (4.12).
Next, we need to show that $z_\epsilon\rightarrow z\doteq (u,v ) \in  \mathcal{C}([0,T];H^{s }\times H^s)  $ where $z$ is the solution to Eq. (1.1).
Our proof is carried out 
through refining the convergence of the family $\{z_\epsilon\}=\{ u_\epsilon,v_\epsilon  \}$ several
times by extracting its subsequences. 
After each extraction, for our convenient discussion and simplicity, 
the resulting subsequence is still labeled as $\{z_\epsilon\}=\{ u_\epsilon,v_\epsilon \}$.

 \textbf{Weak convergence in $L^\infty(I;H^{s }\times H^s) $}.
 The set of functions $\{z_\epsilon\}_{\epsilon\in(0,1]}$ is bounded in the
space $\mathcal{C}(I;H^{s }\times H^s)  \subset L^\infty (I;H^{s }\times H^s) $. By the inequality (4.11), we have
\begin{align}
||z_\epsilon||_{L^\infty (I;H^{s }\times H^s) }=\sup_{t\in I} ||z_\epsilon||_{ H^{s } \times H^s }\leq 4 ||z_0||_{ H^{s }\times H^{s } },\nonumber\ \\
\Rightarrow\{z_\epsilon\}_{\epsilon\in(0,1]}\subset\overline{B}(0,2||z_0||_{H^{s }\times H^{s }})\subset L^\infty (I;H^{s }\times H^s) ,
\end{align}
Alaoglu's theorem  tells us that $\{z_\epsilon\}$ is pre-compact in $\overline{B}(0,2||z_0||_{H^{s }\times H^{s }})\subset L^\infty (I;H^{s }\times H^s) $
with respect to the weak topology. Therefore, we may extract a subsequence $\{z_{\epsilon'}\}$ that converges to an element $z\in\overline{B}(0,2||z_0||_{H^{s }\times H^s})$ { in a weak sense}. Provided we have 
such a construction, then $z$ would be our desired function satisfying the solution estimate (1.13).

\textbf{Convergence in $\mathcal{C}^\sigma(I;H^{{s }-\sigma}\times H^{s-\sigma})$ with $\sigma\in(0,1)$.}
We shall apply the Ascoli's theorem to prove convergence in the spaces. 
Let $\phi\in \mathcal{S}(\mathbb{R})$ and  $\phi(x)>0$, $\forall x\in\mathbb{R}$. We shall prove that there exists a subsequence of $\{\phi z_\epsilon(t)\}$ converging to a function {  $\phi z (t) \in \mathcal{C}(I;H^{{s }-\sigma}\times H^{s-\sigma})$}. Define $z(t)\doteq \{\phi z_\epsilon(t)\}_{\epsilon\in(0,1]}$. Then we know that for each $t\in I$ the set $z(t)\subset H^{s } \times H^{s }$ is pre-compact in $H^{{s }-\sigma}\times H^{s-\sigma} $ as a consequence of the Rellich's theorem. {Thus, the first condition of Ascoli's theorem is met. the remaining is to show that the second condition is also met, i.e. $\phi z_\epsilon(t)$ is equi-continuous.}
To see this,
we shall first show $z_\epsilon\in \mathcal{C}^\sigma(I,H^{{s}-\sigma}\times H^{s-\sigma}) $ for $\epsilon, \sigma\in(0,1)$., 
Then, we prove that the $\mathcal{C}^\sigma(I,H^{{s}-\sigma}\times H^{s-\sigma})$ norm of $z_\epsilon$ satisfies 
\begin{equation}
\begin{split}
 || z_\epsilon(t) ||_{\mathcal{C}^\sigma(I,H^{{s }-\sigma}\times H^{s-\sigma}) } \lesssim ||  z_0||_{H^s\times H^s}+ ||  z_0||_{H^s\times H^s}^{\kappa+1}.
\end{split}
\end{equation}
Let us begin with the following norm definition 
\begin{equation}
\begin{split}
 || z_\epsilon(t) ||_{\mathcal{C}^\sigma(I,H^{{s }-\sigma}\times H^{s-\sigma}) }\doteq \sup_{t\in I}  ||  z_\epsilon(t)||_{H^{{s }-\sigma}\times H^{s-\sigma}}+\sup_{t_1\neq t_2}\frac{||z_\epsilon(t_1)-z_\epsilon(t_2)||_{H^{{s }-\sigma} \times H^{s-\sigma}}}{|t_1-t_2|^\sigma}.
\end{split}
\end{equation}
Because the first term on the right hand side of (4.15) is bounded by Eq. (4.11), we have 
\begin{equation*}
\begin{split}
\sup_{t\in I} ||  z_\epsilon(t)||_{H^{{s }-\sigma}\times H^{s-\sigma} }\leq \sup_{t\in I} || z_\epsilon(t)||_{H^{{s } } \times H^{s }}\leq 2||  z_0||_{H^{s }\times H^{s}}.
\end{split}
\end{equation*}
The second term on the right hand side of (4.15), due to the inequality $x^\sigma\leq 1+x$, generates the following estimate:
\begin{equation*}
\begin{split}
 \sup_{t_1\neq t_2}&\frac{||z_\epsilon(t_1)-z_\epsilon(t_2)||_{H^{{s }-\sigma} \times H^{s-\sigma}}}{|t_1-t_2|^\sigma}= \sup_{t_1\neq t_2}\left(\int_\mathbb{R}(1+\xi^2)^{s }\frac{ |\hat{z}_\epsilon(\xi,t_1)-\hat{z}_\epsilon(\xi,t_2)|^2}{(1+\xi^2)^\sigma|t_1-t_2|^{2\sigma}}d\xi\right)^{1/2}\\
 &\quad\quad\leq \sup_{t_1\neq t_2}\left(\int_\mathbb{R}(1+\xi^2)^{s }\left[1+\frac{1}{(1+\xi^2)|t_1-t_2|^{2}} \right] |\hat{z}_\epsilon(\xi,t_1)-\hat{z}_\epsilon(\xi,t_2)|^2 d\xi\right)^{1/2}\\
  &\quad\quad\leq \sup_{t\in I}(2||z_\epsilon(t ) ||_{H^{{s }}\times H^{s } }+||\partial_t z_\epsilon(t ) ||_{H^{s-1} \times H^{s-1}})\\
  &\quad\quad\lesssim ||z_0 ||_{H^{s}\times H^{s}}+||  z_0 ||_{H^{s }\times H^{s }} ^{\kappa+1},
\end{split}
\end{equation*}
where (4.11) and (4.12) are applied in the last inequality. Combining these bounds together leads to the desired estimate (4.14).
On the other hands, the inequality (4.14) apparently implies 
\begin{equation}
\begin{split}
 || z_\epsilon(t_1)- z_\epsilon(t_2)||_ {H^{{s }-\sigma} \times H^{s-\sigma} }&\leq || z_\epsilon(t) ||_{\mathcal{C}^\sigma(I,H^{{s }-\sigma}\times H^{s-\sigma}) }|t_1-t_2|^\sigma\\
& \lesssim \left(||  z_0||_{H^s\times H^s}+ ||  z_0||_{H^s\times H^s}^{\kappa+1}\right)|t_1-t_2|^\sigma,\quad t_1,t_2\in I, \ \  \forall\epsilon, \sigma \in(0,1).
\end{split}
\end{equation}
Furthermore, we have the following equicontinuity property for $\{\phi z_\epsilon(t)\}_{\epsilon\in(0,1]}$
\begin{equation}
\begin{split}
 || \phi z_\epsilon(t_1)- \phi z_\epsilon(t_2)||_ {H^{{s }-\sigma}\times H^{s-\sigma}   }&\leq ||\phi||_{H^{{s }-\sigma}}  || z_\epsilon(t_1)- z_\epsilon(t_2)||_ {H^{{s }-\sigma}\times H^{s-\sigma} }\\
 & \lesssim \left(||  z_0||_{H^s\times H^s}+ ||  z_0||_{H^s\times H^s}^{\kappa+1}\right)|t_1-t_2|^\sigma.
\end{split}
\end{equation}
The Ascoli's theorem admits
\begin{equation}
\begin{split}
 || \phi z_\epsilon - \phi z ||_  {\mathcal{C}^\sigma(I,H^{{s }-\sigma}\times H^{s-\sigma} ) }=\sup_{t\in I} || \phi z_\epsilon(t) - \phi z(t) ||_  { H^{{s }-\sigma} \times H^{s-\sigma} }\rightarrow0 \text{ for } \epsilon\rightarrow0.
\end{split}
\end{equation}

\textbf{Convergence in $\mathcal{C}(I;\mathcal{C}^1(\mathbb{R})\times \mathcal{C}^1(\mathbb{R})) $.}
Let $\sigma$ satisfy ${s-2}-\sigma>3/2$. Then, applying the Sobolev lemma and (4.18) reveals
\begin{equation}
\begin{split}
 || \phi z_\epsilon - \phi z  ||_{\mathcal{C} (I,\mathcal{C}^1(\mathbb{R})\times \mathcal{C}^1(\mathbb{R}) ) }&=\sup_{t\in I} || \phi z_\epsilon(t) - \phi z(t) ||_{\mathcal{C}^1(\mathbb{R})\times \mathcal{C}^1(\mathbb{R})   }\\
 &\lesssim || \phi z_\epsilon(t) - \phi z(t) ||_  { H^{{s }-\sigma} \times H^{{s }-\sigma}  }\rightarrow0 \text{ for } \epsilon\rightarrow0.
\end{split}
\end{equation}
Therefore, convergence in $\mathcal{C}(I;\mathcal{C}^1(\mathbb{R})\times \mathcal{C}^1(\mathbb{R})) $ has been established. Next, we prove that $z$ solves Eq.(1.1).

\textbf{Verification of $z$ being a solution to the Cauchy problem (1.1).}
Let us first recall 
the generalization of Sobolev spaces 
from the real analysis. Suppose that a system of functions $f_n:I\rightarrow\mathbb{R}$ are continuous. If there is some $t_0\in I$ such that $f_n(t_0)\rightarrow f(t_0)$ as $n\rightarrow \infty$ and $f_n'$ is uniformly convergent to $f'(t)$ on $I$, then $f_n$
also uniformly converges  to $f$ on $I$ and $f'(t)=\lim_{n\rightarrow\infty}f_n'(t)$.
We shall apply this result to the sequence $\{z_\epsilon(t)\}_{\epsilon\in(0,1]}$.
we already show $\{z_\epsilon(t)\}_{\epsilon\in(0,1]}$ is convergent to $z(t)$ in $\mathcal{C}(I;\mathcal{C}^1(\mathbb{R})\times \mathcal{C}^1(\mathbb{R})) $, which implies  the first condition of the theorem is satisfied.
From the convergence $\phi z_\epsilon \rightarrow \phi z$ in $\mathcal{C}(I;\mathcal{C}^1(\mathbb{R})\times \mathcal{C}^1(\mathbb{R}))$, we know  that $z_\epsilon\rightarrow z$ and $\partial_x z_\epsilon\rightarrow \partial_xz$ are pointwisely convergent.
This reveals that the right hand side of (4.1) converges to the corresponding terms without $\epsilon$.  The remaining task is to show that $\partial_t(\phi z_\epsilon)=\left(\partial_t(\phi u_\epsilon),\partial_t(\phi v_\epsilon)\right)$ in $\mathcal{C}(I;\mathcal{C}^1(\mathbb{R})\times \mathcal{C}^1(\mathbb{R}))$ is uniformly convergent.

Due to $m$ and $n$ are in a parallelled situation,   let us only consider the component $m$ while the other component can be treated in the same way.
Starting from 
the first equation of (4.1) and multiplying both sides of the equation by $\phi$, we have
\begin{equation}
\partial_t   (\phi m_\epsilon)  =-\phi J_\epsilon ((J_\epsilon v)^{p}\partial_{x}J_\epsilon m)-\frac{a}{p}\phi J_\epsilon (\partial_x(J_\epsilon v)^p J_\epsilon m).
\end{equation}
Casting $\sigma$ into $s-\sigma-3>-1/2$ with $s-\sigma-3\neq1$ allows us to derive the equicontinuity.  Regarding the first term on the right-hand-side of (4.20), Lemma 4.1(ii) tells us 
\begin{equation*}
\begin{split}
||\phi J_\epsilon [v_\epsilon^p (t_1)&\partial_xm_\epsilon(t_1)]-\phi J_\epsilon [v_\epsilon^p (t_2)\partial_xm_\epsilon(t_2)]||_{H^{s-\sigma-3}}\\
&\lesssim
|| v_\epsilon^p (t_1) ||_{H^{s-\sigma-2}}
||  m_\epsilon(t_1) - m_\epsilon(t_2) ||_{H^{s-\sigma-2}}+|| m_\epsilon (t_1) ||_{H^{s-\sigma-2}}
||  v_\epsilon^p(t_1) - v_\epsilon(t_2)^p ||_{H^{s-\sigma-2}}\\
&\lesssim
|| v_\epsilon^p (t_1) ||_{H^{s-\sigma-2}}
||  m_\epsilon(t_1) - m_\epsilon(t_2) ||_{H^{s-\sigma-2}}\\
&\quad+|| m_\epsilon (t_1) ||_{H^{s-\sigma-2}}
||  v_\epsilon (t_1) - v_\epsilon(t_2)  ||_{H^{s-\sigma-2}}\sum_{j=0}^{p-1}||  v_\epsilon (t_1) ||_{H^{s-\sigma-2}}^{p-1-j}|| v_\epsilon(t_2)  ||_{H^{s-\sigma-2}}^j,
\end{split}
\end{equation*}
which yields 
\begin{equation*}
\begin{split}
 ||\phi J_\epsilon [v^p_\epsilon (t_1)&\partial_xm_\epsilon(t_1)]-\phi J_\epsilon [v^p_\epsilon (t_2)\partial_xm_\epsilon(t_2)]||_{H^{s-\sigma-3}}
 \lesssim \left(||  z_0||_{H^s\times H^s}^{p+1}+ ||  z_0||_{H^s\times H^s}^{p+\kappa+1}\right)|t_1-t_2|^\sigma .
\end{split}
\end{equation*}
The second term on the right hand side of Eq. (2.20) can similarly be estimated
\begin{equation*}
\begin{split}
 ||\phi J_\epsilon [\partial_xv^p_\epsilon (t_1) m_\epsilon (t_1)-\partial_xv^p_\epsilon (t_2) m_\epsilon (t_2)]||_{H^{s-\sigma-3}}
 \lesssim \left(||  z_0||_{H^s\times H^s}^{p+1}+ ||  z_0||_{H^s\times H^s}^{p+\kappa+1}\right)|t_1-t_2|^\sigma .
\end{split}
\end{equation*}
Therefore, by Ascoli's theorem, we conclude that a subsequence { of $\{\phi z_\epsilon(t)\}_{\epsilon\in(0,1]}$} satisfies
\begin{equation}
\partial_t(\phi Z_\epsilon)=
\left\{
\begin{array}{llll}
&\partial_t(\phi m_\epsilon)\rightarrow -\phi v^{p}m_{x} -\frac{a}{p}\phi v^{p }_x m ,\\
&\partial_t( \phi n_\epsilon)\rightarrow -\phi u^{q}n_{x} -\frac{b}{q} \phi u^{q }_x n ,
\end{array}
\right.
\text{ in } \mathcal{C}(I;H^{s-\sigma-3}).
\end{equation}
As per $s-\sigma-1>1/2$ and the Sobolev lemma, we find that $\mathcal{C}(I,\mathcal{C}(\mathbb{R})\times \mathcal{C}(\mathbb{R}))\hookrightarrow \mathcal{C}(I;H^{s-\sigma-1}\times H^{s-\sigma-1})$.

\textbf{Claim: $z\in L^\infty(I;H^{s }\times H^{s })\cap Lip(I;H^{s-1}\times H^{s-1 })$ is a solution to (1.1).}
First, we notice
that $\phi z_\epsilon\rightarrow \phi z$ and $\partial_t(\phi z_\epsilon)\rightarrow \partial_t(\phi z)$ in the space $\mathcal{C}(I,\mathcal{C}(\mathbb{R})\times \mathcal{C}(\mathbb{R}))$, which 
imply that $t \mapsto \phi z(t)$ is a differentiable map.
As we chose { $\phi $  with no zeros}, the formula 
$\partial_t(\phi z_\epsilon)=\phi \partial_t( z_\epsilon)$ allows us to get rid of $\phi$ and $\epsilon$ in (4.21). Thus, we are able to locate 
$z\in Lip(I;H^{s }\times H^{s })$ with the following property
\begin{equation}
||z(t_1)-z(t_2)||_{H^{s }\times H^{s }}\leq \sup_{t\in I} ||\partial_t z(t)||_{H^{s }\times H^{s }}|t_1-t_2|\lesssim ||z_0||_{H^s\times H^s}^{\kappa+1}|t_1-t_2|.
\end{equation}

\textbf{Regularity improvement of $z$ up to $\mathcal{C}(I;H^{s }\times H^{s })$}.
We already  know  $z\in L^\infty(I;H^{s }\times H^{s })\cap Lip(I;H^{s-1 }\times H^{s-1 })$. Let us now prove
$z\in \mathcal{C}(I;H^{s }\times H^{s })$, namely, if $t_n\in I$ converge to $t\in I$ as $n$ goes to infinity, then $\lim_{n\rightarrow\infty}||z(t_n)-z(t)||_{H^{s }\times H^{s }}=0$.
According to the norm definition in
$H^{s }\times H^{s }$, this is equivalent to showing 
\begin{equation}
\lim_{n\rightarrow\infty}\left(||z(t_n)||_{H^{s }\times H^{s }}^2-\langle z(t_n),z(t)\rangle_{H^{s }\times H^{s }}+||z(t)||_{H^{s }\times H^{s }}^2\right)=0.
\end{equation}

\begin{lm} The solution $z\in L^\infty(I;H^{s }\times H^{s })\cap Lip(I;H^{s-1 }\times H^{s -1})$ is continuous on $I$ in the sense of 
weak topology in $H^{s}\times H^{s }$, i.e.
$$\langle z(t_n)-z(t),\varphi\rangle_{H^{s }\times H^s}=0, \text{ for any } \varphi\in H^{s }\times H^s.$$
\end{lm}
\begin{proof}
Let $\varphi \in H^{s }\times H^s$. For any $\epsilon>0$, choose $\psi\in \mathcal{S}(\mathbb{R})$ such that 
$||\varphi-\psi||_{H^{s }\times H^s}<\epsilon/(4||z_0||_{H^s\times H^s})$. The triangle inequality yields 
\begin{equation*}
\begin{split}
|\langle z(t_n)-z(t),\varphi\rangle_{H^{s }\times H^s}|&\leq  |\langle z(t_n)-z(t),\varphi-\psi\rangle_{H^{s }\times H^s}| +|\langle z(t_n)-z(t),\psi\rangle_{H^{s }\times H^s}|\\
 &\leq  || z(t_n)-z(t)||_{H^{s-2}}||\varphi-\psi||_{H^{s }\times H^{s }}  +|| z(t_n)-z(t)||_{H^{s }\times H^s}||\psi||_{H^{s }\times H^s} \\
  &\leq  \epsilon/2+||z_0||^{\kappa+1}_{H^s\times H^s}||\psi||_{H^{s }} |t_n-t|,
\end{split}
\end{equation*}
  where the inequalities   (4.21) and (4.11) are applied.
  Since $||\psi||_{H^{s}}$ is bounded, we select $N$ such that  for any $n>N$, $|t_n-t|<\epsilon/(2||z_0||^{\kappa+1}_{H^s\times H^s}||\psi||_{H^{s }})$.
  Hence, we have 
\begin{equation*}
\begin{split}
|\langle z(t_n)-z(t),\varphi\rangle_{H^{s }\times H^s}|<\epsilon, \forall n>N,
\end{split}
\end{equation*}
which concludes the proof.
\end{proof}
Employing Lemma 4.3 reduces 
(4.23) 
to $\lim_{n\rightarrow\infty} ||z(t_n)||_{H^{s }\times H^s}= ||z(t )||_{H^{s }\times H^s}$, i.e. we prove the map
$t\mapsto ||z(t)||_{H^{s }\times H^s}$ is continuous.    We already know that $||J_\epsilon z(t)||_{H^{s }\times H^s}$
converges to $||  z(t)||_{H^{s }\times H^s}$ pointwise
in $t$ as $\epsilon\rightarrow0$. Thus, it suffices to show that
each $||J_\epsilon z(t)||_{H^{s }\times H^s}$ is Lipschitz with  the bounded Lipschitz constants for the whole family of functions.
Taking a similar work procedure of $||z_\epsilon||_{H^{s }\times H^s}$ in section 4.1 we arrive at
$\frac{d}{dt}||J_\epsilon z(t)||_{H^{s }\times H^s}\lesssim ||z_0||_{H^s\times H^s}^{\kappa+1}$.
Therefore, we conclude that $||z(t)||_{H^{s }\times H^s}$ is Lipschitz and the solution $z$ is in $\mathcal{C}(I;H^{s }\times H^s)$.

\subsection{Uniqueness of solution on a line}
We have already shown that there exists a solution $Z=(m,n)$ to the Cauchy problem (1.1) in $\mathcal{C}(I;H^{s-2})$, which satisfies the estimates (4.11) and (4.12) with lifespan $T =\frac{2^\kappa-1}{2^{\kappa+1}\kappa C_s||  z_0||_{H^{s  } }^{\kappa} }$. In this section we shall prove the solution is unique.
\begin{thm}
For the initial data $z_0\in H^s\times H^s$ with $s>5/2$, the Cauchy problem (1.1) has a unique solution $z=(u,v)$ in the space $\mathcal{C}(I;H^{s })$.
\end{thm}
\begin{proof}
Let $z_0\in H^s\times H^s$ and $s>5/2$. Suppose $Z_1\doteq (u_1,v_1,m_1,n_1)$ and $Z_2\doteq (u_2,v_2,m_2,n_2)$ are two solutions to the Cauchy problem (1.1) with $u_1(x,0)=u_0(x)=u_2(x,0)$ and $v_1(x,0)=v_0(x)=v_2(x,0)$.    Let $M=m_1-m_2,N=n_1-n_2,U=u_1-u_2,V=V_1-V_2$, then we have
\begin{equation*}
\left\{
\begin{array}{llll}
&\partial_{t}M+ \partial_{x}(v_{1} ^{p} M)=-\partial_{x}[ (v_{1} ^{p}- v_{2} ^{p})m_{2}] +\frac{p-a}{p}\partial_x v_{1} ^{p} M +\frac{p-a}{p}\partial_x ( v_{1} ^{p}-  v_{2} ^{p}) m_{2} ,\\
&\partial_{t}N+\partial_{x} (u_{1}^{q}N) =-\partial_{x}[(u_{1}^{q}-u_{2}^{q})n_{2}]  +\frac{q-b}{q}\partial_xu_{1}^{q} N  +\frac{q-b}{q}\partial_x(u_{1}^{q}-u_{2}^{q}) n_{2} ,\\
&M|_{t=0}= N|_{t=0}=0.
\end{array}
\right.
\end{equation*}
 Without losing generality,
 let us just consider the third equation while other two can be handled in a similar way.
 Assume $\sigma\in (1/2,s-3)$, then applying $D^\sigma$ to 
 the first equation and multiplying the result by $D^\sigma E$, we obtain 
\begin{equation}
\begin{split}
\frac{1}{2}\frac{d}{dt}|| M ||_{H^\sigma} ^2 = & - \int_\mathbb{R}D^\sigma M D^\sigma\partial_{x}(v_{1} ^{p} M)dx\\
&-\int_\mathbb{R}D^\sigma M  D^{\sigma } \left(-\partial_{x}[ (v_{1} ^{p}- v_{2} ^{p})m_{2}] +\frac{p-a}{p}\partial_x v_{1} ^{p} M +\frac{p-a}{p}\partial_x ( v_{1} ^{p}-  v_{2} ^{p}) m_{2} \right)dx.
\end{split}
\end{equation}
Apparently, one can verify the following inequalities
\begin{equation}
\begin{split} \left| \int_\mathbb{R}D^\sigma M D^\sigma \partial_x(v_1^pM  )dx\right|
&\lesssim  \left| \int_\mathbb{R} D^\sigma M\left( [D^\sigma \partial_x, v_1^p]M+ v_1^p \partial_xD^\sigma M\right)dx\right| \\
&\lesssim  || [D^\sigma \partial_x, v_1^p]M ||_{L^2}|| M ||_{H^\sigma}+ \frac{1}{2}\left| \int_\mathbb{R} \partial_xv_1^p(D^\sigma M  )^2dx\right| \\
&\lesssim || v_1^p||_{H^\rho}|| M ||_{H^\sigma}^2+||(v_1^p)_x ||_{L^\infty} || M ||_{H^\sigma}^2\\
&\lesssim || v_1^p||_{H^s}|| M ||_{H^\sigma}^2,
\end{split}
\end{equation}
where 
Lemma 4.2(ii) with $\rho>3/2$ and $\sigma+1\leq \rho$ are applied.

Since $Z'$ and $Z$  satisfy the same estimate (4.11), for the first term on the right hand side of (4.24) we obtain
\begin{equation}
\begin{split}
&\left| \int_\mathbb{R}D^\sigma M D^\sigma \partial_{x}\left[ (v_{1} ^{p}- v_{2} ^{p})m_{2} \right]dx\right|
\lesssim || M ||_{H^\sigma} \left| \int_\mathbb{R}  \left( [D^\sigma \partial_x, v_1^p]m_{2}+(v_{1} ^{p}- v_{2} ^{p})\partial_xD^\sigma m_{2}\right)dx\right| \\
&\lesssim  || [D^\sigma \partial_x, (v_{1} ^{p}- v_{2} ^{p})]m_{2} ||_{L^2}|| M ||_{H^\sigma}+ \frac{1}{2}\left| \int_\mathbb{R} \partial_x(v_{1} ^{p}- v_{2} ^{p}) D^\sigma m_{2}   dx\right| || M ||_{H^\sigma}\\
&\lesssim || v_{1} ^{p}- v_{2} ^{p}||_{H^\rho}||m_{2}||_{H^\sigma}|| M ||_{H^\sigma} +||(v_{1} ^{p}- v_{2} ^{p})_x ||_{L^\infty} || m_{2}||_{H^\sigma}|| M ||_{H^\sigma} \\
&\lesssim || M ||_{H^\sigma} ||N||_{H^\sigma}||u_{2}||_{H^s}\sum_{i=0}^{p-1}|| v_{1} ||_{H^s}^{p-1-i} ||v_{2} ||^{i}_{H^s}.
\end{split}
\end{equation}
Meanwhile, the nonlocal term  on the right-hand side of (4.24) becomes
\begin{equation}
\begin{split}
&\left| \int_\mathbb{R}D^\sigma M  D^{\sigma } \left( \frac{p-a}{p}\partial_x v_{1} ^{p} M +\frac{p-a}{p}\partial_x ( v_{1} ^{p}-  v_{2} ^{p}) m_{2} \right)dx\right|\\
&\quad \quad\lesssim  ||M||_{H^\sigma} ||\partial_x v_{1} ^{p} M||_{H^{\sigma }}+||M||_{H^\sigma} ||\partial_x ( v_{1} ^{p}-  v_{2} ^{p}) m_{2}||_{H^{\sigma }}\\
&\quad\quad\lesssim || v_1^p||_{H^s}|| M ||_{H^\sigma}^2+|| M ||_{H^\sigma} ||N||_{H^\sigma}||u_{2}||_{H^s}\sum_{i=0}^{p-1}|| v_{1} ||_{H^s}^{p-1-i} ||v_{2} ||^{i}_{H^s},
\end{split}
\end{equation}
where Lemma 4.1(ii) with $\sigma\in(1/2,s-3)$ is used.

Combining (4.26) with (4.27) gives
\begin{equation*}
\begin{split}
\frac{1}{2}\frac{d}{dt}|| M ||_{H^\sigma} ^2\lesssim  ||z_0||_{H^s\times H^s} ^p\left(|| M||_{H^\sigma}^2   +||M||_{H^\sigma}||N||_{H^\sigma}\right), \ \forall \sigma \in (1/2,s-2),
\end{split}
\end{equation*}
which yields
\begin{equation*}
\begin{split}
 \frac{d}{dt}|| M ||_{H^\sigma}  \lesssim  ||z_0||_{H^s\times H^s}^{p/2} \left(||M||_{H^\sigma}  +||N||_{H^\sigma}\right).
\end{split}
\end{equation*}
Considering a similar estimate by $\frac{d}{dt}||N ||_{H^\sigma}$, we obtain
\begin{equation*}
\begin{split}
 \frac{d}{dt}\left(||M||_{H^\sigma}  +||N||_{H^\sigma}\right)  \lesssim  (||z_0||_{H^s\times H^s}^{p/2} +||z_0||_{H^s\times H^s}^{q/2} )\left(|| M||_{H^\sigma} +||N||_{H^\sigma}  \right).
\end{split}
\end{equation*}
Solving this inequality yields
\begin{equation*}
\begin{split}
   ||M||_{H^\sigma}  +||N||_{H^\sigma}  \lesssim \left( ||M(0)||_{H^\sigma}  +||N(0)||_{H^\sigma}\right)\exp( ||z_0||_{H^s\times H^s}^{p/2} +||z_0||_{H^s\times H^s}^{q/2} )t,
\end{split}
\end{equation*}
 which implies $M=N=0$ due to $ M(0)=N(0)=0$. 
 Furthermore, we have $u_1-u_2=U=\frac{1}{2}e^{-|x|}*M=0,v_1-v_2=U=\frac{1}{2}e^{-|x|}*N=0$ which convey the solution to (1.1) is unique in the spaces $\mathcal{C}(I;H^{s})$.
 So,
 we obtain the size estimate and lifespan for the solution $(u,v)$ given in Theorem 1.1.
\end{proof}

\subsection{Continuity of the data-to-solution map on a line}
In this section, we shall complete the proof of the Hadamard well-posedness for the Cauchy problem (1.1) on the line through showing
that the data-to-solution map $z_0=(u_0(x),v_0(x))\mapsto z(x,t)=(u(x,t),v(x,t))\in \mathcal{C}(I;H^s)\times \mathcal{C}(I;H^s)$ is continuous.
More precisely speaking, 
\begin{thm}
assume  $z_k(x,t)=(u_k(x,t),v_k(x,t))$  and $z(x,t)=((u (x,t),v (x,t)))$ are the solutions corresponding to the initial data $z_{0,k}(x)=(u_{0,k}(x),v_{0,k}(x))$ and $z_{0 }(x)=(u_{0 }(x),v_{0 }(x))$ respectively, and
 $z_{0,k}(x)=(u_{0,k}(x),v_{0,k}(x))\rightarrow (u_{0 }(x),v_{0 }(x))$ in
$H^s\times H^s$, then we have $z_{0,k}(x)=(u_k(x,t),v_k(x,t))\rightarrow z(x,t)=((u (x,t),v (x,t)))$ in $ \mathcal{C}(I;H^s)\times \mathcal{C}(I;H^s)$.
\end{thm}
Due to the presence of the   high-order  derivative  terms $u_xv_{xx}$
and $v_xu_{xx}$, we employ the approach of transforming the original solution $z(x,t)=( u (x,t),v (x,t) )$ into the solution $Z(x,t)=( u (x,t),v (x,t),m (x,t),n (x,t))$ given by Eq.(1.1),
and use the convolution operator $J_\epsilon$ ($\epsilon\in(0,1]$)  to smooth out the initial data.
Let $Z^\epsilon=( u^\epsilon,v^\epsilon,m^\epsilon,n^\epsilon)$ be the solution to Eq.(1.1) with initial data  $J_\epsilon Z_0=(j_\epsilon*u_0, j_\epsilon*v_0,j_\epsilon*(u_0-\partial_x^2u_0),j_\epsilon*(u_0-\partial_x^2u_0))$ and $Z^\epsilon_k=( u_k^\epsilon,v_k^\epsilon,m_k^\epsilon,n_k^\epsilon)$ be the solution with initial data $J_\epsilon Z_{0,k}=(j_\epsilon*u_{0,k}, j_\epsilon*v_{0,k},j_\epsilon*(u_{0,k}-\partial_x^2u_{0,k}),j_\epsilon*(u_{0,k}-\partial_x^2u_{0,k}))$.

By the following triangle inequality
\begin{equation}
\begin{split}
 ||z_k-z||_{\mathcal{C}(I;H^{s})} \leq ||z_k-z_k^\epsilon||_{\mathcal{C}(I;H^{s })}+||z_k^\epsilon-z^\epsilon||_{\mathcal{C}(I;H^{s })}+||z^\epsilon-z||_{\mathcal{C}(I;H^{s })},
\end{split}
\end{equation}
we will prove that, for any $n>N$, each of these terms can be bounded by $\eta/3$ for suitable choices of $\epsilon$ and $N$, but
$\epsilon$ 
only depends on $\eta$, whereas the choice of $N$ is dependent of both $\eta$ and $\epsilon$.

\textbf{Estimation of  $||z_k^\epsilon-z^\epsilon||_{\mathcal{C}(I;H^{s })}$, $||z_k^\epsilon-z^\epsilon||_{\mathcal{C}(I;H^{s })}$ and $||z^\epsilon-z||_{\mathcal{C}(I;H^{s })}$.}
Let $U^\epsilon=u^\epsilon-u_k^\epsilon$,$V^\epsilon=v^\epsilon-v_k^\epsilon,M^\epsilon=m^\epsilon-m_k^\epsilon,N^\epsilon=n^\epsilon-n_k^\epsilon$. Then through a direct calculation we know that 
 $(U^\epsilon,V^\epsilon,M^\epsilon,N^\epsilon)$
 solves the following equation
 \begin{equation*}
\left\{
\begin{array}{llll}
&\partial_{t}M^\epsilon+ \partial_{x}[(v^\epsilon)  ^{p} M^\epsilon]=-\partial_{x}\{[ (v^\epsilon)  ^{p}- (v_{k}^\epsilon) ^{p}]m_{k}^\epsilon\} +\frac{p-a}{p}\partial_x (v^\epsilon)  ^{p} M^\epsilon +\frac{p-a}{p}\partial_x [ (v^\epsilon)  ^{p}-  (v_{k}^\epsilon) ^{p}] m_{k}^\epsilon ,\\
&\partial_{t}N^\epsilon+ \partial_{x}[(u^\epsilon)  ^{q} N^\epsilon]=-\partial_{x}\{[ (u^\epsilon)  ^{q}- (u_{k}^\epsilon) ^{q}]n_{k}^\epsilon\} +\frac{q-b}{q}\partial_x (u^\epsilon)  ^{q} N^\epsilon +\frac{q-b}{q}\partial_x [ (u^\epsilon)  ^{q}-  (u_{k}^\epsilon) ^{q}] n_{k}^\epsilon ,\\
&M=U-\partial_x^2U,N=V-\partial_x^2V,\quad U|_{t=0}= J_\epsilon u_{0 } -J_\epsilon u_{0,k} ,V|_{t=0}=J_\epsilon v_{0 } -J_\epsilon v_{0,k} .
\end{array}
\right.
\end{equation*}
For the sake of simplicity, 
let us only consider the third equation while  other two cases can be treated in a similar way.
Applying the operator $D^{s-2}$ to the left hand side of the first equation,  multiplying by $D^{s-2} P$, and then integrating over $\mathbb{R}$, we obtain
\begin{equation*}
\begin{split}
\frac{1}{2}\frac{d}{dt}& ||M^\epsilon||_{H^{s-2}} ^2 =  \int_\mathbb{R}D^{s-2} M^\epsilon D^{s-2} \partial_{x}[(v^\epsilon)  ^{p} M^\epsilon]dx\\
&+\int_\mathbb{R}D^{s-2} M^\epsilon  D^{{s-2}} \left(-\partial_{x}\{[ (v^\epsilon)  ^{p}- (v_{k}^\epsilon) ^{p}]m_{k}^\epsilon\} +\frac{p-a}{p}\partial_x (v^\epsilon)  ^{p} M^\epsilon +\frac{p-a}{p}\partial_x [ (v^\epsilon)  ^{p}-  (v_{k}^\epsilon) ^{p}] m_{k}^\epsilon\right)dx.
\end{split}
\end{equation*}
Apparently, we have the following fact
\begin{equation*}
\begin{split}
\left| \int_\mathbb{R}D^{s-2} M^\epsilon D^{s-2} \partial_{x}[(v^\epsilon)  ^{p} M^\epsilon]dx\right|&\lesssim\left| \int_\mathbb{R}[D^{s-2}\partial_x,(v^\epsilon)  ^{p}]M^\epsilon\cdot D^{s-2} M^\epsilon dx\right|+\left| \int_\mathbb{R}(v^\epsilon)  ^{p}\partial_x (D^{s-2} M^\epsilon)^2 dx\right| \\
&\lesssim ||[D^{s-2}\partial_x,(v^\epsilon)  ^{p}]M^\epsilon||_{L^2} || D^{s-2}  M^\epsilon||_{L^2}   +||\partial_x(v^\epsilon)^p ||_{L^\infty} ||  M^\epsilon ||_{H^{s-2}}^2\\
&\lesssim || v^\epsilon ||_{H^{s-1}}^p ||  M^\epsilon ||_{H^{s-2}}^2,
\end{split}
\end{equation*}
and
\begin{equation*}
\begin{split}
&\left| \int_\mathbb{R}D^{s-2} M^\epsilon D^{s-2} \partial_{x}\{[ (v^\epsilon)  ^{p}- (v_{k}^\epsilon) ^{p}]m_{k}^\epsilon\} dx\right|\\
&\lesssim\left| \int_\mathbb{R}[D^{s-2}\partial_x,(v^\epsilon)  ^{p}- (v_{k}^\epsilon) ^{p}]m_{k}^\epsilon dx\right|||  M^\epsilon ||_{H^{s-2}}+\left| \int_\mathbb{R}[(v^\epsilon)  ^{p}- (v_{k}^\epsilon) ^{p}]\partial_x  D^{s-2} m_{k}^\epsilon dx\right|||  M^\epsilon ||_{H^{s-2}} \\
&\lesssim ||[D^{s-2}\partial_x,(v^\epsilon)  ^{p}- (v_{k}^\epsilon) ^{p}]m_{k}^\epsilon||_{L^2}  ||  M^\epsilon ||_{H^{s-2}}   +||\partial_x[(v^\epsilon)  ^{p}- (v_{k}^\epsilon) ^{p}] ||_{L^\infty} ||  m_{k}^\epsilon||_{H^{s-2}}||  M^\epsilon ||_{H^{s-2}} \\
&\lesssim ||(v^\epsilon)  ^{p}- (v_{k}^\epsilon) ^{p} ||_{H^{s-1}} ||  m_{k}^\epsilon||_{H^{s-2}} ||  M^\epsilon ||_{H^{s-2}},
\end{split}
\end{equation*}
where we use Kato-Ponce Lemma 4.2(i),  algebra property Lemma 4.1(iii), and the Sobolev's inequality with $s-1>3/2$.
Therefore, the first term on the right-hand side of Eq. (4.29) yields 
\begin{equation*}
\begin{split}
&\left| \int_\mathbb{R}D^{s-2} M^\epsilon D^{s-2}  \left(\frac{p-a}{p}\partial_x (v^\epsilon)  ^{p} M^\epsilon+\frac{p-a}{p}\partial_x [ (v^\epsilon)  ^{p}-  (v_{k}^\epsilon) ^{p}] m_{k}^\epsilon \right) dx\right|\\
&\quad\quad\lesssim\left\| \partial_x (v^\epsilon)  ^{p} M^\epsilon \right\|_{H^{s-2}}||  M^\epsilon ||_{H^{s-2}}+\left\|\partial_x [ (v^\epsilon)  ^{p}-  (v_{k}^\epsilon) ^{p}] m_{k}^\epsilon \right\|_{H^{s-2}}||  M^\epsilon ||_{H^{s-2}}\\
&\quad\quad\lesssim\left\| v^\epsilon    \right\|^{p} _{H^{s-1}}||  M^\epsilon ||^2_{H^{s-2}}+\left\|  (v^\epsilon)  ^{p}-  (v_{k}^\epsilon) ^{p}  \right\|_{H^{s-1}}||  m_{k}^\epsilon ||_{H^{s-2}}||  M^\epsilon ||_{H^{s-2}}\\
&\quad\quad\lesssim\left\| v^\epsilon    \right\|^{p} _{H^{s-1}}||  M^\epsilon ||^2_{H^{s-2}}+ ||  M^\epsilon ||_{H^{s-2}}||  N^\epsilon ||_{H^{s-1}}||  u_{k}^\epsilon ||_{H^{s}}\sum_{i=0}^{p-1}\left\|   v^\epsilon     \right\|^{p-1-i}_{H^{s-1}}\left\|  v_{k}^\epsilon         \right\|^{i}_{H^{s-1}},
\end{split}
\end{equation*}
where Lemma 2.1(ii) is applied as $s-2>1/2$.
Since $Z^\epsilon$ and $Z^\epsilon_k$ satisfy the estimate $||z^\epsilon_k||_{H^{s}},||z^\epsilon ||_{H^{s}}\lesssim ||z_0||_{H^{s }\times H^{s }}\lesssim 1$ (see  (4.11)) and 
 $||M^\epsilon||_{H^{s-2}}\cong|| U^\epsilon||_{H^{s}},||N^\epsilon||_{H^{s-2}}\cong|| V^\epsilon||_{H^{s}}$, we may derive 
\begin{equation*}
\begin{split}
 \frac{d}{dt}||U^\epsilon||_{H^{s }}  \lesssim C_s\left( ||   U^\epsilon||_{H^{s }}  +|| V^\epsilon ||_{H^{s }}  \right).
\end{split}
\end{equation*}
A similar approach could lead to the following result: 
\begin{equation*}
\begin{split}
 \frac{d}{dt}\left(  ||   U^\epsilon||_{H^{s }}  +|| V^\epsilon ||_{H^{s }}   \right) \leq  C_s  \left( ||   U^\epsilon||_{H^{s }}  +|| V^\epsilon ||_{H^{s }} \right),
\end{split}
\end{equation*}
which can be solved 
for all $t\in[0,T]$ with $T$ defined by Eq. (1.12) to generate 
\begin{equation*}
\begin{split}
  &||   U^\epsilon(t)||_{H^{s }}  +|| V^\epsilon(t) ||_{H^{s }}  \leq\left( ||   U^\epsilon(0)||_{H^{s}}  +|| V^\epsilon(0) ||_{H^{s}}\right)\exp(C_sT).
\end{split}
\end{equation*}
After $\epsilon$ is chosen, we take $N$ sufficiently large so that $$||   U^\epsilon(0)||_{H^{s}}  +|| V^\epsilon(0) ||_{H^{s }}   <\frac{\eta}{3}\exp(- C_s T ).$$
Therefore, we have $||z^\epsilon-z_k^\epsilon||_{\mathcal{C}(I;H^{s })}<\eta/3$.
Adopting the same approach as  above, we can get
$||z_k-z_k^\epsilon||_{\mathcal{C}(I;H^{s })}$, $||z^\epsilon-z||_{\mathcal{C}(I;H^{s })}<\eta/3$. Thus, the continuity of the data-to-solution map for Eq. (1.1)  has been proved.

\subsection{Well-posedness on a circle}
If the initial data $(u_0,v_0)\in H^s(\mathbb{T}), s>5/2$ is given on a circle, then the CCCH system can be dealt with a approach similar to the line (or nonperiodic) problem with a few modifications.
We need a construction of the mollifier $j_\epsilon$ and may begin with a function $j\in\mathcal{S}(\mathbb{R})$ and the periodic functions $j_\epsilon$ by $j_\epsilon\doteq \frac{1}{2\pi}\sum_{n\in\mathbb{Z}}\hat{j}(\epsilon n) e^{inx}$, which 
admit an analogous estimate to (2.2). Then, the subsequent procedure makes us observe that the proof of existence, uniqueness and continuous dependence can be replicated without any difficulty.

\section{Nonuniform dependence of the strong solution to Eq.(1.1)}

\subsection{Approximate solutions}
Let us first construct a two-parameter family of approximate solutions through using a similar method to \cite{HM22,HK}, then estimate the error and at last the $H^r$-norm error.
The approximate solution $u^{\omega,\lambda}=u^{\omega,\lambda}(t,x)$ and $v^{\omega,\lambda}=v^{\omega,\lambda}(t,x)$ to (1.1) will consist of a low and a high frequency parts, i.e.,
$$u^{\omega,\lambda}=u_l+u_h,v^{\omega,\lambda}=v_l+v_h,$$
where $\omega $  is a bounded constant, and $\lambda>0$. The high frequency part is given by
\begin{align}
u_h=u^{h,\omega,\lambda}=\lambda^{-\frac{ \delta}{2p}-s}\phi\left(\frac{x}{\lambda^{\delta/p}}\right)\cos (\lambda x-\omega^p  t),
v_h=v^{h,\omega,\lambda}=\lambda^{-\frac{\delta}{2q}-s}\varphi\left(\frac{x}{\lambda^{\delta/q }}\right)\cos (\lambda x-\omega^q  t),
\end{align}
with the cutoff function $\phi,\varphi\in C ^\infty$  satisfying
\begin{equation*}
\phi(x),\varphi(x)=\left\{
\begin{array}{llll}
1, \text{ if } |x|<1,\\
0, \text{ if } |x|\geq2.
\end{array}
\right.
\end{equation*}
Simultaneously, the low frequency parts $u_l=u_{l,\omega,\lambda}(t,x)$ and $v_l=v_{l,\omega,\lambda}(t,x)$ are the solution  to Eq.(1.1) with the following initial data
\begin{equation}
\left\{
\begin{array}{llll}
\partial_{t}u_l+v_l^{p}\partial_{x}u_l+D^{-2}\left[ \frac{a}{p} u_l\partial_xv_l^{p}+ \frac{p-a }{p}\partial_xv^{p}_l\partial_{x}^2u_l\right]+D^{-2}\partial_x[\partial_{x}v^{p}_l\partial_{x}u_l]=0, &t\in\mathbb{R},x\in\mathbb{R},\\
\partial_{t}v_l+u_l^{q}\partial_{x}v_l+D^{-2}\left[ \frac{b}{q} v_l\partial_xu_l^{q}+ \frac{q-b }{q}\partial_xu^{q}_l\partial_{x}^2v_l\right]+D^{-2}\partial_x[\partial_{x}u^{q}_l\partial_{x}v_l]=0, &t\in\mathbb{R},x\in\mathbb{R},\\
u_l(0,x)=\omega\lambda^{-\frac{1}{ q}}\tilde {\phi} \left(\frac{x}{\lambda^{\frac{\delta}{ q}}}\right),  v_l(0,x)=\omega\lambda^{-\frac{1}{p }}\tilde{\varphi}\left(\frac{x}{\lambda^{\frac{\delta}{p }}}\right),& t=0,x\in\mathbb{R},
\end{array}
\right.
\end{equation}
where $\tilde{\phi},\tilde{\varphi}\in C_0^\infty$ satisfies 
\begin{equation}
\tilde{\phi}^q(x)=1,  \quad \tilde{\varphi}^p(x)=1 \text{ if } x\in  \text{supp} \varphi\cup\text{supp} \phi.
\end{equation}
Let us now study properties of $(u_l,v_l)$ and $(u_h,v_h)$. The high frequency part $(u_h,v_h)$ 
satisfies
\begin{align*}
  &||u_h(t)||_{L^\infty} \lesssim \lambda^{-\frac{\delta}{2p}-s },||v_h(t)||_{L^\infty} \lesssim \lambda^{-\frac{\delta}{2q}-s },\\
  &||\partial_xu_h(t)||_{L^\infty} \lesssim \lambda^{-\frac{\delta}{2p}-s+1 },||\partial_xv_h(t)||_{L^\infty} \lesssim \lambda^{-\frac{\delta}{2q}-s +1}.
\end{align*}
To estimate $||u_h(t)||_{H^r} ||v_h(t)||_{H^r}$, we need the following results.
\begin{lm} (See \cite{HK}.)
Let $\psi\in \mathcal{S}(\mathbb{R})$, $\delta>0$ and $\alpha\in\mathbb{R}$. Then for any $s\geq 0$ we have 
\begin{equation}
\lim_{\lambda\rightarrow\infty}\lambda^{-\frac{1}{2}\delta-s}\left\|\psi\left(\frac{x}{\lambda^\delta}\right)\cos(\lambda x-\alpha)\right\|_{H^s}=\frac{1}{\sqrt{2}}||\psi||_{L^2}.
\end{equation}
Relation (5.4) is also true if $\cos$
 is replaced by $\sin$.
\end{lm}
So, this Lemma tells us
\begin{align*}
||u_h(t)||_{H^r}\lesssim \lambda^{r-s}\lambda^{-\frac{\delta}{2p}-r}\left\|\psi\left(\frac{x}{\lambda^{\delta/p}}\right)\cos(\lambda x-\alpha)\right\|_{H^r} \lesssim \lambda^{r-s}, \quad ||v_h(t)||_{H^s}\lesssim \lambda^{r-s}, \text{ for } \lambda\gg1.
\end{align*}

 Obviously, 
 the low frequency part 
 $(u_l,v_l)=0$ under the zero initial condition of Eq. (5.2) with $\omega=0$. 
 For $\omega \not=0$, but bounded, basic properties of $(u_l,v_l)$ are summarized in the following lemma.
\begin{lm}
Let $\omega$ be bounded, $0<\delta<2$, and $\lambda\gg1$. Then the initial-value problem (5.2) has a unique solution $(u_l,v_l)\in \mathcal{C}([0,T];H^s(\mathbb{R}))\times \mathcal{C}([0,T];H^s(\mathbb{R}))$ with $s>5/2$. Moreover, for all $r\geq0$ this solution satisfies the following estimate
\begin{equation}
||u_l(t)||_{H^r}\leq c_r \lambda^{ \frac{\delta-2}{2q} },||v_l(t)||_{H^r}\leq c_r \lambda^{ \frac{\delta-2}{2p} }, \quad 0\leq t\leq 1.
\end{equation}
\end{lm}
\begin{proof}
Apparently, for any function $\phi\in \mathcal{S} (\mathbb{R})$ we have
\begin{equation}
\left\|\phi\left(\frac{x}{\lambda^{k\delta}}\right)\right\|_{H^r}\leq \lambda^{k\delta/2}\left\|\phi\right\|_{H^r}.
\end{equation}
In fact, as per the relation $\widehat{\phi\left(\frac{x}{\rho}\right)}(\xi)=\rho\widehat{\phi }(\rho\xi)$, making the change of variables $\eta=\lambda^{k\delta} \xi$ leads to 
\begin{equation*}
\begin{split}
 \left\|\phi\left(\frac{x}{\lambda^{k\delta}}\right)\right\|_{H^r}^2&=\frac{1}{2\pi}\int_\mathbb{R} (1+\xi^2)^r \left|\lambda^{k\delta}\widehat{\phi}(\lambda^{k\delta})\right|^2d\xi\\
 &= \frac{\lambda^{k\delta} }{2\pi}\int_\mathbb{R} \left(1+\frac{\eta^2}{\lambda^{2k\delta}}\right)^r \left| \widehat{\phi}(\eta)\right|^2d\eta\\
  &= \frac{\lambda^{k\delta} }{2\pi}\int_\mathbb{R} \left(1+\eta^2\right)^r \left| \widehat{\phi}(\eta)\right|^2d\eta\\
 & =\lambda^{k\delta }\left\|\phi\right\|_{H^r}^2.
\end{split}
\end{equation*}
Therefore, from the inequlity (5.6) we know that the initial data $u_l(0),v_l(0)$
satisfy the following estimate
 \begin{equation*}
|| u_l(0)||_{H^r}\leq |\omega|\lambda^{\frac{\delta-2}{2q}}\left\|\tilde {\phi} \right\|_{H^r}
\end{equation*}
which  decays if $\delta\leq 2$ and $\omega$ is bounded.
Furthermore, using the estimate (1.12) from Theorem 1.7, we obtain the lifespan $T\geq \frac{2^\kappa-1}{2^{\kappa+1}\kappa C_s||  z_0||_{H^{s  } }^{\kappa} }\geq1$ for $\lambda\gg1$ and $\delta\leq 2$.
So, if $r\geq0$ then the estimate (1.13) of Theorem 1.7 yields 
\begin{equation*}
||u_l(t)||_{H^r}\leq  ||u_l(t)||_{H^{r+3}}\leq C_s ||u_l(0)||_{H^{r+3}}\leq C_s \lambda^{\frac{\delta-2}{2q}}.
\end{equation*}
This complete Lemma 5.2.
\end{proof}

Substituting the approximate solutions $(u^{\omega,\lambda},v^{\omega,\lambda})$
 into Eq.(1.1),  and noticing that $(u_l,v_l)$ is a solution to Eq.(5.2), we obtain the following error
\begin{equation*}
\begin{split}
F(u^{\omega,\lambda},v^{\omega,\lambda})=F_1+F_2+F_3+F_4+F_5,\\
\tilde{F}(u^{\omega,\lambda},v^{\omega,\lambda})=  \tilde{F}_1+\tilde{F}_2+\tilde{F}_3+\tilde{F}_4+\tilde{F}_5,
\end{split}
\end{equation*}
where
\begin{equation*}
\begin{split}
F_1&=\partial_t u_h+v_l^p\partial_xu_h,\\
F_2&=\left(\sum_{j=1}^pC_jv_l^{p-j}v_h^j\right)\partial_x(u_l+u_h),\\
F_3&=\frac{p-a }{p}D^{-2}\left[ \partial_x^2(u_l+u_h)\partial_x\left(\sum_{j=1}^pC_jv_l^{p-j}v_h^j\right)+\partial_x^2u_h\partial_xv_l^p\right],\\
F_4&=\frac{a}{p}D^{-2}\left[ (u_l+u_h)\partial_x\left(\sum_{j=1}^pC_jv_l^{p-j}v_h^j\right)+u_h\partial_xv_l^p\right],\\
F_5&= D^{-2} \partial_x\left[ \partial_x (u_l+u_h)\partial_x\left(\sum_{j=1}^pC_jv_l^{p-j}v_h^j\right)+\partial_x u_h\partial_xv_l^p\right].
\end{split}
\end{equation*}

\subsection{Estimation of the error $F$ in the $H^\theta$-norm}
For our convenience, let us just focus on 
the estimates for the case $r\geq0$ and $1<\delta<2$:
\begin{equation*}
\begin{split}
 &||u_l(t)||_{H^r}\lesssim \lambda^{\frac{\delta-2}{2q}}, ||v_l(t)||_{H^r}\lesssim \lambda^{\frac{\delta-2}{2p}},\\
  &||u_h(t)||_{L^\infty} \lesssim \lambda^{-\frac{\delta}{2p}-s },||v_h(t)||_{L^\infty} \lesssim \lambda^{-\frac{\delta}{2q}-s },\\
  &||\partial_xu_h(t)||_{L^\infty} \lesssim \lambda^{-\frac{\delta}{2p}-s+1 },||\partial_xv_h(t)||_{L^\infty} \lesssim \lambda^{-\frac{\delta}{2q}-s +1},\\
  &||u_h(t)||_{H^r}, ||v_h(t)||_{H^r} \lesssim \lambda^{-s+r  }.
\end{split}
\end{equation*}
If $\theta>1/2$, by the Sobolev's lemma, Lemma 4.1(i) presents the following algebraic property
\begin{equation*}
||fg||_{H^\theta}\leq c_\theta||f||_{\mathcal{C}^1}||g||_{H^\theta}.
\end{equation*}
Next, let us estimate the error $F$ in the $H^\theta$ norm where $\theta\in(3/2,s)$.\\

\textbf{Estimation of $F_1$ in the $H^\theta$-norm}
Apparently,  $\tilde{\varphi}^p\left(\frac{x}{\lambda^{\delta/p}}\right) \phi\left(\frac{x}{\lambda^{\delta/p}}\right) =\phi\left(\frac{x}{\lambda^{\delta/p}}\right) $ and $\partial_t u_h$ can be rewritten as
\begin{equation*}
\begin{split}
\partial_tu_h(x,t)&=\omega^p  \lambda^{-\frac{\delta}{2p}-s} \tilde{\varphi}^p\left(\frac{x}{\lambda^{\delta/p}}\right)  \phi\left(\frac{x}{\lambda^{\delta/p}}\right)\sin (\lambda x-\omega^p t)\\
&=\lambda^{1-\frac{\delta}{2p}-s}  v_l^p(x,0) \phi\left(\frac{x}{\lambda^{\delta/p}}\right)\sin (\lambda x-\omega^p  t),
\end{split}
\end{equation*}
and
\begin{equation*}
\begin{split}
\partial_x u_h(x,t)= -\lambda^{1-\frac{\delta}{2p}-s}\phi\left(\frac{x}{\lambda^{\delta/p}}\right)\sin (\lambda x-\omega^p  t)+\lambda^{ -\frac{3\delta  }{2p}-s} \phi'\left(\frac{x}{\lambda^{\delta/p}}\right)\cos (\lambda x-\omega^p  t).
\end{split}
\end{equation*}
Furthermore, we have 
\begin{align}
F_1&=-[v_l^p(t,x)-v_l^p(0,x)] \lambda^{1-\frac{\delta}{2p}-s}\phi\left(\frac{x}{\lambda^{\delta/p}}\right)\sin (\lambda x-\omega^p t)\tag{*}\\
&\quad+ [v_l^p(t,x)]\lambda^{ -\frac{3\delta}{2p}-s} \phi'\left(\frac{x}{\lambda^{\delta/p}}\right)\cos (\lambda x-\omega^p  t).\tag{**}
\end{align}
Applying the algebraic property leads to
\begin{equation*}
\begin{split}
\|(*)\|_{H^\theta}&\lesssim \lambda^{1-\frac{\delta}{2p}-s} \left\| [v_l^p(t)-v_l^p(0)]  \phi\left(\frac{x}{\lambda^{\delta/p}}\right)\sin (\lambda x-\omega^p  t)\right\|_{H^\theta}\\
&\lesssim \lambda^{1-\frac{\delta}{2p}-s}\cdot \left\|v_l^p(t)-v_l^p(0) \right\|_{H^\theta}\left\|  \phi\left(\frac{x}{\lambda^{\delta/p}}\right)\sin (\lambda x-\omega^p  t)\right\|_{H^\theta}\\
&\lesssim \lambda^{1-\frac{\delta}{2p}-s}\lambda^{\theta+\delta/2p} \left\|v_l^p(t)-v_l^p(0) \right\|_{H^\theta}\\
&\lesssim \lambda^{1-s+\theta}\cdot \left\|v_l^p(t)-v_l^p(0) \right\|_{H^\theta}.
\end{split}
\end{equation*}
To estimate the $H^\theta$-norm of the difference $v_l^p(t)-v_l^p(0)$, we adopt the fundamental theorem of calculus in time variable to obtain
\begin{equation*}
\begin{split}
  \left\|v_l^p(t)-v_l^p(0)\right\|_{H^\theta} \leq\int_0^t\left\| v_l^{p-1}(x,\tau) \right\|_{H^\theta}\left\|\partial_tv_l(x,\tau) \right\|_{H^\theta} d\tau, t\in[0,T].\\
\end{split}
\end{equation*}
Thus, it follows from Eq.(5.2) that
\begin{equation*}
\begin{split}
 \left\|\partial_tv_{l} \right\|_{H^\theta}&\lesssim  \left\|u_l^{q}\partial_{x}v_l\right\|_{H^\theta}+\left\|v_l\partial_xu_l^{q}\right\|_{H^{\theta-2}} +\left\|\partial_xu^{q}_l\partial_{x}^2v_l\right\|_{H^{\theta-2}}+\left\| \partial_{x}u^{q}_l\partial_{x}v_l \right\|_{H^{\theta-1}} \\
&\lesssim\left\|  u_l\right\|_{H^{s+1}} ^q \left\|  v_l  \right\|_{H^{s+1}}.
\end{split}
\end{equation*}
Therefore, we have
\begin{equation*}
\begin{split}
||(*)||_{H^\theta}\lesssim \lambda^{1-s+\theta} \left\|  u_l\right\|_{H^{s+1}} ^q \left\|  v_l  \right\|_{H^{s+1}}^p \lesssim \lambda^{-1-s  +\delta +\theta},\quad\lambda \gg1.
\end{split}
\end{equation*}
On the other hands, we know
\begin{equation*}
\begin{split}
||(**)||_{H^\theta}&=\left\| v_l^p(x,t) \lambda^{ -\frac{3\delta}{2p}-s}\partial_x\phi\left(\frac{x}{\lambda^{\delta/p}}\right)\cos (\lambda x-\omega^p   t)\right\|_{H^\theta}\\
&\lesssim\lambda^{ -\frac{3\delta}{2p}-s} \left\|v_l^p (x,t)\right\|_{H^{\theta}} \left\| \phi'\left(\frac{x}{\lambda^{\delta/p}}\right)\cos (\lambda x-\omega^p  t)\right\|_{H^\theta}\\
&\lesssim\lambda^{-\frac{3\delta}{2p}-s} \lambda^{\theta+\delta/2p}\lambda^{ (\delta-2)/2}\\
&\lesssim\lambda^{-s +\delta/2-\delta/p-1+\theta},\quad\lambda \gg1.
\end{split}
\end{equation*}
Hence, we obtain
\begin{equation*}
\begin{split}
||F_1||_{H^\theta} &\lesssim||(*)||_{H^\theta}+ ||(**)||_{H^\theta}  \lesssim   \lambda^{-1-s  +\delta +\theta},\quad\lambda \gg1.
\end{split}
\end{equation*}

\textbf{Estimation of $F_2$ in the $H^\theta$-norm} is given by
\begin{equation*}
\begin{split}
||F_2||_{H^\theta}&=   \left\|\left(\sum_{j=1}^pC_jv_l^{p-j}v_h^j\right)\partial_x(u_l+u_h)\right\|_{H^\theta} \\
&\lesssim\left\|\left(\sum_{j=1}^pC_jv_l^{p-j}v_h^j\right)\partial_x u_l \right\|_{H^\theta} +\left\| \sum_{j=1}^pC_jv_l^{p-j}v_h^j \partial_x u_h \right\|_{H^\theta} \\
&\lesssim\sum_{j=1}^p\left\|  v_l\right\|^{p-j}_{H^\theta}  \left\|v_h \right\|^j_{H^\theta}  \left\|\partial_x u_l \right\|_{H^\theta} +\sum_{j=1}^p\left\|  v_l\right\|^{p-j}_{H^\theta} \left( \left\|v_h \right\|^j_{H^\theta}  \left\|\partial_x u_h \right\|_{L^\infty} +\left\|v_h \right\|^j_{L^\infty}  \left\|\partial_x u_h \right\|_{H^\theta}\right)\\
&\lesssim\sum_{j=1}^p\lambda^{\frac{(p+1-j)(\delta-2)}{2p}+j(-s+\theta) }
+ \sum_{j=1}^p\lambda^{\frac{(p -j)(\delta-2)}{2p}+j(-s+\theta) -\frac{ \delta }{2p}-s+1}+ \sum_{j=1}^p\lambda^{\frac{(p-j)(\delta-2)}{2p}-j(s+\frac{ \delta }{2p})-s+\theta+1}\\
&\lesssim\sum_{j=1}^p\lambda^{ j(-s+\theta) }
+ \sum_{j=1}^p\lambda^{ j(-s+\theta) -\frac{ \delta }{2p}-s+1} + \sum_{j=1}^p\lambda^{ -j(s+\frac{ \delta }{2p})-s+\theta+1}\\
&\lesssim \lambda^{ -s+\theta }
+ \lambda^{   -2s+\theta  -\frac{ \delta }{2p} +1} + \lambda^{ - 2s-\frac{ \delta }{2p}  +\theta+1}.
\end{split}
\end{equation*}

 \textbf{Estimation of $F_{3}$ in the $H^\theta$-norm.} For $\theta-1>1/2$, it follows from Lemma 4.1(ii) that
\begin{equation*}
\begin{split}
||F_3||_{H^\theta}&=\left\|\partial_x^2(u_l+u_h)\right\|_{H^{\theta-2}}\left\|\partial_x\left(\sum_{j=1}^pC_jv_l^{p-j}v_h^j\right)\right\|_{H^{\theta-1}}
+\left\|\partial_x^2u_h\right\|_{H^{\theta-2}} \left\|\partial_xv_l^p\right\|_{H^{\theta-1}}\\
&\lesssim  (\left\|  u_l \right\|_{H^{\theta}}+\left\|  u_h \right\|_{H^{\theta}}) \sum_{j=1}^p\left\|  v_l\right\|^{p-j}_{H^{\theta}}  \left\|v_h \right\|^j_{H^{\theta}}+\left\|u_h\right\| _{H^{\theta}} \left\| v_l \right\|^p_{H^{{\theta} }}\\
&\lesssim  (\lambda^{\frac{\delta-2}{2p}}+\lambda^{-s+\theta})\sum_{j=1}^p\lambda^{\frac{(p-j)(\delta-2)}{2p}+j( \theta-s)}+\lambda^{-s+\theta+\frac{p(\delta-2)}{2p}}.
\end{split}
\end{equation*}

 \textbf{Estimation of $F_{4}$ in the $H^\theta$-norm.} For $\theta-1>1/2$, it follows from Lemma 4.1 that
\begin{equation*}
\begin{split}
||F_4||_{H^\theta}&=\left\|(u_l+u_h)\partial_x\left(\sum_{j=1}^pC_jv_l^{p-j}v_h^j\right)+u_h\partial_xv_l^p\right\|_{H^{\theta-2}} \\
  &\leq\left\|(u_l+u_h)\right\|_{H^{\theta-1}}\left\| \sum_{j=1}^pC_jv_l^{p-j}v_h^j \right\|_{H^{\theta-1}}+\left\|u_h\right\|_{H^{\theta-1}} \left\| v_l^p \right\|_{H^{\theta-1}} \\
&\lesssim  (\left\|  u_l \right\|_{H^{\theta-1}}+\left\|  u_h \right\|_{H^{\theta-1}}) \sum_{j=1}^p\left\|  v_l\right\|^{p-j}_{H^{\theta-1}}  \left\|v_h \right\|^j_{H^{\theta-1}}+\left\|u_h\right\| _{H^{\theta-1}} \left\| v_l \right\|^p_{H^{{\theta-1} }}\\
&\lesssim  (\left\|  u_l \right\|_{H^{\theta}}+\left\|  u_h \right\|_{H^{\theta}}) \sum_{j=1}^p\left\|  v_l\right\|^{p-j}_{H^{\theta}}  \left\|v_h \right\|^j_{H^{\theta}}+\left\|u_h\right\| _{H^{\theta}} \left\| v_l \right\|^p_{H^{{\theta} }}\\
&\lesssim  (\lambda^{\frac{\delta-2}{2p}}+\lambda^{-s+\theta})\sum_{j=1}^p\lambda^{\frac{(p-j)(\delta-2)}{2p}+j( \theta-s)}+\lambda^{-s+\theta+\frac{p(\delta-2)}{2p}}.
\end{split}
\end{equation*}

 \textbf{Estimation of $F_{5}$ in the $H^\theta$-norm.} For $\theta-1>1/2$, it follows from Lemma 4.1 that
\begin{equation*}
\begin{split}
||F_5||_{H^\theta}&=\left\|\partial_x (u_l+u_h)\partial_x\left(\sum_{j=1}^pC_jv_l^{p-j}v_h^j\right)+\partial_x u_h\partial_xv_l^p\right\|_{H^{\theta-1}} \\
&\lesssim  (\left\|  u_l \right\|_{H^{\theta}}+\left\|  u_h \right\|_{H^{\theta}}) \sum_{j=1}^p\left\|  v_l\right\|^{p-j}_{H^{\theta}}  \left\|v_h \right\|^j_{H^{\theta}}+\left\|u_h\right\| _{H^{\theta}} \left\| v_l \right\|^p_{H^{{\theta} }}\\
&\lesssim  (\lambda^{\frac{\delta-2}{2p}}+\lambda^{-s+\theta})\sum_{j=1}^p\lambda^{\frac{(p-j)(\delta-2)}{2p}+j( \theta-s)}+\lambda^{-s+\theta+\frac{p(\delta-2)}{2p}}.
\end{split}
\end{equation*}

Collecting all error estimates together produces the following theorem.
\begin{thm}
Let $s>5/2 $, $3/2<\theta< s $, and $0<\delta<\min\{2,1+s  -\theta\}$.
If $\omega$ is bounded in $\mathbb{R}$ and $\lambda\gg1$, then we have 
\begin{equation}
||F||_{H^\theta},||\tilde{F}||_{H^\theta}\lesssim \lambda^{-\theta_s}, \text{ for } \lambda\gg1,0<t<T,
\end{equation}
where
\begin{equation}\theta_s =1+s-\delta-\theta>0.
\end{equation}
\end{thm}

\subsection{Estimation between approximate and actual solutions}
Let us now estimate the difference between the approximate and actual solutions.
Let $z_{\omega,\lambda}(t,x)=(u_{\omega,\lambda}(t,x),v_{\omega,\lambda}(t,x))$ be the solution to Eq.(1.1) with initial data given by the approximate solution
 $z^{\omega,\lambda}(t,x)=(u^{\omega,\lambda}(t,x),v^{\omega,\lambda}(t,x))$ evaluated at zero time, that is, $z_{\omega,\lambda}(t,x)$ satisfies
\begin{equation}
\left\{
\begin{array}{llll}
\partial_{t}u_{\omega,\lambda}+v_{\omega,\lambda}^{p}\partial_{x}u_{\omega,\lambda}+D^{-2}\left[ \frac{a}{p} u_{\omega,\lambda}\partial_xv_{\omega,\lambda}^{p}+ \frac{p-a }{p}\partial_xv^{p}_{\omega,\lambda}\partial_{x}^2u_{\omega,\lambda}\right]+D^{-2}\partial_x[\partial_{x}v^{p}_{\omega,\lambda}\partial_{x}u_{\omega,\lambda}]=0,\\
\partial_{t}v_{\omega,\lambda}+u_{\omega,\lambda}^{q}\partial_{x}v_{\omega,\lambda}+D^{-2}\left[ \frac{b}{q} v_{\omega,\lambda}\partial_xu_{\omega,\lambda}^{q}+ \frac{q-b }{q}\partial_xu^{q}_{\omega,\lambda}\partial_{x}^2v_{\omega,\lambda}\right]+D^{-2}\partial_x[\partial_{x}u^{q}_{\omega,\lambda}\partial_{x}v_{\omega,\lambda}]=0,  \\
u_{\omega,\lambda}(0,x)=u^{\omega,\lambda}(0,x)=\omega\lambda^{-\frac{1}{ q}}\tilde {\phi} \left(\frac{x}{\lambda^{ \delta/ q}}\right)+\lambda^{-\frac{ \delta}{2p}-s}\phi\left(\frac{x}{\lambda^{\delta/p}}\right)\cos (\lambda x ),\\
v_{\omega,\lambda}(0,x)=v^{\omega,\lambda}(0,x)=\omega\lambda^{-\frac{1}{ p}}\tilde {\phi} \left(\frac{x}{\lambda^{ \delta/ p}}\right)+\lambda^{-\frac{ \delta}{2q}-s}\phi\left(\frac{x}{\lambda^{\delta/q}}\right)\cos (\lambda x ).
\end{array}
\right.
\end{equation}
Noticing that $z_{\omega,\lambda}(0,x)=(u_{\omega,\lambda}(0,x),v_{\omega,\lambda}(0,x))\in H^s\times H^s,s>0$, it follows from Lemmas 5.1 and 5.2 that
\begin{equation}
\begin{split}
||u_{\omega,\lambda}(t,x)||_{H^s}\lesssim||u_{\omega,\lambda}(0,x)||_{H^s}\leq ||u_{l}(0)||_{H^s}+||u_h(0)||_{H^s}\lesssim \lambda^{ \frac{\delta-2}{2q}}+1\lesssim1, \lambda\gg1,\\
||v_{\omega,\lambda}(t,x)||_{H^s}\lesssim||v_{\omega,\lambda}(0,x)||_{H^s}\leq ||v_{l}(0)||_{H^s}+||v_h(0)||_{H^s}\lesssim \lambda^{ \frac{\delta-2}{2p}}+1\lesssim1, \lambda\gg1.
\end{split}
\end{equation}
Thus, by Theorem 1.7 with $s>5/2$ we have that for any bounded $\omega$ 
and $\lambda\gg1$, (5.9) has a unique solution $z_{\omega,\lambda}\in \mathcal{C}([0,T];H^s)\times \mathcal{C}([0,T];H^s)$ where 
\begin{equation*}
T\gtrsim \frac{2^\kappa-1}{2^{\kappa+1}\kappa C_s||  z_{\omega,\lambda}(0)||_{H^{s  } }^{\kappa} }\gtrsim \frac{1}{(\lambda^{ \frac{\delta-2}{2\kappa}}+1)^\kappa},\lambda\gg1.
\end{equation*}
Now, let us check if the difference between two sequences of the approximate solutions $z^{\omega,\lambda}(t)-z_{\omega,\lambda}(t)$ goes to zero at time $t=0$ and { stay apart at $t>0$}.
Let
$$U=u^{\omega,\lambda}-u_{\omega,\lambda}\doteq u_1-u_2,V=v^{\omega,\lambda}-v_{\omega,\lambda}\doteq v_1-v_2.$$
Then, $(U,V)$ solves the following equation
\begin{equation}
\left\{
\begin{array}{llll}
&\partial_t U= F(u^{\omega,\lambda},v^{\omega,\lambda} )-  E_1-E_2-E_3-E_4,\\
&\partial_t V= \tilde{F}(u^{\omega,\lambda},v^{\omega,\lambda})- \tilde{E}_1-\tilde{E}_2-\tilde{E}_3-\tilde{E}_4 ,\\
&U(x,0)=V(x,0)=0, x\in \mathbb{R},
\end{array}
\right.
\end{equation}
where
\begin{equation*}
\begin{split}
E_1&=v^p_1 \partial_xu_1-v^p_2 \partial_xu_2=v_1^p\partial_xU+\left(\sum_{i=0}^{p-1}v_1^{p-1-i}v_2^{i}\right)V\partial_xu_2,\\
E_2&=\frac{p-a }{p}D^{-2}\left[  \partial_xv^{p}_1\partial_{x}^2u_1-\partial_xv^{p}_2\partial_{x}^2u_2\right]=\frac{p-a }{p}D^{-2}\left[ \partial_x^2U\partial_xv_1^p+\partial_x^2u_2\partial_x\left(V\sum_{i=0}^{p-1}v_1^{p-1-i}v_2^{i}\right)\right],\\
E_3&=\frac{a}{p} D^{-2}\left[ u_1\partial_xv_1^{p}-u_1\partial_xv_1^{p}\right]=\frac{a}{p} D^{-2}\left[ U\partial_xv_1^p+u_2\partial_x\left(V\sum_{i=0}^{p-1}v_1^{p-1-i}v_2^{i}\right)\right],\\
E_4&=D^{-2}\partial_x[\partial_{x}v^{p}_1\partial_{x}u_1-\partial_{x}v^{p}_2\partial_{x}u_2]= D^{-2}\partial_x\left[ \partial_x U\partial_xv_1^p+\partial_x u_2\partial_x\left(V\sum_{i=0}^{p-1}v_1^{p-1-i}v_2^{i}\right)\right].
\end{split}
\end{equation*}

\begin{pro}
Let $s>5/2 $, $3/2<\theta< s-1 $, and $0<\delta<\min\{2,1+s  -\theta\}$. If $\omega$ is  bounded in  $\mathbb{R}$ and $\lambda\gg1$, then
\begin{equation}
||U(t)||_{H^\theta},||V(t)||_{H^\theta} \lesssim\lambda^{-\theta_s}, \text{ for } \lambda\gg1,0\leq t\leq T,
\end{equation}
where $\theta_s$ is defined by (5.8).
\end{pro}
\begin{proof}
Plugging the operator $D^\theta$ onto both sides of the first equation in (5.11), multiplying by $ D^\theta p$, and integrating the resulting equation, we obtain
\begin{equation}
\begin{split}
\frac{1}{2}\frac{d}{dt}||U(t)||_{H^\theta}^2&=\int_\mathbb{R}  D^\theta F D^\theta U dx - \int_\mathbb{R}  D^\theta (E_1+E_2+E_3+E_4) D^\theta Udx.
\end{split}
\end{equation}
Next, let us estimate each term on the right-hand side of (5.13).\\
\textbf{Estimation of $E_1$.}  As per 
the Cauchy-Schwartz inequality and Theorem 5.1, we have
\begin{equation}
\begin{split}
 \left|\int_\mathbb{R}  D^\theta F D^\theta Udx\right|\leq ||F||_{H^\theta}||p||_{H^\theta}\lesssim \lambda^{-\theta_s}||U(t)||_{H^\theta}.
\end{split}
\end{equation}
\textbf{Estimation of $E_2$.} 
The second term on the right-hand side of (5.13) can be estimated by
\begin{equation}
\begin{split}
\left|\int_\mathbb{R}  D^\theta  v_1^p\partial_xU  D^\theta Udx\right|&\leq\left|\int_\mathbb{R}  D^\theta \left(\partial_x(v_1^pU)- U\partial_xv_1^p  \right)D^\theta Udx\right|\\
&\leq \left|\int_\mathbb{R}  [D^\theta \partial_x,v_1^p]U    D^\theta Udx\right|+\left|\int_\mathbb{R} v_1^p D^\theta \partial_xU D^\theta Udx\right|+||v_1||_{H^\theta}^p ||U||_{H^\theta}^2\\
&\leq \left\| [D^\theta \partial_x,v_1^p]U     \right\|_{L^2}||U||_{H^\theta}+\frac{1}{2}\left|\int_\mathbb{R} \partial_xv_1^p (D^\theta U)^2dx\right|+||v_1||_{H^\theta}^p ||U||_{H^\theta}^2\\
&\leq ( ||v_1||^p_{H^\theta}||U||_{L^\infty}+||\partial_x v_1^p||_{L^\infty}||U||_{H^\theta})||U||_{H^\theta}+\frac{3}{2}||v_1||_{H^\theta}^p ||U||_{H^\theta}^2\\
&\lesssim||v_1||_{H^\theta}^p ||U||_{H^\theta}^2.
\end{split}
\end{equation}
Thus, we get
\begin{equation}
\begin{split}
\left|\int_\mathbb{R}  D^\theta  E_1  D^\theta Udx\right| & \leq  \left|\int_\mathbb{R}  D^\theta  v_1^p\partial_xU  D^\theta Udx\right|+ \left|\int_\mathbb{R}  D^\theta \left(\sum_{i=0}^{p-1}v_1^{p-1-i}v_2^{i}\right)V\partial_xu_2     D^\theta Udx\right|\\
&\lesssim||v_1||_{H^\theta}^p ||U||_{H^\theta}^2 + ||U||_{H^\theta}||V||_{H^\theta}\left(\sum_{i=0}^{p-1}||v_1||^{p-1-i}_{H^\theta}||v_2||_{H^\theta} ^{i}\right)||u_2||_{H^{\theta+1}}.
\end{split}
\end{equation}
If $\theta-1>\frac{1}{2}$, by Lemma 4.1 we have 
\begin{equation}
\begin{split}
\left|\int_\mathbb{R}  D^\theta  E_2 D^\theta Udx\right| &\lesssim \left\|  \partial_x^2U\partial_xv_1^p\right\|_{H^{\theta-2}}\left\|U\right\|_{H^{\theta }}+ \left\|\partial_x^2u_2\partial_x\left(V\sum_{i=0}^{p-1}v_1^{p-1-i}v_2^{i}\right)\right\|_{H^{\theta-2}}\left\|U\right\|_{H^{\theta }}\\
 &\lesssim \left\|  \partial_x^2U\right\|_{H^{\theta-2}}\left\|\partial_xv_1^p\right\|_{H^{\theta-1}}\left\|U\right\|_{H^{\theta }}+ \left\|\partial_x^2u_2\right\|_{H^{\theta-2}}\left\|\partial_x\left(V\sum_{i=0}^{p-1}v_1^{p-1-i}v_2^{i}\right)\right\|_{H^{\theta-1}}\left\|U\right\|_{H^{\theta }}\\
 &\lesssim  \left\| v_1^p\right\|_{H^{\theta}}\left\|U\right\|^2_{H^{\theta }}+ \left\|V\right\|_{H^{\theta}}\left\|U\right\|_{H^{\theta }} \left\|u_2\right\|_{H^{\theta}}\sum_{i=0}^{p-1} \left\|v_1\right\|^{p-1-i}_{H^{\theta}}
 \left\|v_2\right\|^{i}_{H^{\theta}},
\end{split}
\end{equation}
\begin{equation}
\begin{split}
\left|\int_\mathbb{R}  D^\theta  E_3 D^\theta Udx\right| &\lesssim \left\|  U\partial_xv_1^p\right\|_{H^{\theta-1}}\left\|U\right\|_{H^{\theta }}+ \left\|u_2\partial_x\left(V\sum_{i=0}^{p-1}v_1^{p-1-i}v_2^{i}\right)\right\|_{H^{\theta-1}}\left\|U\right\|_{H^{\theta }}\\
 &\lesssim  \left\| v_1^p\right\|_{H^{\theta}}\left\|U\right\|^2_{H^{\theta }}+ \left\|V\right\|_{H^{\theta}}\left\|U\right\|_{H^{\theta }} \left\|u_2\right\|_{H^{\theta}}\sum_{i=0}^{p-1} \left\|v_1\right\|^{p-1-i}_{H^{\theta}}
 \left\|v_2\right\|^{i}_{H^{\theta}},
\end{split}
\end{equation}
and
\begin{equation}
\begin{split}
\left|\int_\mathbb{R}  D^\theta  E_4 D^\theta Udx\right| &\lesssim \left\|  \partial_x U\partial_xv_1^p\right\|_{H^{\theta-1}}\left\|U\right\|_{H^{\theta }}+ \left\|\partial_x u_2\partial_x\left(V\sum_{i=0}^{p-1}v_1^{p-1-i}v_2^{i}\right)\right\|_{H^{\theta-1}}\left\|U\right\|_{H^{\theta }}\\
 &\lesssim  \left\| v_1\right\|^p_{H^{\theta}}\left\|U\right\|^2_{H^{\theta }}+ \left\|V\right\|_{H^{\theta}}\left\|U\right\|_{H^{\theta }} \left\|u_2\right\|_{H^{\theta}}\sum_{i=0}^{p-1} \left\|v_1\right\|^{p-1-i}_{H^{\theta}}
 \left\|v_2\right\|^{i}_{H^{\theta}}.
\end{split}
\end{equation}
\textbf{The final differential inequality.}
Combining the above estimates (5.14)-(5.19) leads to
$$\frac{1}{2}\frac{d}{dt} ||U||_{H^\theta}^2 \lesssim A_1  ||U||_{H^\theta}^2 +B_1   ||U||_{H^\theta}||V||_{H^\theta} +\lambda^{-\theta_s} ||U||_{H^\theta},\text { for }\lambda \gg1,$$
that is
$$ \frac{d}{dt}||U||_{H^\theta}\lesssim A_1 ||U||_{H^\theta}+ B_1 ||V||_{H^\theta} +\lambda^{-\theta_s},$$
where $A\lesssim ||v^{\omega,\lambda}||_{H^\theta}^p =||v^{\omega,\lambda}||^p_{H^s}   \lesssim1$,  and $B=(\left\|u_2\right\|_{H^{\theta}}+\left\|u_2\right\|_{H^{\theta+1}})\sum_{i=0}^{p-1} \left\|v_1\right\|^{p-1-i}_{H^{\theta}}
 \left\|v_2\right\|^{i}_{H^{\theta}}\lesssim1.$
In the above calculation, we utilized the following well-posedness inequality
\begin{equation}
\begin{split}
 \left\|u_1\right\|_{H^{\theta}}=\left\|u^{\omega,\lambda}\right\|_{H^{\theta}}\leq \left\|u^{\omega,\lambda}\right\|_{H^{s}}\lesssim \lambda^{ \frac{\delta-2}{2q}}+1, \left\|u_2\right\|_{H^{\theta}}=\left\|u_{\omega,\lambda}\right\|_{H^{\theta}}\leq \left\|u_{\omega,\lambda}\right\|_{H^{s}}\lesssim \lambda^{ \frac{\delta-2}{2q}}+1.\\
 \left\|v_1\right\|_{H^{\theta}}=\left\|v^{\omega,\lambda}\right\|_{H^{\theta}}\leq \left\|v^{\omega,\lambda}\right\|_{H^{s}}\lesssim \lambda^{ \frac{\delta-2}{2p}}+1, \left\|v_2\right\|_{H^{\theta}}=\left\|v_{\omega,\lambda}\right\|_{H^{\theta}}\leq \left\|v_{\omega,\lambda}\right\|_{H^{s}}\lesssim \lambda^{ \frac{\delta-2}{2p}}+1.\\
\end{split}
\end{equation}
In a similar way, we are able to get
$$ \frac{d}{dt}||V||_{H^\theta}\lesssim A_1 ||V||_{H^\theta}+ B_1 ||U||_{H^\theta} +\lambda^{-\theta_s}.$$
Therefore, we have
$$ \frac{d}{dt}(||U||_{H^\theta}+||V||_{H^\theta})\lesssim M( ||U||_{H^\theta}+||V||_{H^\theta})   +\lambda^{-\theta_s},$$
which can be solved with the initial condition $||U(0)||_{H^r}=||V(0)||_{H^r}=0$ in the form of 
\begin{equation}
||U(t)||_{H^r} \lesssim\lambda^{-\theta_s}, ||V(t)||_{H^r} \lesssim\lambda^{-\theta_s},\text{ for } \lambda\gg1,0\leq t\leq T,
\end{equation}
where $\theta_s$ is defined by (5.8). This concludes the proof of proposition 5.1.
\end{proof}

Let us now present the proof of Theorem \ref{result8}.
\begin{proof}[Proof of Theorem \ref{result8}]
Let $s>5/2$, 
and assume $(u_{1,\lambda}(t,x),v_{1,\lambda}(t,x))$ and $(\tilde{u}_{0,\lambda}(t,x),\tilde{v}_{0,\lambda}(t,x))$ are the unique solutions to (5.9) with initial data $(u_{1,\lambda}(0,x),v_{1,\lambda}(0,x))$ and $(\tilde{u}_{0,\lambda}(0,x),\tilde{v}_{0,\lambda}(0,x))$, respectively.

It follows from Theorem 1.7 that these solutions belong to $\mathcal{C}([0,T];H^s)\times \mathcal{C}([0,T];H^s)$.
By (5.10) and the assumptions right after Theorem 1.7, we see that $T$ is independent of $\lambda\gg1$ and $0<\delta<\min\{2,1+s  -\theta\}$.
Denoting $k=[s]+2 $ and applying estimate (1.13), we have
 \begin{equation}
||u_{\omega,\lambda}(t)||_{H^k},||v_{\omega,\lambda}(t)||_{H^k} \lesssim
||z^{\omega,\lambda}(0)||_{H^k}\lesssim \lambda^{k-s}, \text{ for } \lambda\gg1,0\leq t\leq T.
\end{equation}
By (5.4) and (5.5), we obtain
\begin{equation}
\begin{split}
||u^{\omega,\lambda}(t)||_{H^k},||v^{\omega,\lambda}(t)||_{H^k}\leq||z_l(t)||_{H^k}+
||z_h(t)||_{H^k} \lesssim   \lambda^{k-s}.
\end{split}
\end{equation}
So, from (5.22)-(5.23) we have the following estimate for the difference between $z^{\omega,\lambda}$ and $z_{\omega,\lambda}$ in the $H^k$-norm:
\begin{equation}
\begin{split}
||u^{\omega,\lambda}(t)-u_{\omega,\lambda}(t)||_{H^k},||v^{\omega,\lambda}(t)-v_{\omega,\lambda}(t)||_{H^k} \lesssim  \lambda^{k-s}, 0\leq t\leq T.
\end{split}
\end{equation}
On the other hand, applying (5.12) with the choice of $\omega\in\{0,1\}$ generates 
\begin{equation}
\begin{split}
||u^{\omega,\lambda}(t)-u_{\omega,\lambda}(t)||_{H^\theta} ,||v^{\omega,\lambda}(t)-v_{\omega,\lambda}(t)||_{H^\theta}\lesssim  \lambda^{-\theta_s}, 0\leq t\leq T.
\end{split}
\end{equation}
By the interpolation inequality with $s_1=\theta$ and $s_2=[s]+2=k$
\begin{equation}
\begin{split}
||f||_{H^s} \leq ||f||_{H^{s_1}}^\frac{s_2-s}{s_2-s_1} ||f||_{H^{s_2}}^\frac{s-s_1}{s_2-s_1},
\end{split}
\end{equation}
 and the estimates (5.24) and (5.25), we obtain
\begin{equation}
\begin{split}
||u^{\omega,\lambda}(t)-u_{\omega,\lambda}(t)||_{H^s}&\leq ||u^{\omega,\lambda}(t)-u_{\omega,\lambda}(t)||_{H^\theta}^\frac{k-s}{k-\theta} ||u^{\omega,\lambda}(t)-u_{\omega,\lambda}(t)||_{H^k}^\frac{s-\theta}{k-\theta}\\
&\lesssim \lambda^\frac{-\theta_s(k-s)}{k-\theta}\lambda^\frac{(k-s)(s-\theta)}{k-\theta} \\
&\lesssim \lambda^\frac{-(\theta_s-s+\theta)(k-s)}{k-\theta} \\
&\lesssim \lambda^\frac{-(1-\delta )(k-s)}{k-\theta},\\
||v^{\omega,\lambda}(t)-v_{\omega,\lambda}(t)||_{H^s}&\lesssim \lambda^\frac{-(1-\delta )(k-s)}{k-\theta}.
\end{split}
\end{equation}
Obviously, we must have $\frac{(1- \delta )(k-s)}{k-r}>0$, which is equivalent to $\delta< 1$.

Next, we shall apply the estimate (5.27) to prove nonuniform dependence when $s>5/2$.\\
\textbf{Behavior at $t=0$.} Since $0<\delta<1$, at $t=0$ we have
\begin{equation}
\begin{split}
||u_{1,\lambda}(0)-u_{0,\lambda}(0)||_{H^s}
&\lesssim \lambda^\frac{ \delta-2 }{2p}||\tilde{\phi}||_{H^s}\rightarrow0 \text{ as } \lambda\rightarrow\infty.\\
||v_{1,\lambda}(0)-v_{0,\lambda}(0)||_{H^s}
&\lesssim \lambda^\frac{ \delta-2 }{2q}||\tilde{\varphi}||_{H^s}\rightarrow0 \text{ as } \lambda\rightarrow\infty.
\end{split}
\end{equation}
\textbf{Behavior at time $t>0$.} Let us write
\begin{equation}
\begin{split}
||u_{1,\lambda}(t)-u_{0,\lambda}(t)||_{H^s}&\geq ||u^{1,\lambda}(t)-u^{0,\lambda}(t)||_{H^s} -||u^{1,\lambda}(t)-u_{1,\lambda}(t)||_{H^s}-||u^{0,\lambda}(t)-u_{0,\lambda}(t)||_{H^s}.
\end{split}
\end{equation}
Applying the estimate (5.25) to the last two terms in (5.29) generates 
\begin{equation}
\begin{split}
||u_{1,\lambda}(t)-u_{0,\lambda}(t)||_{H^s}&\geq ||u^{1,\lambda}(t)-u^{0,\lambda}(t)||_{H^s}-c \lambda^\frac{-(1- \delta)(k-s)}{k-r},
\end{split}
\end{equation}
that is
\begin{equation}
\begin{split}
\liminf_{\lambda\rightarrow\infty}||u_{1,\lambda}(t)-u_{0,\lambda}(t)||_{H^s}&\geq \liminf_{\lambda\rightarrow\infty} ||u^{1,\lambda}(t)-u^{0,\lambda}(t)||_{H^s},
\end{split}
\end{equation}
So, it suffices to estimate $||u^{1,\lambda}(t)-u^{0,\lambda}(t)||_{H^s}$ below. Since
\begin{equation}
\begin{split}
u^{1,\lambda}(t)-u^{0,\lambda}(t)&=\lambda^{-\delta/2p-s}\phi\left(\frac{x}{\lambda^{\delta/p}}\right)\left[\cos(\lambda x- t)-\cos(\lambda x)\right]+u_{l,1,\lambda}(t)\\
&=\lambda^{-\delta/2p-s}\phi\left(\frac{x}{\lambda^{\delta/p}}\right) \sin(\lambda x-\frac{t}{2})\sin \frac{t}{2} +u_{l,1,\lambda}(t),\\
\end{split}
\end{equation}
we have
\begin{equation}
\begin{split}
||u^{1,\lambda}(t)-u^{0,\lambda}(t)||_{H^s}\geq &\lambda^{-\delta/2p-s}\left\|\phi\left(\frac{x}{\lambda^{\delta/p}}\right) \sin(\lambda x-\frac{t}{2})\right\|_{H^s}|\sin t/2| -c_s\lambda^{  (\delta-2) /2p}.
\end{split}
\end{equation}
Thus, using the relation
\begin{equation}
\begin{split}
\frac{1}{\sqrt{2}}||\phi||_{L^2}&=\lim_{\lambda\rightarrow\infty}  \lambda^{-\delta/2p-s}\left\|\phi\left(\frac{x}{\lambda^{\delta/p}}\right) \sin(\lambda x-\frac{t}{2})\right\|_{H^s},
\end{split}
\end{equation}
and (5.33), we obtain 
\begin{equation}
\begin{split}
\liminf_{\lambda\rightarrow\infty}||u^{1,\lambda}(t)-u_{0,\lambda}(t)||_{H^s}&\geq  \frac{1}{\sqrt{2}}||\phi||_{L^2} |\sin t|.
\end{split}
\end{equation}
Noticing that $|\sin t|=\sin t, 0\leq t\leq \pi$ gives 
\begin{equation}
\begin{split}
\liminf_{\lambda\rightarrow\infty}||u_{1,\lambda}(t)-u_{0,\lambda}(t)||_{H^s}&\geq  \frac{1}{\sqrt{2}}||\phi||_{L^2} \sin t,
\end{split}
\end{equation}
for $0\leq t\leq \min \{T,2\pi\}$.

In short, 
there exist two sequences of
solutions $z=(u_\lambda(t),v_\lambda(t))$ and $\tilde{z}(t)=(\tilde{u}_\lambda(t),\tilde{v}_\lambda(t))$ to the differential Eq.(1.1) in $\mathcal{C} ([0, T ]; H^s
(\mathbb{R}))\times \mathcal{C} ([0, T ]; H^s
(\mathbb{R}))$ such that
\begin{equation}
||u_\lambda(t)||_{H^s}+||\tilde{u}_\lambda(t)||_{H^s}+||v_\lambda(t)||_{H^s}+||\tilde{v}_\lambda(t)||_{H^s}\lesssim 1,
\end{equation}
\begin{equation}
\lim_{n\rightarrow\infty}||u_\lambda(0)-\tilde{u}_\lambda(0)||_{H^s}=\lim_{n\rightarrow\infty}||v_\lambda(0)-\tilde{v}_\lambda(0)||_{H^s}=0,
\end{equation}
and
\begin{equation*}
\begin{split}
\liminf_{n\rightarrow\infty}||u_\lambda(0)-\tilde{u}_\lambda(0)||_{H^s}> 0,0<t<\min\{2\pi , T\},\\
\liminf_{n\rightarrow\infty}||v_\lambda(0)-\tilde{v}_\lambda(0)||_{H^s}> 0,0<t<\min\{2\pi, T\}.
\end{split}
\end{equation*}
 This concludes the proof of Theorem 1.8.
\end{proof}

\subsection{Nonuniform dependence on a circle}
To prove Theorem 1.8 on a circle, similar to the nonperiodic case,  it suffices to show that there exist   two sequences of
solutions $z=(u_\lambda(t),v_\lambda(t))$ and $\tilde{z}(t)=(\tilde{u}_\lambda(t),\tilde{v}_\lambda(t))$ to the differential Eq.(1.1) in $\mathcal{C} ([0, T ]; H^s
(\mathbb{T}))\times \mathcal{C} ([0, T ]; H^s
(\mathbb{T}))$ such that
\begin{equation*}
||u_\lambda(t)||_{H^s}+||\tilde{u}_\lambda(t)||_{H^s}+||v_\lambda(t)||_{H^s}+||\tilde{v}_\lambda(t)||_{H^s}\lesssim 1,
\end{equation*}
\begin{equation*}
\lim_{n\rightarrow\infty}||u_\lambda(0)-\tilde{u}_\lambda(0)||_{H^s}=\lim_{n\rightarrow\infty}||v_\lambda(0)-\tilde{v}_\lambda(0)||_{H^s}=0,
\end{equation*}
and
\begin{equation*}
\begin{split}
\liminf_{n\rightarrow\infty}||u_\lambda(0)-\tilde{u}_\lambda(0)||_{H^s}\gtrsim \sin t,\\
\liminf_{n\rightarrow\infty}||v_\lambda(0)-\tilde{v}_\lambda(0)||_{H^s}\gtrsim \sin t.
\end{split}
\end{equation*}

To see this, we consider the approximate solutions in the form of 
$$u^{\omega,\lambda}=\omega \lambda^{-\frac{1}{p}}+\lambda^{-s}\cos(\lambda x-\omega^p t), v^{\omega,\lambda}=\omega \lambda^{-\frac{1}{q}}+\lambda^{-s}\cos(\lambda x-\omega^q t),$$
where $\lambda\in \mathbb{Z}^+$ and $\omega=\pm1$.
Eq.(1.1) is reduced to Eq.(1.8) with $ k=2,\beta=\gamma=0$.
Similar to the proof on a line and the proof on a circle for the CH and the 2CH in \cite{HKM,Thompson}, the results listed above for Eq.(1.1) on a circle can be
carried out.

\section{The peaked traveling wave solutions of Eq.(1.1)}
 In this section, we prove Theorem 1.9 through constructing some appropriate sequences of peakon solutions.
 First,  let us show that the peakon
formulas (1.14-1.15) define a weak solution to Eq.(1.1) on a line and on a circle. Furthermore, we derive
the formulation (1.16-1.17) for multi-peakon solutions.

\begin{proof} [Proof of Theorem \ref{result9}]
\textbf{The non-periodic peakon solution in the form of (1.14).}
Without loss of generality, we set  $x_{0}=0$.  First, let us rewrite the model (1.1) as
\begin{equation}
\left\{
\begin{array}{llll}
u_t+v^pu_x+I_1(u,v)=0,\\
v_t+u^qv_x+I_2(u,v)=0,
\end{array}
\right.
\end{equation}
where
\begin{equation*}
\left\{
\begin{array}{llll}
I_1(u,v)=(1-\partial_x^2)^{-1}[ a  v^{p-1}v_xu+(p-a)v^{p-1}v_xu_{xx}]+p(1-\partial_x^2)^{-1}\partial_x(v^{p-1}v_xu_x),\\
I_2(u,v)=(1-\partial_x^2)^{-1}[ b  u^{q-1}u_xv+(q-b)u^{q-1}u_xv_{xx}]+q(1-\partial_x^2)^{-1}\partial_x(u^{q-1}u_xv_x).
\end{array}
\right.
\end{equation*}
Noticing that
$$u_{t}=   \mathrm{sgn}( x-ct)cu, u_{x}=-\mathrm{sgn}( x-ct) u, v_{t}=  \mathrm{sgn}( x-ct) cu, v_{x}=-\mathrm{sgn}( x-ct) v, $$
then we have
\begin{align}
  u_{t}+v^pu_{x}=-(-cu+ v^pu )\mathrm{sgn}(x-ct),
  v_{t}+u^qv_{x}=-(-cv+ u^qv )\mathrm{sgn}(x-ct).
\end{align}
On the other hand, a simple computation reveals 
\begin{align*}
I_1(u,v)
&=\frac{ 1}{2}\int_\mathbb{R} e^{-|x-y|} [ a  v^{p-1}v_yu+(p-a)v^{p-1}v_yu_{yy}] (t,y)dy + \frac{p}{2}\partial_x\int_\mathbb{R} e^{-|x-y|} (v^{p-1}v_yu_y) (t,y)dy \\
&=- \frac{a \alpha\beta^p}{2} \int_\mathbb{R}  \mathrm{sgn}( y-ct)e^{-|x-y|}e^{-(p+1)| y-ct |}(t,y)dy\\
&\quad+\frac{ (p-a)\alpha\beta^p}{4}\int_\mathbb{R} \partial_y\left[\mathrm{sgn}( y-ct)e^{- | y-ct |}\right]^2e^{-|x-y|}e^{- (p-1)| y-ct |} dy\\
&\quad- \frac{p \alpha\beta^p}{2}  \int_\mathbb{R} \mathrm{sgn}^2( y-ct)\mathrm{sgn}( x-y)e^{-|x-y|} e^{-(p+1)| y-ct |}dy \\
&=\alpha\beta^p\int_\mathbb{R}  \big[-\frac{a }{2}\mathrm{sgn}( y-ct)-\frac{3 p-a }{4}\mathrm{sgn}^2( y-ct)\mathrm{sgn}( x-y)+\frac{(p-a)(p-1)}{4}\mathrm{sgn}^3( y-ct)\big]  \\
&\quad\quad\quad\quad \quad\quad e^{-|x-y| -(p+1)| y-ct |}(t,y)dy.
\end{align*}
If $x< ct$, we derive
\begin{align*}
I_1(u,v)&=\alpha\beta^p \big( \frac{(a-p)(p+2)}{4}\int_{-\infty}^{x} e^{(p+2)y-x-(p+1)ct} dy+\frac{p(a+4-p)}{4} \int_x^{ct} e^{ -py-x+(p+1)ct}  dy\\
&\quad+\frac{(p-a)(p+2)}{4} \int_{ct}^\infty e^{-(p+2)y+x+(p+1)ct}dy\big)\\
&= \alpha\beta^p  \left(-e^{(p+1)(x-ct)}+ e^{ x-ct } \right).
\end{align*}
If $x> ct$,   we deduce
\begin{align*}
I_1(u,v)&=\alpha\beta^p \big( \frac{(a-p)(p+2)}{4}\int_{-\infty}^{ct} e^{(p+2)y-x-(p+1)ct} dy+\frac{p(p-a-4)}{4} \int_{ct}^x e^{ -py-x+(p+1)ct}  dy\\
&\quad+\frac{(p-a)(p+2)}{4} \int_{x}^\infty e^{-(p+2)y+x+(p+1)ct}dy\big)\\
&= \alpha\beta^p  \left(e^{(p+1)(x-ct)}- e^{ x-ct } \right).
\end{align*}
Therefore, we obtain
\begin{align}
I_1(u,v)=  (-\beta^pu+ v^pu)\mathrm{sgn}(x-ct),
\end{align}
and 
\begin{align}
I_2(u,v)=  (-\alpha^qv+ u^qv)\mathrm{sgn}(x-ct).
\end{align}
Combining  Eqs. (6.2-6.4)  with the assumption $\alpha=c^{1/q},\beta=c^{1/p}$, one may immediately know that the
first equation of the system (6.1) holds in the sense of distribution.   Therefore, we complete the
proof of the theorem.

\textbf{The periodic peakon solution in the forms of (1.15). }
In order to show that Eq.(6.1) is equivalent to Eq. (1.1), let us start from 
the original system(1.1).
Let $f\in L_{loc}^1(X)$, where $X$ is an open set of $\mathbb{R}$. Assume that $f'$ exists and is continuous except at a single point $x_0\in X$ and
$f'\in L_{loc}^1(X)$; then the left- and right-handed limits $f(x_0^{\pm})$ exist and $(T_f)'=T_{f'}+[f(x_0^+)-f(x_0^-)]\delta_{x_0}$, where $T_f$ is the
distribution associated to the function $f$ and $\delta_{x_0}$ is the Dirac delta distribution centered at $x=x_0$.
Denote $K\doteq x-ct-2\pi\left[\frac{x-ct}{2\pi}\right]-\pi$.
Noticing that
\begin{align*}
u_{t}= cu,u_x=\alpha \sinh  K,u_{xx}=u -2\alpha \sinh(\pi)\delta_{ct},\\
v_{t}= cv,v_x=\beta \sinh  K,v_{xx}=v -2\beta \sinh(\pi)\delta_{ct},
\end{align*}
where $\delta_{ct}$ is the periodic Dirac delta distribution centered at $x=ct$ mod $2\pi$, we have 
$u-u_{xx}=2\alpha \sinh (\pi)\delta_{ct}$  and
\begin{align*}
  (1-\partial_x^2)u_t=-2c\alpha \sinh\pi \delta_{ct}'.
\end{align*}
Employing the hyperbolic identity $\cosh^2x=1+\sinh^2x$ yields
\begin{align*}
  \partial_x^2(v^p u_{ x})&= \alpha\beta^p  \partial_x^2(\sinh K \cosh^{p } K )=\alpha\beta^p \partial_x( \cosh^{p+1}K+p\cosh^{p-1}K\sinh^2K-2 \sinh(\pi)\cosh^p(\pi)\delta_{ct})\\
  &=\alpha\beta \partial_x((p+1)\cosh^{p+1}K-p\cosh^{p-1}K-2 \sinh(\pi)\cosh^p(\pi)\delta_{ct})\\
  &=\alpha\beta ^p[(p+1)^2\cosh^p K\sinh K-p(p-1)\cosh^{p-2}K\sinh K -2 \sinh(\pi)\cosh^p(\pi)\delta_{ct}'].
\end{align*}
Then, we find 
\begin{align*}
 (1- \partial_x^2)(v^p u_{ x})& =\alpha\beta^p  [(1-(p+1)^2)\cosh^p K\sinh K+p(p-1)\cosh^{p-2}K\sinh K +2 \sinh(\pi)\cosh^p(\pi)\delta_{ct}'].
\end{align*}
Similarly, we have 
\begin{align*}
p\partial_x(v^{p-1}v_xu_x)& =p\alpha\beta^p\partial_x(\cosh^{p-1}K\sinh^2K)=\alpha\beta^p\partial_x(\cosh^{p+1}K-\cosh^{p-1}K)\\
 &=p\alpha\beta^p ((p+1)\cosh^{p}K\sinh K-(p-1)\cosh^{p-2}K\sinh K),
\end{align*}
\begin{align*}
 &v_xu_{xx} = \frac{\alpha\beta}{2}\partial_x  \sinh^2  K   =   \frac{\alpha\beta}{2}  \partial_x  (\cosh^2K-1) =  \alpha\beta \sinh K \cosh K,
\end{align*}
and 
\begin{align*}
a  v^{p-1}v_xu+(p-a)v^{p-1}v_xu_{xx}=p \alpha\beta^p \sinh K \cosh^p K.
\end{align*}
Therefore, we have
\begin{align*}
m_{t}+v^{p}m_{x}&+av^{p-1}v_{x}m\\
&=(1-\partial_x^2)u_t+ (1- \partial_x^2)(v^p u_{ x})+  a  v^{p-1}v_xu+(p-a)v^{p-1}v_xu_{xx} +p \partial_x(v^{p-1}v_xu_x)\\
&= -2c\alpha \sinh\pi \delta_{ct}'+2\alpha\beta^p    \sinh(\pi)\cosh^p(\pi)\delta_{ct}'.
\end{align*}
In a similar way, for $n(t)$ we obtain
\begin{align*}
n_{t}+v^{q}n_{x}+bv^{q-1}v_{x}n= -2c\beta \sinh\pi \delta_{ct}'+2\alpha^q\beta    \sinh(\pi)\cosh^p(\pi)\delta_{ct}'.
\end{align*}
So, the periodic peaked function (1.15) is a solution to  the   equation (1.1) if and only if $\alpha=\frac{c^{1/q}}{\cosh(\pi)},\beta=\frac{c^{1/p}}{\cosh(\pi)}$.

\textbf{Multi-peakon solutions for Eq.(1.1)}. 

Let us use an adhoc definition for $u_x(x,t),v_x(x,t)$ given by
\begin{equation*}
u_x(t,x)=-\sum_{i=1}^M\mathrm{sgn}( x-g_i(t)) f_i(t)e^{-|x-g_i(t)|}, v_x(t,x)=-\sum_{j=1}^N\mathrm{sgn}( x-g_j(t)) f_j(t)e^{-|x-g_j(t)|},
\end{equation*}
which imply 
$u_x(x,t),v_x(x,t)$ are equal to $\langle u_x(x,t)\rangle=\frac{1}{2}(u_x(x^-,t)+u_x(x^+,t))$ and $\langle v_x(x,t)\rangle=\frac{1}{2}(v_x(x^-,t)+v_x(x^+,t))$, respectively.
To prove this rigorously, one should consider the nonlocal form of
Eq.(1.1). Similarly, we may compute
\begin{equation*}
u_{xx}(t,x)=u(x,t)-2\sum_{i=1}^M f_i(t)\delta_{g_i(t)}, v_{xx}(t,x)=v(x,t)-2\sum_{j=1}^N f_j(t)\delta_{g_j(t)},
\end{equation*}
i.e., $m=u-u_{xx}=2\sum_{i=1}^M f_i(t)\delta_{g_i(t)},n=v-v_{xx}=2\sum_{j=1}^N f_j(t)\delta_{g_j(t)}$, which lead to 
\begin{align*}
 m_{t}+v^{p}m_{x}+\frac{a}{p}(v^{p})_{x}m &= 2\sum_{i=1}^M[-f_i\dot{g}_i\partial_x(\delta_{g_i})  +\dot{f}_i\delta_{g_i}+v^{p}f_i \partial_x(\delta_{g_i}) +\frac{a}{p}(v^{p})_{x}f_i \delta_{g_i}] \\
&= 2\sum_{i=1}^M[(- f_i\dot{g}_i+v^{p}f_i )  \partial_x( \delta_{g_i}) +(\dot{f}_i +\frac{a}{p}(v^{p})_{x}f_i) \delta_{g_i}],
\end{align*}
where $\dot{f}=\partial_t f $.
Casting a test function  $\varphi\in C_0^\infty(\mathbb{R})$ and  $ (f,\delta_{g_i})=f(g_i)$ on the both sides of the equation yields
\begin{align*}
& \left(m_{t}+v^{p}m_{x}+\frac{a}{p}(v^{p})_{x}m,\varphi \right)\\
&\quad\quad\quad=2\sum_{i=1}^M  - f_i\dot{g}_i\left(  \partial_x\delta_{g_i}   ,\varphi\right)+2\sum_{i=1}^M f_i\left( v^{p}   \partial_x \delta_{g_i}  ,\varphi\right)+2\sum_{i=1}^M\left( (\dot{f}_i +\frac{a}{p}(v^{p})_{x}f_i)\varphi, \delta_{g_i} \right)\\
&\quad\quad\quad=2\sum_{i=1}^M   f_i\dot{g}_i\left(   \delta_{g_i}  ,\partial_x\varphi \right)-2\sum_{i=1}^M f_i\left(  (v^{p})_x\varphi+ (v^{p})\partial_x\varphi   , \delta_{g_i}\right)+2\sum_{i=1}^M\left( (\dot{f}_i +\frac{a}{p}(v^{p})_{x}f_i)\varphi, \delta_{g_i} \right)\\
&\quad\quad\quad=2\sum_{i=1}^M\left[f_i(\dot{g}_i-v^p(g_i)) \varphi _x(q_i)+(\dot{f}_i+\frac{a-p}{p} (v^p)_x(g_i)f_i)\varphi(g_i)\right].
\end{align*}
Similarly, we have
\begin{align*}
& \left(n_{t}+v^{q}n_{x}+\frac{b}{q}(v^{q})_{x}n,\varphi \right)=2\sum_{j=1}^N\left[h_j(\dot{k}_j-u^q(k_j)) \varphi _x(q_j)+(\dot{h}_j+\frac{b-q}{q} (u^q)_x(k_j)h_j)\varphi(k_j)\right].
\end{align*}
Therefore, the multi-peakon is a solution to Eq.(1.1) if and only if $g_i(t),h_j(t)$ and amplitudes $f_i(t),k_j(t)$
satisfy the ODE system given by (1.18).
\end{proof}

\section{H\"{o}lder continuous in $H^r$-topology}

Theorem 1.7 tells us that the CCCH initial value problem is well-posed for $(u_0,v_0)\in H^s\times H^s,s>5/2$ with $(u,v)\in \mathcal{C}([0,T];H^s)\times \mathcal{C}([0,T];H^s)$.
Moreover, in Theorem 1.8 we have shown that the solution map $(u_0,v_0)\in H^s\times H^s \mapsto (u,v)\in \mathcal{C}([0,T];H^s)\times \mathcal{C}([0,T];H^s)$ is continuous but not uniformly continuous in $ H^s \times  H^s $.

Let us now make a  further investigation about the continuity properties for the solution map in H\"{o}lder spaces $H^r, r<s$ with initial data still in $H^s,s>5/2$.
More precisely,
we consider two solutions  of equation (1.1), $ z=(u_1,v_1)$ and $ w=(u_2,v_2)$,  which emanate from the initial data $ z_0=(u_{0,1},v_{0,1})$ and $w_0=(u_{0,2},v_{0,2} )$, respectively.
We want to show that if the initial data $z_0,w_0 $ are assigned 
in a ball with radius $\rho$ in $H^s$, i.e.,
\begin{equation}
\begin{split}
||z_0||_{H^s}\leq \rho,||w_0||_{H^s}\leq \rho,s>5/2,
\end{split}
\end{equation}
then we have
$$||z(t)-w(t)||_{H^r}\lesssim ||z_0-w_0||_{H^r}^\alpha,r<s, $$
where the H\"{o}lder exponent $\alpha$ is to be determined. 

The proof of Theorem 1.10 is inspired from the work on b-family equation \cite{CLZ}, the Novikov equation \cite{HH4}, the FORQ equation \cite{HM33}, and the two-component Camassa-Holm equation \cite{Thompson}.
The follow three lemmas are  needed to prove Theorem 1.10.
\begin{lm}(See \cite{Taylor2}.)
If   $f\in H^{s -1}$ and $g\in H^\theta$, then
 $$||[D^\theta\partial_x,f]g||_{L^2 }\leq c_{\theta,s} ||f||_{H^{s-1 } }|| g||_{H^\theta } ,\theta+1\geq0, s -1>3/2,\theta+1\leq s -1.$$
\end{lm}
\begin{lm}(See\cite{HM33,K}.)
If $\theta >0$, then $H^\theta \cap L^\infty$
is an algebra. Moreover, we have

(i) $||fg||_{H^\theta  }\leq c_\theta ||f||_{H^\theta }||g||_{H^\theta  } , \text{ for } \theta >1/2.$

(ii) $||fg||_{H^{\theta  } }\leq c ||f||_{H^{\theta +1} }||g||_{H^{\theta  }}, \text{ for }  \theta >-1/2.$

(iii) $||fg||_{H^{\theta } }\leq c ||f||_{H^{s-2 } }||g||_{H^{\theta }}, \text{ for } -1\leq\theta\leq 0, s-1>3/2, \theta+s\geq 2.$
\end{lm}

\begin{lm}(See\cite{HM33,K}.)
Suppose $\sigma_1<\sigma<\sigma_2$ and $f\in H^\sigma_1$. Then, we have
\begin{equation*}
\begin{split}
||f||_{H^{\sigma}} \leq ||f||_{H^{\sigma_1}}^\frac{\sigma_2-\sigma}{\sigma_2-\sigma_1} ||f||_{H^{\sigma_2}}^\frac{\sigma-\sigma_1}{\sigma_2-\sigma_1},
\end{split}
\end{equation*}
\end{lm}

\begin{proof}[Proof of Theorem \ref{result10}]
Differentiating the CCCH system with respect to  $x$,  simplifying the resulting
equation, and letting $t=-t$, $f=u_x$ and $g=v_x$, we get 
the following equation
\begin{equation}
\left\{
\begin{array}{llll}
&\partial_tu= v^p\partial_xu +I_{11}(u,v,f,g)+I_{12}(u,v,f,g)+\partial_xI_{13}(u,v,f,g), \\
&\partial_tv= u^q\partial_xv +I_{21}(u,v,f,g)+I_{22}(u,v,f,g)+\partial_xI_{23}(u,v,f,g),\\
&\partial_tf= v^p\partial_xf +\partial_xI_{11}(u,v,f,g)+\partial_xI_{12}(u,v,f,g)+I_{13}(u,v,f,g),\\
&\partial_tg= u^q\partial_xg +\partial_xI_{21}(u,v,f,g)+\partial_xI_{22}(u,v,f,g)+ I_{23}(u,v,f,g),\\
&u_0(x)=u(x,0),v_0(x)=v(x,0),f_0(x)=\partial_xu_0(x)\doteq f_0(x),g_0(x)=\partial_xv_0(x)\doteq g_0(x).
\end{array}
\right.
\end{equation}
where
\begin{equation*}
\left\{
\begin{array}{llll}
I_{11}(u,v,f,g)=\frac{a}{p} (1-\partial_x^2)^{-1}  (\partial_x(v^{p }) u),\quad \quad \quad I_{21}(u,v,f,g)&=\frac{b}{q} (1-\partial_x^2)^{-1}  (\partial_x(u^{q}) v),\\
I_{12}(u,v,f,g)= \frac{p-a }{p}(1-\partial_x^2)^{-1}(\partial_x(v^{p }) \partial_xf ) ,\quad  I_{22}(u,v,f,g)&= \frac{q-b }{q}(1-\partial_x^2)^{-1}(\partial_x(u^{q}) \partial_xg ),\\
I_{13}(u,v,f,g)= (1-\partial_x^2)^{-1} (\partial_x(v^{p})f),\quad\quad \quad  I_{23}(u,v,f,g)&=(1-\partial_x^2)^{-1} (\partial_x(u^{q})g).
\end{array}
\right.
\end{equation*}

It should be pointed out that in the periodic case the integration is over $\mathbb{T}$.
Since all estimates are the
same on both the line and the circle, in what follows we shall keep using the notation
of the line.

Let us consider two solutions  $ \phi=(u_1,v_1,f_1,g_1)$ and $ \varphi=(u_2,v_2,f_2,g_2)$  to equation (7.2),
which correspond to the initial data $ \phi_0=(u_{0,1},v_{0,1},\partial_xu_{0,1},\partial_xv_{0,1})$ and $\varphi_0=(u_{0,2},v_{0,2},\partial_xu_{0,2},v_{0,2} )$, respectively.
If the initial data $\phi_0,\varphi_0 $ are located in a ball with radius $\rho$ in $H^{s-1}$, i.e.,
\begin{equation}
\begin{split}
||\phi_0||_{H^{s-1}}\leq \rho,||\varphi_0||_{H^{s-1}}\leq \rho,s>5/2,
\end{split}
\end{equation}
then, from the estimate (1.13) in Theorem 1.7 we have
\begin{equation}
\begin{split}
||\phi(t)||_{H^{s-1}}\leq \rho,||\varphi(t)||_{H^{s-1}}\leq \rho,s>5/2.
\end{split}
\end{equation}

Let $U=u_1-u_2, V=v_1-v_2,F=f_1-f_2, G=g_1-g_2$. Then $U,V,F,G$ satisfy the following system
\begin{align*}
&\partial_tU= v_1^p\partial_xU+V\partial_xu_2\sum_{i=0}^{p-1}v_1^{p-1-i}v_2^i+I_{11}(u_1,v_1,f_1,g_1)-I_{11}(u_2,v_2,f_2,g_2)  \\
&\quad\quad\quad+I_{12}(u_1,v_1,f_1,g_1)-I_{12}(u_2,v_2,f_2,g_2)+\partial_xI_{13}(u_1,v_1,f_1,g_1)-\partial_xI_{13}(u_2,v_2,f_2,g_2),\\
&\partial_tV= u_1^q\partial_xV+U\partial_xv_2\sum_{i=0}^{q-1}u_1^{q-1-i}u_2^i+I_{21}(u_1,v_1,f_1,g_1)-I_{21}(u_2,v_2,f_2,g_2)  \\
&\quad\quad\quad+I_{22}(u_1,v_1,f_1,g_1)-I_{22}(u_2,v_2,f_2,g_2)+\partial_xI_{23}(u_1,v_1,f_1,g_1)-\partial_xI_{23}(u_2,v_2,f_2,g_2),\\
&\partial_tF= v_1^p\partial_xF+V\partial_xf_2\sum_{i=0}^{p-1}v_1^{p-1-i}v_2^i+\partial_xI_{11}(u_1,v_1,f_1,g_1)-\partial_xI_{11}(u_2,v_2,f_2,g_2)  \\
&\quad\quad\quad+\partial_xI_{12}(u_1,v_1,f_1,g_1)-\partial_xI_{12}(u_2,v_2,f_2,g_2)+ I_{13}(u_1,v_1,f_1,g_1)- I_{13}(u_2,v_2,f_2,g_2),\tag{F}\\
&\partial_tG= u_1^q\partial_xG+U\partial_xg_2\sum_{i=0}^{q-1}u_1^{q-1-i}u_2^i+\partial_xI_{21}(u_1,v_1,f_1,g_1)-\partial_xI_{21}(u_2,v_2,f_2,g_2)  \\
&\quad\quad\quad+\partial_xI_{22}(u_1,v_1,f_1,g_1)-\partial_xI_{22}(u_2,v_2,f_2,g_2)+I_{23}(u_1,v_1,f_1,g_1)- I_{23}(u_2,v_2,f_2,g_2),\\
&U_0(x)=u_{0,1}(x )-u_{0,2}(x),V_0(x)=v_{0,1}(x )-v_{0,2}(x),\\
&F_0(x)=\partial_x(u_{0,1}(x )-u_{0,2}(x)),G_0(x)=\partial_x(v_{0,1}(x )-v_{0,2}(x)).
\end{align*}

\textbf{Lipschitz continuity in $A_1 $.}
We shall show that the solution map for (1.1) is Lipschitz continuous for $(s,r)\in A_1$, that is, the H\"{o}lder exponent { $\alpha=1$}
 in the domain $A_1$.
Applying the operator $D^r$ to both sides of the equation (F), multiplying by $ D^r F$, and integrating, we obtain
\begin{equation}
\begin{split}
\frac{1}{2}\frac{d}{dt}||F||_{H^r}^2  &=   \int_\mathbb{R} D^r \left(v_1^p \partial_xF \right) D^r Fdx + \int_\mathbb{R} D^r \left( V \partial_xf_2\sum_{i=0}^{p-1}v_1^{p-1-i}v_2^i\right) D^r Fdx\\
&\quad+\int_\mathbb{R} D^r (\partial_xI_{11}(u_1,v_1,f_1,g_1 )-\partial_xI_{11}(u_2,v_2f_2,v_2 ))    D^r Fdx\\
&\quad+ \frac{a}{p}\int_\mathbb{R} D^{r} (\partial_xI_{12}(u_1,v_1,f_1,g_1  )-\partial_xI_{12}(u_2,v_2f_2,v_2 ) )    D^rFdx\\
&\quad+\frac{a}{p}\int_\mathbb{R} D^{r} ( I_{13}(u_1,v_1,f_1,g_1  )- I_{13}(u_2,v_2f_2,v_2 ))     D^r Fdx\\
&\doteq B_1+B_2+ B_3+ B_4+ B_5.
\end{split}
\end{equation}
We need to estimate the right-hand side of (7.5). Apparently, $ D^r$ is commutative.

\textbf{Estimation of $B_1 $.}
A direct calculation sends the first  term on the right-hand sides of (7.5) to
\begin{equation}
\begin{split}
|B_1|&=\left|\int_\mathbb{R} D^r [\partial_x\left(v_1^p F \right)-F\partial_x v_1^p  ] D^r Fdx\right|\\
& \lesssim \left|\int_\mathbb{R} [ D^r  \partial_x,v_1^p]F       D^r Fdx\right|+\left|\int_\mathbb{R}(v_1^p D^r  \partial_xF)    D^r Fdx\right|+ \left|\int_\mathbb{R} D^r  F\partial_x v_1^p    D^r Fdx\right| .
\end{split}
\end{equation}
The first integral can be estimated through the Calderon-Coifman-Meyer commutator described in Lemma 7.1 for $r+1\geq0, s-1 >3/2,r+1\leq s-1 $.
Employing the algebraic property $ || v_1^p||_{H^{s-1 }}\lesssim || v_1||^p_{H^{s-1 }}\lesssim \rho^p$ and the Sobolev inequality $ || (v_1^p)_x||_{L^\infty}\lesssim || v_1^p||_{H^{s -1}}\lesssim \rho^p$ yields
\begin{equation*}
\begin{split}
 \left|\int_\mathbb{R} [ D^r  \partial_x,v_1^p]F       D^r Fdx\right|
&\lesssim  ||[ D^r  \partial_x,v_1^p]F||_{L^2}  || F||_{H^r}
\lesssim  || v_1^p||_{H^{s-1 }}|| F||_{H^r}^2 \lesssim  \rho^p || F||_{H^r}^2.
\end{split}
\end{equation*}
The second integral of (7.6) can be handled through integration by parts and the Sobolev's lemma
\begin{equation*}
\begin{split}
\left|\int_\mathbb{R}(v_1^p D^r  \partial_xF)    D^r Fdx\right|
 \lesssim   \left|\int_\mathbb{R}(v_1^p)_x(  D^r   F   )^2dx\right|
 \lesssim  ||(v_1^p)_x||_{L^\infty} || F||_{H^r}^2 \lesssim  \rho^p || F||_{H^r}^2.
\end{split}
\end{equation*}
The third integral of (7.6) would be calculated as follows
\begin{equation*}
\begin{split}
&\left|\int_\mathbb{R} D^r  F\partial_x v_1^p   D^r Fdx\right|\lesssim|| F\partial_x v_1^p||_{H^{r }}|| F||_{H^r}\\
&\quad\lesssim \left\{
\begin{array}{llll}
||  v_1^{p-1}g_1||_{H^{r  }}|| F||^2_{H^r}\lesssim  ||  v_1^{p-1}g_1|| _{H^{s-1  }}|| F||^2_{H^r}\lesssim\rho^p || F||^2_{H^r},\text{ for }  1/2<r\leq s-1 ;\\
||v_1^{p-1}g_1||_{H^{r+1 }}|| F||^2_{H^r} \lesssim||v_1^{p-1}g_1|| _{H^{s-1  }}|| F||^2_{H^r}\lesssim\rho^p || F||^2_{H^r},\text{ for }  -1/2<r\leq 1/2, r+2\leq s;\\
||v_1^{p-1}g_1||_{H^{s-1}}|| F||^2_{H^r}\lesssim     \rho^p || F||^2_{H^r}, \text{ for } -1\leq r \leq -1/2, s-1>3/2, r +s\geq 2.
\end{array}
\right.
\end{split}
\end{equation*}
To sum up, we arrive at 
\begin{equation*}
\begin{split}
|B_1|\lesssim\rho^p  || F||_{H^r}^2 ,\text{ for }(r,s)\in\{ 1/2<r<s-1\}\cup\{-1/2< r\leq 1/2, r+2\leq s\}\cup\{ -1\leq r \leq 0,   r +s\geq 2\}.
\end{split}
\end{equation*}

\textbf{Estimation of $B_2 $.}
A long computation yields the following estimate of $B_2$:
\begin{equation*}
\begin{split}
|B_2|&=\left|\int_\mathbb{R} D^r \left( V \partial_xf_2\sum_{i=0}^{p-1}v_1^{p-1-i}v_2^i \right)D^r Fdx\right|\lesssim \left\|V \partial_xf_2\sum_{i=0}^{p-1}v_1^{p-1-i}v_2^i\right\|_{H^{r }}|| F||_{H^r}\\
& \lesssim \left\{
\begin{array}{llll}
\left\|\partial_xf_2\sum_{i=0}^{p-1}v_1^{p-1-i}v_2^i \right\|_{H^{r  }}|| V|| _{H^r}|| F|| _{H^r}\lesssim ||  f_2|| _{H^{r+1  }}\left\|\sum_{i=0}^{p-1}v_1^{p-1-i}v_2^i \right\|_{H^{r  }}|| F|| _{H^r}|| V|| _{H^r}\\
\quad \quad \lesssim||   f_2|| _{H^{s-1 }}\left\|\sum_{i=0}^{p-1}v_1^{p-1-i}v_2^i\right\|_{H^{s  }}|| F|| _{H^r}|| V|| _{H^r}\lesssim\rho^p || F|| _{H^r}|| V|| _{H^r},\text{ for }  1/2<r\leq s-2;\\
\left\|\partial_xf_2\sum_{i=0}^{p-1}v_1^{p-1-i}v_2^i \right\|_{H^{r+1  }}|| V|| _{H^r}|| F|| _{H^r}\lesssim ||   f_2|| _{H^{r+2}}\left\|\sum_{i=0}^{p-1}v_1^{p-1-i}v_2^i \right\|_{H^{r+1  }}|| F|| _{H^r}|| V|| _{H^r}\\
\quad \quad\lesssim||   f_2|| _{H^{s-1 }}\left\|\sum_{i=0}^{p-1}v_1^{p-1-i}v_2^i\right\|_{H^{s  }}|| F|| _{H^r}|| V|| _{H^r}\lesssim\rho^p || F|| _{H^r}|| V|| _{H^r}, \text{ for }  -1/2< r\leq 1/2, r+3\leq s;\\
\left\|\partial_xf_2\sum_{i=0}^{p-1}v_1^{p-1-i}v_2^i \right\|_{H^{s-2 }}|| V|| _{H^r}|| F|| _{H^r}\lesssim ||  f_2|| _{H^{s-1}}\left\|\sum_{i=0}^{p-1}v_1^{p-1-i}v_2^i \right\|_{H^{s-1  }}|| F|| _{H^r}|| V|| _{H^r}\\
\quad \quad \lesssim\rho^p || F|| _{H^r}|| V|| _{H^r}, \text{ for } -1\leq r \leq 0, s-1>3/2, r +s\geq 2.
\end{array}
\right. \\
&\lesssim\rho^p  || F||_{H^r}^2 ,\text{ for }(r,s)\in\{ 1/2<r\leq s-2\}\cup\{-1/2< r\leq 1/2, r+3\leq s\}\cup\{ -1\leq r \leq -1/2,   r +s\geq 2\} .
\end{split}
\end{equation*}

\textbf{Estimation of $B_4 $.}
We shall show the details for the estimate of the most delicate nonlocal term $B_4$ while other nonlocal terms can be done likewise. 
A direct calculation generates
\begin{equation*}
\begin{split}
|B_4| &\lesssim
 || \partial_x (v^{p }_1) \partial_xf_1-\partial_x(v^{p }_2)  \partial_x  f_2 ||_{H^{r-1  }} || F||_{H^r}\\
 &\lesssim (\left\|\partial_x(v_1^p-v_2^p)\partial_x f_2\right\|_{H^{r -1 }}+\left\|\partial_xv_1^p\partial_x F \right\|_{H^{r -1 }}) || F|| _{H^r}.
\end{split}
\end{equation*}
Applying Lemma 7.2 yields
\begin{equation*}
\begin{split}
\left\|\partial_x(v_1^p-v_2^p)\partial_x f_2\right\|_{H^{r -1 }}&\lesssim \left\|Vg_1\partial_xf_2\sum_{i=0}^{p-2}v_1^{p-2-i}v_2^i \right\|_{H^{r-1  }}+\left\|v_2^{p-1}G\partial_x f_2 \right\|_{H^{r-1  }}\\
& \lesssim \left\{
\begin{array}{llll}
\left\| g_1\partial_xf_2\sum_{i=0}^{p-2}v_1^{p-2-i}v_2^i \right\|_{H^{r-1  }}||V||_{H^r}+\left\| v_2^{p-1} \partial_x f_2 \right\|_{H^{r-1  }}||G||_{H^r}\\
\quad \quad\lesssim
 \rho^p  (||V||_{H^r}+||G||_{H^r}), \text{ for }  1/2< r \leq s-1;\\
\left\|Vg_1\partial_xf_2\sum_{i=0}^{p-2}v_1^{p-2-i}v_2^i \right\|_{H^{r   }}+\left\|v_2^{p-1}G\partial_x f_2 \right\|_{H^{r   }}\\
\quad \quad\lesssim
\left\|V \right\|_{H^{r   }}||\partial_xf_2g_1\sum_{i=0}^{p-2}v_1^{p-2-i}v_2^i ||_{H^{s-2}}+\left\|G \right\|_{H^{r   }}||  v_2^{p-1}\partial_x f_2||_{H^{s-2}}\\
\quad \quad\lesssim
 \rho^p  (||V||_{H^r}+||G||_{H^r}), \text{ for } -1\leq r \leq 1/2, s-1>3/2, r +s\geq 2.
\end{array}
\right.
\end{split}
\end{equation*}
{ Following a similar proof of Lemma 3  in \cite{HH4} }, we can get
  $||fg||_{H^{r-2} }\leq c ||f||_{H^{s-1 } }||g||_{H^{r-2}}  \text{ for } 0\leq r\leq 2, s >3/2, r+s\geq 3 $. Using this fact and Lemma 7.2, we get
\begin{equation*}
\begin{split}
&\left\|\partial_xv_1^p\partial_x F \right\|_{H^{r -1 }}\\
& \quad\lesssim \left\{
\begin{array}{llll}
\left\| v_1^{p-1}g\right\|_{H^{r   }}\left\|\partial_x F \right\|_{H^{r -1 }} \lesssim
\left\| v_1^{p-1}g\right\|_{H^{s-1   }}\left\| F \right\|_{H^{r   }}, \text{ for }  1/2< r \leq s-1 ;\\
\left\|v_1^{p-1}g\right\|_{H^{s -1 }}\left\|\partial_x F \right\|_{H^{r -1 }}\lesssim
\left\| v_1^{p-1}g\right\|_{H^{s-1   }}\left\| F \right\|_{H^{r   }}, \text{ for } -1\leq r \leq 1/2, s-1>3/2, r +s\geq 2.
\end{array}
\right.
\end{split}
\end{equation*}
Therefore
\begin{equation*}
\begin{split}
|B_4|\lesssim\rho^p  || F||_{H^r}^2 ,\text{ for }(r,s)\in\{ 1/2<r\leq s-1\}\cup\{-1\leq r \leq 1/2, s-1>3/2, r +s\geq 2\} .
\end{split}
\end{equation*}

After estimating other terms in a similar manner, 
we obtain,
\begin{equation*}
\begin{split}
 &\frac{d}{dt}(|| U||_{H^r} +|| V||_{H^r}+||F||_{H^r}+|| G||_{H^r})  \leq \mathcal{C} (|| U||_{H^r} +|| V||_{H^r}+||F||_{H^r}+|| G||_{H^r}),
\end{split}
\end{equation*}
 for  $(r,s)\in\{-1\leq r\leq -1/2,r+s\geq 2\}\cup \{-1/2< r\leq 1/2,r+3\leq s\}\cup\{1/2< r,r+2\leq s\}$ and $C=C(s,r,p,q,\rho)$.
Solving the above inequality yields
\begin{equation*}
\begin{split}
|| U(t)||_{H^r} +|| V(t)||_{H^r}+||F(t)||_{H^r}+|| G(t)||_{H^r}  \leq e^{CT}\left(|| U(0)||_{H^r} +|| V(0)||_{H^r}+||F(0)||_{H^r}+|| G(0)||_{H^r} \right),
\end{split}
\end{equation*}
Noticing $F=f_1-f_2=\partial_x(u_1-u_2)=\partial_xU$ and  $G=g_1-g_2=\partial_x(v_1-v_2)=\partial_xV$, by the above  inequality we have
\begin{equation*}
\begin{split}
 &   ||U(t)||_{H^{r+1}}+|| V(t)||_{H^{r+1}}    \lesssim e^{CT}( ||U(0)||_{H^{r+1}}+|| V(0)||_{H^{r+1}}).
\end{split}
\end{equation*}
Lowering down the Sobolev index from $r+1$ to $r$ and adjusting the range accordingly,
we have
\begin{equation*}
\begin{split}
 &   ||U(t)||_{H^{r+1}}+|| V(t)||_{H^{r+1}}    \lesssim e^{CT}( ||U(0)||_{H^{r+1}}+|| V(0)||_{H^{r+1}}),
\end{split}
\end{equation*}
where  $(r,s)\in A_1\doteq\{0\leq r\leq 1/2,r+s\geq 3\}\cup \{1/2< r\leq 3/2,r+2\leq s\}\cup\{3/2< r,r+1\leq s\}$.
Recalling the definition $ z=(u_1,v_1)$ and $ w=(u_2,v_2)$ reveals
\begin{equation*}
\begin{split}
 &   ||z(t)-w(t)||_{H^{r }}  \lesssim e^{C_{r,s,p,q,\rho}T}( ||z(0)-w(0)||_{H^{r }},
\end{split}
\end{equation*}
where  $(r,s)\in A_1\doteq\{0\leq r\leq 1/2,r+s\geq 3\}\cup \{1/2< r\leq 3/2,r+2\leq s\}\cup\{3/2< r,r+1\leq s\}$.
This completes the proof of Lipschitz continuity in the region $A_1$.

\textbf{H\"{o}lder continuity in $A_2$.} By the Lipschitz continuity in $A_1$ and the condition $r\leq 3-s$, we obtain
\begin{equation*}
\begin{split}
 ||z(t)-w(t)||_{H^r}   \leq||z(t)-w(t)||_{H^{3-s}}  \leq e^{C_{r,s,p,q,\rho}T} ||z(0)-w(0)||_{H^{3-s}}.
\end{split}
\end{equation*}
Interpolating between the $H^r$ and the $H^s$ norms (Lemma 7.3 with $\sigma_1=r$, $\sigma=s-3$ and $\sigma_2=s$) produces
\begin{equation*}
\begin{split}
||z(0)-w(0)||_{H^{3-s}} \leq ||z(0)-w(0)||_{H^{r}}^\frac{2s-3}{s-r} ||z(0)-w(0)||_{H^{s}}^\frac{3-s-r}{s-r}\leq c_{r,s,\rho}||z(0)-w(0)||_{H^{r}}^\frac{2s-3}{s-r}.
\end{split}
\end{equation*}

\textbf{H\"{o}lder continuity in $A_3$.} For the case $s-2\leq r<s$, interpolating between the $H^{s-2}$ and the $H^s$ norms (Lemma 7.3 with $\sigma_1=s-2$, $\sigma=r$ and $\sigma_2=s$) generates
\begin{equation*}
\begin{split}
||z(t)-w(t)||_{H^{r}} \leq ||z(t)-w(t)||_{H^{s-2}}^\frac{s-r}{2} ||z(t)-w(t)||_{H^{s}}^\frac{r-s+2}{2},
\end{split}
\end{equation*}
and by the well-posedness size estimate (7.1),  we find
\begin{equation*}
\begin{split}
||z(t)-w(t)||_{H^{s}}\lesssim ||z_0||_{H^s}+||w_0||_{H^s} \lesssim \rho,
\end{split}
\end{equation*}
therefore, we have
\begin{equation*}
\begin{split}
||z(t)-w(t)||_{H^{r}} \leq C_{r,s,p,q,\rho}||z(t)-w(t)||_{H^{s-2}}^\frac{s-r}{2}.
\end{split}
\end{equation*}
By the Lipschitz continuity in $A_1$ and the condition $s-2\leq r<s$, we obtain
\begin{equation*}
\begin{split}
||z(t)-w(t)||_{H^{r}} \leq C_{r,s,\rho}||z(0)-w(0)||_{H^{s-2}}^\frac{s-r}{2}  \leq C_{r,s,\rho}||z(0)-w(0)||_{H^{r}}^\frac{s-r}{2},
\end{split}
\end{equation*}
which is the desired H\"{o}lder continuity in $A_3$.

\textbf{H\"{o}lder continuity in $A_4$.} For the case $s-1<r<s$, letting $\sigma_1=s-1$, $\sigma=r$ and $\sigma_2=s$ in Lemma 7.3 leads to 
\begin{equation*}
\begin{split}
||z(t)-w(t)||_{H^{r}} \leq ||z(t)-w(t)||_{H^{s-1}}^ {s-r} ||z(t)-w(t)||_{H^{s}}^ {r-s+1}.
\end{split}
\end{equation*}
By the Lipschitz continuity in $A_1$ for $r=s-1$ and the size estimate, we arrive at
\begin{equation*}
\begin{split}
||z(t)-w(t)||_{H^{r}} \leq C_{r,s,\rho}||z(0)-w(0)||_{H^{s-1}}^ {s-r} \leq C_{r,s,\rho}||z(0)-w(0)||_{H^{r}}^ {s-r},
\end{split}
\end{equation*}
which completes the proof of Theorem 1.10.
\end{proof}

\section*{Acknowledgments} { The first author would like to thank Professors
Zhijun Qiao and Shuxia Li for their kind hospitality and encouragement during his visit
in the University of Texas-Rio Grande Valley. This work is partly supported   by National Science Fund of China (Grant No.11371384) and National Science Fund for Young Scholars of China (Grant No. 11301573), University Young Core  Teacher Foundation of Chongqing, Technology Research Foundation of Chongqing Educational Committee (Grant No. KJ1400503), Natural Science Foundation of Chongqing  (Grant No. cstc2014jcyjA00008),   the Talent Project of Chongqing Normal University(Grant No. 14CSBJ05).}

\end{document}